\documentclass[a4paper,12pt]{article}
\pdfoutput=1 

\usepackage[utf8]{inputenc}
\usepackage[english]{babel}
\usepackage{amssymb, amsmath}
\usepackage{amsthm}
\usepackage{eucal}
\usepackage{mathalpha}
\usepackage{amsfonts}
\usepackage{mathrsfs}
\usepackage{setspace}
\usepackage{graphicx}
\usepackage{subcaption}
\usepackage{slashed}
\usepackage[sorting=none, backend=biber]{biblatex}
\usepackage{array,booktabs}
\usepackage{lmodern}
\usepackage{tikz}
\usepackage{microtype}
\usepackage{hyperref}
\usepackage{authblk}
\usepackage{yfonts}
\usepackage{csquotes}

\usepackage{tikz-cd}

\theoremstyle{definition}
\newtheorem{definition}{Definition}[section]
\newtheorem{remark}{Remark}[section]

\newtheorem{theorem}{Theorem}
\newtheorem{prop}{Proposition}[section]
\newtheorem{lemma}{Lemma}[section]
\newtheorem{corollary}{Corollary}[section]

\makeatletter
\renewenvironment{proof}[1][\proofname]{%
   \par\pushQED{\qed}\normalfont%
   \topsep6\p@\@plus6\p@\relax
   \trivlist\item[\hskip\labelsep\bfseries#1\@addpunct{.}]%
   \ignorespaces
}{%
   \popQED\endtrivlist\@endpefalse
}

\newcommand{\p}{\partial}

\newcommand{\beq}{\begin{equation}}
\newcommand{\eeq}{\end{equation}}
\newcommand{\eeeem}{\end{multline}}
\newcommand{\bem}{\begin{multline}}
\newcommand{\bqa} {\begin{eqnarray}}
\newcommand{\eqa} {\end{eqnarray}}
\newcommand{\eps}{\varepsilon}

\newcommand{\m}{m}
\newcommand{\scrI}{{\mathscr{I}}}

\newcommand{\bmul}{\begin{multline}}
\newcommand{\emul}{\end{multline}}

\DeclareMathOperator{\Sym}{Sym}

\DeclareMathOperator{\Id}{Id}

\DeclareMathOperator{\diam}{diam}
\DeclareMathOperator{\Ad}{Ad}

\DeclareMathOperator{\supp}{supp}
\DeclareMathOperator{\dist}{dist}

\DeclareMathOperator{\End}{End}
\DeclareMathOperator{\Aut}{Aut}

\def \ra {\rightarrow}

\newcommand{\No}{{\mathcal N}}

\newcommand{\CA}{{\mathcal A}}
\newcommand{\CB}{{\mathcal B}}

\newcommand{\CE}{{\mathcal E}}

\newcommand{\CH}{{\mathcal H}}
\newcommand{\CI}{{\mathcal I}}

\newcommand{\CM}{{\mathcal M}}
\newcommand{\CN}{{\mathcal N}}
\newcommand{\CL}{{\mathcal L}}
\newcommand{\CO}{{\mathcal O}}

\newcommand{\CU}{{\mathcal U}}
\newcommand{\CV}{{\mathcal V}}

\newcommand{\NN}{{\mathbb N}}
\newcommand{\ZZ}{{\mathbb Z}}
\newcommand{\RR}{{\mathbb R}}

\newcommand{\SA}{{\mathscr A}}
\newcommand{\SAl}{{\mathscr A}_{\ell}}
\newcommand{\SAal}{{\mathscr A}_{a \ell}}

\newcommand{\Orf}{{\mathscr F}_{\infty}}
\newcommand{\Orfm}{{\mathscr F}^{+}_{\infty}}
\newcommand{\bd}{\mathbf{\epsilon}}
\newcommand{\bD}{\mathbf{\Delta}}

\newcommand{\mfkDal}{{\mathfrak{D}}_{al}}
\newcommand{\mfkDl}{{\mathfrak{D}}_{l}}
\newcommand{\mfkd}{{\mathfrak d}}
\newcommand{\mfkdal}{{{\mathfrak d}_{al}}}
\newcommand{\mfkdl}{{{\mathfrak d}_{l}}}

\newcommand{\mfkDpal}{{\mathfrak D}^{\psi}_{al}}

\newcommand{\mfkdpal}{{{\mathfrak d}^{\psi}_{al}}}

\newcommand{\mfkg}{{\mathfrak g}}

\newcommand{\sfa}{{\mathsf a}}

\newcommand{\chA}{{\mathsf a}}
\newcommand{\chB}{{\mathsf b}}

\newcommand{\chF}{{\mathsf f}}
\newcommand{\chG}{{\mathsf g}}
\newcommand{\chH}{{\mathsf h}}

\newcommand{\chJ}{{\mathsf j}}

\newcommand{\chQ}{{\mathsf q}}

\newcommand{\cha}{{\mathsf a}}
\newcommand{\chb}{{\mathsf b}}
\newcommand{\chc}{{\mathsf c}}

\newcommand{\chf}{{\mathsf f}}
\newcommand{\chg}{{\mathsf g}}
\newcommand{\chh}{{\mathsf h}}
\newcommand{\chj}{{\mathsf j}}
\newcommand{\chk}{{\mathsf k}}
\newcommand{\chm}{{\mathsf m}}

\newcommand{\chq}{{\mathsf q}}

\newcommand{\cht}{{\mathsf t}}

\newcommand{\tchq}{\tilde{{\mathsf q}}}

\newcommand{\derA}{{\mathsf{A}}}
\newcommand{\derB}{{\mathsf{B}}}

\newcommand{\derF}{{\mathsf{F}}}
\newcommand{\derG}{{\mathsf{G}}}
\newcommand{\derH}{{\mathsf{H}}}
\newcommand{\derQ}{{\mathsf{Q}}}

\newcommand{\OL}{\mathcal{O}(L^{-\infty})}
\newcommand{\Or}{\mathcal{O}(r^{-\infty})}

\newcommand{\Br}{{\mathbb B}}

\newcommand{\Frechet}{{Fr\'{e}chet}}
\newcommand{\Mobius}{{M\"obius}}
\newcommand{\Cech}{{\v{C}ech}}

\newcommand{\partialarrow}{{\overset{\partial}{\longrightarrow}}}

\def \l {\left(}
\def \r {\right)}

\def \lal {\langle}
\def \ral {\rangle}

\title{Local Noether theorem for quantum lattice systems and topological invariants of gapped states}

\author{Anton Kapustin, Nikita Sopenko \smallskip\\ 
{\it California Institute of Technology, Pasadena, CA 91125, USA}}

\begin{document}

\maketitle

\abstract{We study generalizations of the Berry phase for quantum lattice systems in arbitrary dimensions. For a smooth family of gapped ground states in $d$ dimensions, we define a closed $d+2$-form on the parameter space which generalizes the curvature of the Berry connection. Its cohomology class is a topological invariant of the family. When the family is equivariant under the action of a compact Lie group $G$, topological invariants take values in the equivariant cohomology of the parameter space. These invariants unify and generalize the Hall conductance and the Thouless pump. A key role in these constructions is played by a certain differential graded \Frechet -Lie algebra attached to any quantum lattice system. As a by-product, we describe  ambiguities in charge densities and  conserved currents for arbitrary lattice systems with rapidly decaying interactions.}

\section{Introduction}

Over the last couple of decades there has been a growing interest in classifying phases of quantum matter at zero temperature and a  non-vanishing energy gap for local excitations. Such phases often do not fit into the Landau paradigm, and it is remarkably hard to describe the observables  which distinguish them in a mathematically precise way.\footnote{For 2d lattice systems it is often assumed that gapped phases are in 1-1 correspondence with unitary Topological Quantum Field Theories. But it is difficult to justify such an assumption, and there is no general way to determine TQFT data for a given gapped Hamiltonian. In higher dimension, the relations between gapped phases and TQFT is even less understood.}  

Phases of non-interacting fermions can be classified using K-theory, as pointed out in the pioneering works \cite{ryuetal,ryuetal2,kitaev2009periodic}. However, these methods do not generalize to interacting systems.
Since the notion of a phase, strictly speaking, applies only in  infinite-volume, it is natural to use methods of quantum statistical mechanics to construct invariants of gapped phases. Such invariants are usually called topological, because they do not vary in suitably defined families. Recently, these methods have been
used to define topological invariants of 1d \cite{ogataTindex,ogata2019classification,ogatabourne,kapustin2020classification} and 2d \cite{ThoulessHall,ogata2021h,sopenko2021} gapped lattice systems with symmetries and arbitrarily strong interactions which decay rapidly with distance. In Refs.  \cite{bachmann2012automorphic,moon2020automorphic} (see also  \cite{ThoulessHall,kapustin2020classification,sopenko2021}) a quantum phase was defined as an equivalence class of pure states of  many-body systems with respect to finite-time evolutions generated by Hamiltonians with rapidly decaying interactions. This is in accordance with a view \cite{QImeetsQM} that a gapped quantum phase is a pattern of entanglement of the ground-state wave function, while the Hamiltonian is unimportant. On the other hand, ground states of gapped local Hamiltonian are known to be quite special (for example, their entanglement entropy is believed to obey the area law), and these special properties are crucial for defining topological invariants. This raises the issue of characterizing gapped ground states of local Hamiltonians among all pure quantum states of many-body systems. 

In this paper we develop a systematic approach to the construction of topological invariants of quantum lattice systems. Our main tool is a certain differential graded \Frechet-Lie algebra (DGFLA) which can be attached to any quantum state $\psi$ of a lattice system. Recall that a differential graded Lie algebra is a $\ZZ$-graded vector space equipped with a grading-compatible Lie bracket and a differential (a degree-1 derivation of the Lie bracket which squares to zero). In our case the vector space is an infinite-dimensional \Frechet\ space and the bracket and the differential are continuous w.r. to \Frechet\ topology, so we get a DGFLA. Roughly speaking, it contains symmetries of $\psi$ and their higher analogs. Let $\psi$ be a gapped ground state of a Hamiltonian with interactions decaying faster than any power of distance. We call such states gapped. The key property of gapped states is that the cohomology of the corresponding DGFLA is trivial. This enables one to define topological invariants for any gapped state by solving a suitable Maurer-Cartan equation.

The first application of our method is the construction of invariants of smooth\footnote{In a sense defined below.} families of gapped states. They generalize the curvature of the Berry connection to the many-body setting and for $d$-dimensional lattice systems take values in the degree-$(d+2)$ de Rham cohomology of the parameter space. Higher Berry classes allow one to probe the topology of the space of gapped states and to detect high-codimension critical loci in the phase diagram where the gap closes \cite{diabolical}. In the case of {\it invertible
} quantum lattice systems, higher Berry classes were first discussed by A. Kitaev \cite{Kitaev_talk}; for general gapped systems explicit formulas for them were written down in \cite{kapustin2020higherA} using a choice of a family of Hamiltonians. Our construction makes it clear that higher Berry classes depend only on the family of states, not the Hamiltonians used to define them. \footnote{It is believed that for invertible states these classes are quantized. Ref. \cite{kapustin2020higherA} argued (non-rigorously) that this is the case for families parameterized by spheres. We plan to address quantization for generic families of invertible states elsewhere.}

More generally, given a manifold $\CM$ with an action of a compact Lie group $G$, one may consider smooth $G$-equivariant families of gapped states over $\CM$. We show how to attach a degree $d+2$ class in the $G$-equivariant cohomology of $\CM$ to any such family by solving a suitable Maurer-Cartan equation. In the case when $d=1$, $G=U(1)$, and $\CM=S^1$ with a trivial $G$-action, the corresponding invariant is the Thouless pump. In the case when $d=2$, $G=U(1)$, and $\CM=\text{pt}$, the corresponding invariant is the Hall conductance. Thus both the Hall conductance and the Thouless pump arise from a $G$-equivariant version of the Berry phase. More generally, 
when $\CM=\text{pt}$, the equivariant cohomology can be identified with the space of $G$-invariant polynomials on the Lie algebra $\mfkg$ of $G$. Thus, for any even $d$, our construction attaches a $G$-invariant polynomial of degree $(d+2)/2$ to a $G$-invariant gapped state in $d$ dimensions. This is in agreement with the TQFT approach, which predicts that topological invariants of such systems should take values in the space of Chern-Simons forms of degree $d+1$. It is well known that the latter arise from invariant polynomials on the Lie algebra. 

Some of our results have applications outside of the area of gapped phases. As a key preparatory step, for any quantum lattice system we define a DGFLA which encodes all local\footnote{In the sense that interactions are either finite-range or decay faster than any power of the distance.} Hamiltonians, their densities, currents, and their higher analogs. We prove that the cohomology of this DGFLA vanishes. This kinematic result implies, among other things, that any local continuous symmetry of a local lattice Hamiltonian gives rise to a local conserved current. Since the relation between symmetries and local conservation laws is usually referred to as the Noether theorem, we call the DGFLA attached to a lattice system the local Noether complex. 

The current corresponding to a symmetry is not unique, but the ambiguity can be completely described. Non-uniqueness of the energy current, in particular, is the cause of many complications in the theory of heat  transport, see e.g. \cite{CooperHalperinRuzin}. The effect of these ambiguities on the standard Kubo-Greenwood formulas for transport coefficients has been studied in \cite{gauge,KapSpo}. Hopefully, our results help clarify such issues, at least for lattice systems.

Our main technical tools are the Lieb-Robinson bound in the form proved in \cite{nachtergaele2006propagation,matsuta2017improving} and linear integral transforms introduced in \cite{osborne2007simulating}. One technical novelty is the use of a certain dense subalgebra $\SAal$ in the algebra of quasilocal observables $\SA$. This subalgebra is not complete with respect to the norm topology of $\SA$, but it is complete with respect to a topology defined by a non-decreasing family of norms labeled by a non-negative integer. Such topological vector spaces are known as (graded) \Frechet\ spaces. The central role played by $\SAal$ and its relatives necessitates a systematic use of calculus in \Frechet\ spaces as described, for example, in \cite{Hamilton}.

The paper is organized as follows. In Section 2 we set up the notation and define the algebra of almost local observables. In Section 3 we define DGFLAs attached to lattice systems, explain their relation to currents, and prove the ``Local Noether Theorem''.  In Section 4 we define DGFLAs attached to states and use them to construct topological invariants of smooth families of gapped states, including the $G$-equivariant case. We
also construct an example of a topologically non-trivial family of quasi-1d states associated to a $G$-invariant Symmetry Protected Topological state in two dimensions. This construction relates higher Berry curvature to non-abelian Hall conductance. The more technical results are relegated to  appendices.
\\

\noindent
{\bf Acknowledgements:}
We would like to thank Karl-Hermann Neeb and Bruno Nachtergaele for reading a preliminary draft of the paper and Bowen Yang for discussions. We are especially grateful to Bruno Nachtergaele for pointing out to us the relevance of the improved Lieb-Robinson bounds from \cite{matsuta2017improving,elsemachadonayakyao} and for informing us about his forthcoming work \cite{nachtergaelesimsyoungupcoming}. This research was supported in part by the U.S.\ Department of Energy, Office of Science, Office of High Energy Physics, under Award Number DE-SC0011632. A.K. was also supported by the Simons Investigator Award.
\\

\section{Preliminaries}

\subsection{The lattice} \label{ssec:lattice}

Let $\Lambda\subset\RR^d$ be a Delone set. That is, $\Lambda$ is both uniformly discrete:
\begin{equation} 
\bd :=\inf_{\substack{j,k\in\Lambda\\ j\neq k}} |j-k|  >0,
\end{equation}
and relatively dense:
\begin{equation}
\bD := \sup_{x\in\RR^d} \inf_{j\in \Lambda} |x-j|<\infty .
\end{equation}
In what follows we assume the metric on $\RR^d$ has been rescaled so that $\bD<1/2$. Then every open ball of diameter $1$ has a nonempty intersection with $\Lambda$, and the cardinality of the intersection is upper bounded by $C_d\bd^d$ for some $C_d$ depending only on $d$. Elements of $\Lambda$ are called sites. The set of finite subsets of $\Lambda$ will be denoted ${\rm Fin}(\Lambda)$. It is a directed set with respect to inclusion.

The complement of a subset $X \subset \RR^d$ is denoted by $\overline{X}$.  The diameter of a subset $X\subset \RR^d$ is defined by
\beq
\text{diam}(X) := \sup_{x,y \in X} |x-y|.
\eeq
The distance between any subsets $X$ and $Y$ is defined by
\beq
\dist(X,Y) := \inf_{x \in X, y \in Y} |x-y|.
\eeq
An open ball of radius $r$ with the center at $x \in \RR^d$ is defined by $B_x(r) := \{y\in\RR^d : |x-y|< r\}$. Its complement $\overline{B_x(r)}$ in $\RR^d$ is also denoted $B^c_x(r)$. 
%More generally, for any finite set $X\subset\RR^d$ we define %$B_X(r):=\{y\in\RR^d : {\rm dist}(y,X) < r\}$. Note that %$B_X(0)=\emptyset$ for any such $X$. 
For any $Y\subset\RR^d$ we denote by $\chi_Y$ the indicator function of $Y$: $\chi_Y(x)=1$ if $x\in Y$ and $\chi_Y(x)=0$ if $x\notin Y$. 
% Finally, $\theta:\RR\ra\RR$ will denote the indicator function of $[0,\infty)$.

\subsection{Functions}

We will denote by $\Orf$ the space of functions $\RR_{r \geq 0} \ra \RR$
which decay at infinity faster than any power of $r$. We denote the subset of monotonically decreasing non-negative functions by $\Orfm$. The space $\Orf$ is a {\Frechet} space topologized by a sequence of norms 
\beq
\|f\|_{\alpha} := \sup_{r \geq 0} (1+r)^{\alpha} |f(r)|,\ \alpha\in\NN_0.
\eeq
A sequence $f_k, k\in\NN,$ converges to 0 iff $\|f_k\|_{\alpha} \to 0$ for all  $\alpha$.

Finally, $\theta:\RR\ra\RR$ will denote the indicator function of $[0,\infty)$.

\subsection{Observables}
To each site $j\in\Lambda$ we attach an on-site algebra $\SA_j = \End(\CV_j)$ for some Hilbert space $\CV_j$ of  dimension $d_j = \dim \CV_j$. For any $\Gamma\in \text{Fin}(\Lambda)$ we define 
\beq
\SA_\Gamma=\bigotimes_{j\in\Gamma} \SA_j.
\eeq
By convention, $\SA_\emptyset={\mathbb C}$. For $X \subset \RR^d$ we let
\beq
\SA_X := \underset{\Gamma \subset X}\varinjlim\, \SA_{\Gamma}.
\eeq
The algebra of local observables is defined as
\beq
\SAl := \underset{\Gamma}\varinjlim\, \SA_{\Gamma} .
\eeq
It is a normed $*$-algebra whose center is isomorphic to $\mathbb{C}$. For any $j\in \Lambda$ we let $\Pi_j:
\SAl\ra\SAl$ be the averaging over local unitaries supported at $j$, i.e. for any $\CB\in\SAl$ we let $\Pi_j(\CB)=\int  U \CB U^* d\mu_j(U)$, where $d\mu_j(U)$ is the Haar measure on the group of unitary elements of $\SA_j$ normalized so that $\int 1\, d\mu_j(U)=1$. It is well-known that $\Pi_j$ is positive and does not increase the norm. The maps $\Pi_j$ for different $j$ commute, therefore for any $\Gamma \in \text{Fin}(\Lambda)$ we can define  $\Pi_{\Gamma}:=\prod_{j\in \Gamma} \Pi_j$. For $X\subset \RR^d$ we let
\beq
\Pi_X := \underset{\Gamma \subset X} \varinjlim\, \Pi_{\Gamma}.
\eeq
We will use without comment various obvious properties of the averaging maps $\Pi_X$, for example $\Pi_X\circ \Pi_Y=\Pi_{X\cup Y}$ for any $X,Y \subset \RR^d$. The following simple estimate is also sometimes useful:
\beq\label{eq:commutatorlowerbound}
\|\CA-\Pi_X(\CA)\|\leq \sup_{\CB\in\SA_X}\frac{\|[\CA,\CB]\|}{\|\CB\|}
\eeq

Note that $\Pi_\Lambda=\underset{\Gamma}\varinjlim\, \Pi_{\Gamma}$ maps any $\CB\in\SAl$ to the center of $\SAl$ and thus can be thought of as a state (i.e. a positive linear function) on $\SAl$. We will denote it $\langle\cdot\rangle_\infty$ since it is the infinite temperature state. This state can also be defined as the unique tracial state on $\SAl$. More generally, for any $X\subset \RR^d$ we define the conditional expectation value $\SAl\ra\SA_X$ by 
\beq 
\CA\mapsto \CA\vert_X := \Pi_{\overline{X}}(\CA).
\eeq 
Conditional expectation value does not increase the norm. 

In quantum statistical mechanics, one usually defines the algebra of quasi-local observables $\SA$ as the norm completion of $\SAl$. The conditional expectation value extends to $\SA$.
In this paper the $C^*$-algebra $\SA$ will not be used much. Instead a certain dense  sub-algebra $\SAal\subset \SA$ will play a central role. Its  definition depends on the metric structure of $\Lambda$. In contrast, $\SA$ depends only on the list of dimensions of the spaces $\CV_j$. Roughly speaking, $\SAal$ consists of quasi-local observables which on a ball of radius $r$ can be approximated by local observables with an $\Or$ error. 
% For this reason, we call $\SAal$ the \textit{algebra of almost local observables}.

% One possible definition of $\SAal$ is as follows.
To define $\SAal$ we first introduce some notation. 
%The conditional expectation value of an observable $\CA \in \SA$ on a subset $X \subset \Lambda$, that is an average over the tracial state in the complement of $X$, is denoted by $\CA|_X$. It satisfies $\|\CA|_X\| \leq \CA$. 
For any $\CA\in\SA$, $x \in\RR^d$, and $r\geq 0$ let
% \beq\label{eq:fjAr}
% f_x(\CA,r):= \lVert \CA-\CA|_{B_x(r)}\rVert.
% \eeq
\beq\label{eq:fjAr}
f_x(\CA,r):=\inf_{\CB\in\SA_{B_x(r)}} \lVert \CA-\CB\rVert.
\eeq
This is a monotonically decreasing non-negative function of $r$ which takes values in $[0,\|\CA\|]$, approaches zero as $r\ra\infty$, and for any $x \in\RR^d$ and any $r\geq 0$ satisfies the triangle inequality:
\beq
f_x(\CA_1+\CA_2,r)\leq f_x(\CA_1,r)+f_x(\CA_2,r).
\eeq
Thus for a fixed $x\in \RR^d$ and fixed  $r\geq 0$ $f_x(\cdot,r)$ is a seminorm on $\SA$ (see Appendix \ref{app:frechet} for a brief reminder on seminorms and the topologies on vector spaces defined by countable families of seminorms). Note for future use that if $\Gamma \in {\rm Fin}(\Lambda)$ and $\CA\in\SA_{\Gamma}$, then $f_x(\CA,r)=0$ for any $r>\diam(\{x\} \cup \Gamma)$, and thus $f_x(\CA,r)\leq \|\CA\| \theta(r-\diam(\{x\} \cup \Gamma))$. 
% Note also that $f_x(\CA,0)=\|\CA-1\cdot \langle\CA\rangle_\infty\|$. 
% Here $\langle\cdot\rangle_\infty$ is the extension of the state $\Pi_\Lambda$ to $\SA$.

Given $b(r)\in\Orf$ we will say that an observable $\CA\in\SA$ is $b$-localized at $x$ if $f_x(\CA,r)\leq b(r)$ for all $r$. Note that for such an observable $f_x(\CA,r) \in \Orfm$.

For $j\in \Lambda$, let us define the norms
\beq\label{eq:primednorms}
 \|\CA\|'_{j,\alpha}:=\lVert \CA\rVert +\sup_r (1+r)^\alpha f_j(\CA,r),\ \alpha\in\NN.
\eeq

\begin{lemma} \label{lma:Aal}
Let $\SA_{al}$ be a $\ast$-subalgebra $\SA_{al} \subset \SA$, and let us fix $j \in \Lambda$. The following characterizations of $\SAal$ are all equivalent:
\begin{itemize}
    \item[(1)] $\SAal$ is a subspace of $\SA$ defined by the condition $\|\CA\|'_{j,\alpha}<\infty$ for all $\alpha\in\NN$.
    \item[(2)] $\SAal$ consists of elements $\CA \in \SA$ which are $b$-localized at $j$ for some $b(r) \in \Orfm$.
    \item[(3)] $\SAal$ is the completion of the algebra $\SAl$ with respect to the norms $\|\cdot\|'_{j,\alpha}$, $\alpha\in\NN$.
\end{itemize}
$\SAal$ thus defined does not depend on $j$.
\end{lemma}

\begin{proof}
The implication (1) $\Rightarrow$ (2) follows directly from the definition of the norms $\|\cdot\|'_{j,\alpha}$. The definition of the completion with respect to norms implies (3) $\Rightarrow$ (1).

To show that (2) implies (3), let $\{\CA^{(n)} \}$ be a sequence of observables $\CA^{(n)} \in \SA_{B_j(n)}$ for $n \in \NN$ for which the infimum in the definition of $f_j(\CA,n)$ is reached: $f_j(\CA,n)=\|\CA-\CA^{(n)}\|.$ We have
% To show that (2) implies (3), let $\{\CA^{(n)} \}$ be a sequence of observables $\CA^{(n)} = \CA|_{B_j(n)}$ for $n \in \NN$. We have
% \begin{multline}
% \|\CA^{(n+1)}-\CA^{(n)}\|'_{j,\alpha} = \|\CA^{(n+1)}-\CA^{(n)}\| + \sup_r (1+r)^{\alpha} f_j(\CA^{(n+1)}-\CA^{(n)},r) \leq \\ \leq \|\CA^{(n+1)}-\CA^{(n)}\| (1+ (n+2)^{\alpha})  \leq \\ \leq (1+ (n+2)^{\alpha}) (f_j(\CA,n+1)+f_j(\CA,n)).
% \end{multline}
\begin{multline}
\|\CA-\CA^{(n)}\|'_{j,\alpha} = \|\CA-\CA^{(n)}\| + \sup_r (1+r)^{\alpha} f_j(\CA-\CA^{(n)},r) \leq \\ \leq f_j(\CA,n)+\sup_{r\geq n} (1+r)^\alpha f_j(\CA,r).
\end{multline}
Therefore $\{\CA^{(n)} \}$ converges in the norms $\|\cdot\|'_{j,\alpha}$ to $\CA$.

Independence of $j$ is manifest in the characterization (2).
\end{proof}

\begin{definition}
The space of almost local observables  $\SAal$ is the  $\ast$-subalgebra satisfying the equivalent characterizations of Lemma \ref{lma:Aal}.
\end{definition}

As recalled in Appendix \ref{app:frechet}, one can use the family of norms $\|\cdot\|'_{j,\alpha}$ for a fixed $j\in\Lambda$ to define a topology on $\SAal$. In Appendix \ref{app:Aal} we show that this topology on $\SAal$ does not depend on $j$ and turns $\SAal$ into a \textit{\Frechet}  \textit{algebra}.

We denote by $\mfkd_{\Gamma}$, $\mfkdl$ and $\mfkdal$ the real subspaces of $\SA_{\Gamma}$, $\SAl$ and $\SAal$, respectively, defined 
by two conditions: $\CA^*=-\CA$ (anti-self-adjoint) and $\langle \CA \rangle_\infty =0$ (traceless). All these spaces are real Lie algebras with respect to the commutator. The Lie algebra $\mfkdal$ is the completion of $\mfkdl$ with respect to the norms $\|\cdot\|'_{j,\alpha}$. Since the bracket is continuous, it is a \textit{\Frechet-Lie algebra}. 

Note that for any  $\CA\in\mfkdl,\mfkdal$  we have $\|\CA\| \geq f_x(\CA,0) \geq \|\CA\|/2$ for any $x \in \RR^d$. Therefore, upon restriction to traceless anti-self-adjoint observables one can replace the norms (\ref{eq:primednorms}) with an equivalent set of norms
\beq\label{eq:norms}
\|\CA\|_{j,\alpha}:=\sup_r (1+r)^\alpha f_j(\CA,r),\ \alpha\in\NN_0 .
\eeq
In this paper we will mostly work with traceless anti-self-adjoint observables and then will use the norms (\ref{eq:norms}). Two other equivalent sets of norms are defined in Appendix \ref{app:normfamilies}.

\section{Complexes of currents}

\subsection{Hamiltonians and interactions}

A Hamiltonian for a lattice system is an unbounded densely-defined real derivation of the algebra $\SA$. All Hamiltonians of physical interest have the form
\beq
\CA\mapsto \sum_{\Gamma\in {\rm Fin}(\Lambda)} [\Phi_\Gamma,\CA],
\eeq
where $\Phi_\Gamma\in \mfkd_\Gamma$ satisfies $\Phi_\Gamma^*=-\Phi_\Gamma$. Typically one also assumes that $\sup_\Gamma \lVert\Phi_\Gamma\rVert <\infty$. In the mathematical physics literature the function $\Phi:{\rm Fin}(\Lambda)\ra\mfkdl$, $\Gamma\mapsto\Phi_\Gamma$, is known as an interaction. 

The domain of such a derivation depends on how rapidly $\lVert\Phi_\Gamma\rVert$ decays with the size of $\Gamma$. As a minimum, it should be  defined everywhere on  $\SAl$. Further, it should be possible to exponentiate a physically sensible derivation to an automorphism of $\SA$. Before we can describe suitable decay conditions, however, we need to deal with the fact that the map from interactions to derivations is many-to-one and that there is a large ``gauge freedom" in choosing the function $\Phi$ for a given derivation. Any decay  condition on $\Phi$ should respect this freedom. Unfortunately, it is not straightforward to describe this ``gauge freedom''. To rectify the situation, one may impose a suitable ``gauge condition'' on $\Phi$ so that a derivation determines $\Phi$ uniquely. One natural condition is to demand that for any proper inclusion $\Gamma'\subset\Gamma$ and any $\CA\in\mfkd_{\Gamma'}$ one has $\langle\Phi_\Gamma \CA\rangle_\infty=0$. This condition implies that $\Phi_\Gamma$ is not localized on any proper subset of $\Gamma$. Its advantage is that it does not depend on any choices. But this gauge condition is difficult to work with when the interaction has exponential or slower than exponential decay because the number of finite subsets of $B_j(r)\cap \Lambda$ grows as $e^{C r^d}$. Later we will describe a convenient but non-canonical ``gauge condition'' on $\Phi$. 

Another drawback of describing a Hamiltonian via an interaction $\Phi$ is that there is no natural notion of energy density (and therefore also of energy current). As an alternative, we may study  derivations of the form
\beq\label{eq:Fder}
\delta_{\chh}: \CA\mapsto \sum_{j\in\Lambda} [\chh_j,\CA],
\eeq
where each $\chh_j$ is a traceless anti-self-adjoint observable which in some sense is localized in the neighborhood of the site $j$ and can be interpreted as $i$ times the energy density on this site. The gauge freedom is present in this approach as well, since the observables $\chh_j$ are not uniquely determined by the derivation $\delta_\chh$. However, it is fairly straightforward to parameterize this gauge freedom and to define a class of derivations which is mathematically natural and is large enough to describe lattice systems with rapidly decaying interactions.

For that, we introduce a certain chain complex, which we call {\it the complex of currents}. We first describe this complex for finite-range interactions, and then show how to complete it to a  complex of rapidly decaying currents using a family of norms, in the same way as the algebra $\SAl$ can be completed to $\SAal$.

\subsection{The complex of finite-range currents}

\subsubsection{Definition of the complex}

% We say that $\CA \in \mfkdl$ is $(C,R)$-localized at $j \in \Lambda$ if $\CA \in \SA_{B_j(R)}$ and $\|\CA\|\leq C$. 

For any non-negative integer $q$ we define a uniformly local (UL) $q$-chain as a skew-symmetric function $\cha: \Lambda^{q+1}\ra \mfkdl$ for which there are constants $C>0$ and $R>0$ such that for any $\{j_0,j_1,...,j_q\} \subset \Lambda$ and any $a \in\{ 0,1,...,q\}$ we have $\cha_{j_0 ... j_q} \in \SA_{B_{j_a}(R)}$ and $\|\cha_{j_0 ... j_q}\| \leq C$. The smallest possible value of $R$ is called the range of $\cha$. If the distance between any $j_a$ and $j_b$ is greater than $2R$, then $\cha_{j_0 ... j_q} = 0$. The space of UL $q$-chains will be denoted $C_q(\mfkdl)$. The boundary operator 
$\partial_{q}: C_{q}(\mfkdl)\ra C_{q-1}(\mfkdl)$
has the form 
\beq\label{eq:partial}
\left(\partial_{q}\cha\right)_{j_1\ldots j_q}:=\sum_{j_0\in\Lambda} \cha_{j_0 j_1\ldots j_q}.
\eeq
Since $\cha_{j_0 j_1\ldots j_q}$ is skew-symmetric in indices, we have $\partial_{q}\circ\partial_{q+1}=0$ for all $q$. 
Note that to any UL 0-chain $\cha:j\mapsto\cha_j$ one can attach a derivation of $\SAl$
\beq
\delta_{\cha}: \CA \mapsto \sum_{j\in\Lambda} [\cha_j,\CA],
\eeq
We will call derivations of this form UL derivations. Physically, they correspond to finite-range Hamiltonians. It is easy to see that a UL derivation $\delta_\cha$ is not affected if we add to $\cha$ an exact 0-chain. We will show that this is the only gauge freedom associated to UL derivations. That is, we will show that the space of UL derivations is isomorphic to the zeroth homology $H_0(\mfkdl)$ of the complex
\beq \label{eq:ULcomplex}
\ldots \stackrel{\partial_{2}}{\ra} C_{1}(\mfkdl)\stackrel{\partial_{1}}{\ra} C_0(\mfkdl) \stackrel{}{\ra} 0.
\eeq

\subsubsection{Homology}

To compute the homology of $(C_{\bullet}(\mfkdl),\partial)$, we introduce the following definition.
\begin{definition}
A brick is a subset of $\RR^d$ of the form $\{(x_1,\ldots,x_d): n_i\leq x_i < m_i,\ i=1,\ldots, d\}$, where $n_i$ and $m_i$ are integers satisfying  $n_i<m_i$. The empty brick is the empty subset.
\end{definition}
We denote the set of all bricks in $\RR^d$ together with the empty brick by $\Br_d$. $\Br_d$ is a poset with respect to inclusion.\footnote{In fact it is a lattice. That is, every two elements of $\Br_d$ have a join and a meet in $\Br_d$.} Any finite subset of $\Lambda$ is contained in some brick.

Recall that for any $Y\subset \RR^d$ we let $\mfkd_Y=\mfkdl\cap \SA_Y$. Clearly, if $Z\subset Y$, then $\mfkd_Z\subset\mfkd_Y$. For any $Y\in\Br_d$ we define $\mfkd^Y$ to be the orthogonal complement of the subspace 
\beq
\sum_{\substack{Z\in\Br_d\\ Z\subsetneq Y}}\mfkd_Z 
\eeq
in $\mfkd_Y$ with respect to the inner product $\langle \CA,\CB\rangle :=\langle \CA^*\CB\rangle_\infty$. Elements of $\mfkd^Y$ are anti-self-adjoint traceless local observables which are localized on $Y$ but not on any brick which is a proper subset of $Y$. It is easy to see that for any $Y\in\Br_d$ we have a direct sum decomposition $\mfkd_Y=\bigoplus_{Z\in\Br_d,Z\subseteq Y} \mfkd^Z$ and $\mfkdl=\bigoplus_{Y\in \Br_d}\mfkd^Y$. For any $\CA\in\mfkdl$ we will denote by $\CA^Y$ its component in $\mfkd^Y$. For an example of a brick decomposition see Fig. \ref{fig:brickdecomposition}. Clearly, if $\CA \in \mfkd_{X}$, then $\CA^Y=0$ whenever $Y\cap X = \emptyset$.  
%It is also easy to see that $\CA^Y=0$ whenever $\diam(Y)>\sqrt d(\diam(X)+2)$. 
Additional properties of the brick expansion can be found in Appendix \ref{app:bricks}.

\begin{figure}
    \centering
\begin{tikzpicture}[scale=.8]
\filldraw[color=red!60, fill=red!5, very thick](1-0.1,1-0.1) rectangle (5+0.1,4+0.1);
\filldraw[color=red!60, fill=red!5, very thick](1,1) rectangle (3,3);
\filldraw[color=red!60, fill=red!5, very thick](1+.1,1+.1) rectangle (3-.1,2-.1);

% \draw[red, very thick] (0,0) -- (1.5*3.4641,1.4*2);
% \draw[red, very thick] (0,0) -- (-1.5*3.4641,1.4*2);
% \draw[red, very thick] (0,0) -- (0,-5);

\foreach \x in {0.5,1.5,...,5.5}{                           % Two indices running over each
    \foreach \y in {0.5,1.5,...,4.5}{                       % node on the grid we have drawn 
    \node[draw,black,circle,inner sep=.7pt,fill] at (\x,\y) {}; % Places a dot at those points
    }
}

\node  at (2,1+0.5) {$Y_1$};
\node  at (2,2+0.5) {$Y_2$};
\node  at (4,3) {$Y_3$};

% \draw [gray, very thick] (4,0) arc [radius=4, start angle=0, end angle= 360];
% \node  at (0,1.5*2+0.3) {$A_2$}; 
% \node  at (-1.5*1-0.4,-1.5*1-0.4) {$(\frac12,\frac12)$};
% \node  at (7.5*1+0.4,7.5*1+0.4) {$(\frac92,\frac92)$};
% \node  at (1.5*1.7321,-1.5*1-0.4) {$A_1$};
% \node  at (3.5,3.5) {$D$};

\end{tikzpicture}
    \caption{Let $\Lambda \subset \RR^2$ be the lattice $(\ZZ + \frac12)^2$ with two-dimensional on-site Hilbert spaces. Let $\sigma^{i}_{(x,y)}$ be the Pauli matrix observables on site $(x,y)$. The brick decomposition of $\CA = \sigma^z_{(\frac32,\frac32)} \sigma^x_{(\frac52,\frac32)} + \sigma^x_{(\frac32,\frac32)} \sigma^y_{(\frac52,\frac52)} + \sigma^x_{(\frac32,\frac32)} \sigma^z_{(\frac32,\frac52)} \sigma^z_{(\frac92,\frac72)} + \sigma^x_{(\frac32,\frac32)} \sigma^y_{(\frac92,\frac52)} \sigma^z_{(\frac72,\frac72)}$ has three terms $\CA^{Y_1} = \sigma^z_{(\frac32,\frac32)} \sigma^x_{(\frac52,\frac32)}$, $\CA^{Y_2} = \sigma^x_{(\frac32,\frac32)} \sigma^y_{(\frac52,\frac52)}$, $\CA^{Y_3}  = \sigma^x_{(\frac32,\frac32)} \sigma^z_{(\frac32,\frac52)} \sigma^z_{(\frac92,\frac72)} + \sigma^x_{(\frac32,\frac32)} \sigma^y_{(\frac92,\frac52)} \sigma^z_{(\frac72,\frac72)}$ for bricks $Y_1,Y_2,Y_3$ shown on the picture.}
    \label{fig:brickdecomposition}
\end{figure}

% It is also shown there that any $b$-localized $\CA\in\mfkdal$ can be written as an absolutely convergent (in the \Frechet\ topology) sum 
% \beq
% \CA=\sum_{Y\in \Br_d} \CA^Y,
% \eeq
% and for any $j\in\Lambda$ one has an estimate
% \beq\label{AYestimate}
% \| \CA^Y \| \leq  4^d b(\diam(\{j\}\cup Y)).
% \eeq

% Using this result, one can easily prove the vanishing of $H_k(\mfkdl)$ and $H_k(\mfkdal)$ in positive degree. 

Let $h_q : C_q(\mfkdl) \to C_{q+1}(\mfkdl)$ be a map defined by
\beq \label{eq:contrho}
h_q(\cha)_{j_0\ldots j_{q+1}}=\sum_{Y\in \Br_d}\sum_{k=0}^{q+1} (-1)^k \frac{\chi_Y(j_k)}{|Y\cap\Lambda|} \chA^{Y}_{j_0\ldots \widehat j_k \ldots j_{q+1}}.
\eeq
where $\widehat{j}_k$ denotes the omission of $j_k$. Since the sum over $Y$ is finite, $h_q$ is well-defined. Note that $h_{q-1}\circ\partial_q +\partial_{q+1}\circ h_q={\rm id}$ for any $q>0$. Therefore we have the following 
\begin{theorem} \label{thm:poincareUL}
$H_{q}(\mfkdl)=0$ for all $q>0$.
\end{theorem}
% \begin{proof}
% For $\chA\in C_q(\mfkdl)$ let
% \beq \label{eq:contrho}
% \chB_{j_0\ldots j_{q+1}}=\sum_{Y\in \Br_d}\sum_{k=0}^{q+1} (-1)^k \frac{\chi_Y(j_k)}{|Y\cap\Lambda|} \chA^{Y}_{j_0\ldots \widehat j_k \ldots j_{q+1}},
% \eeq
% The sum over $Y\in \Br_d$ is finite, thus $\chB_{j_0\ldots j_{q+1}}$ is a well-defined local observable. Furthermore, if $\chA$ has range $R$, then $\chA^Y_{j_0\ldots j_q}=0$ whenever $\diam(Y)>2R+2$. Thus eq. (\ref{eq:contrho}) defines a linear map $h_q:C_q(\mfkdl)\ra C_{q+1}(\mfkdl)$ with the desired properties and $H_q(\mfkdl)=0$ for $q>0$. 

% If $\chA\in C_q(\mfkdal)$, we can define $\chB\in C_{q+1}(\mfkdal)$ by the same formula. The sum over $Y$ is absolutely \Frechet-convergent by the estimate (\ref{AYestimate}). The same computation as in the UL case shows that the resulting map $h_q$ is a contracting homotopy for $q>0$. 
% We omit the straightforward check that for $q\geq 0$  $h_q$ is a continuous UAL map. 
% \end{proof}
To describe the homology in degree zero, we introduce the following definition. 
\begin{definition}
Let $\mfkDl$ be the space of bounded functions $\derA:\Br_d\ra \mfkdl$, $Y\mapsto \derA^Y$, such that 
\begin{itemize}
    \item $\derA^Y\in \mfkd^Y$,
    \item $\derA^Y=0$ for all $Y$ of sufficiently large diameter.
\end{itemize}
% Let $\mfkDal$ be the space of functions $\derF:\Br_d\ra \mfkdl$, $Y\mapsto \derF^Y$ such that 
% \begin{itemize}
%     \item $\derF^Y\in \mfkd^Y$,
%     \item $\|\derF\|_\alpha:=\sup_{Y\in \Br_d} (1+\diam(Y))^\alpha \|\derF^Y\|<\infty,\quad \forall  \alpha\in\NN_0.$
% \end{itemize}
\end{definition}
\noindent
We also extend the definition of $\p_q$ and $h_q$ by introducing maps $\p_0:C_0(\mfkdl) \to \mfkDl$ and $h_{-1}: \mfkDl \to C_0(\mfkdl)$ defined by 
\beq
(\p_{0} \cha)^Y := \sum_{j \in \Lambda} \cha^Y_j,
\eeq
\beq \label{eq:hminus1}
h_{-1}(\derA)_j :=\sum_{Y\in \Br_d} \frac{\chi_j(Y)}{|Y\cap\Lambda|}\derA^Y.
\eeq
Note that the sums on the r.h.s. of both expressions are finite and thus well-defined. We have 
\beq 
h_{q-1}\circ\partial_q +\partial_{q+1}\circ h_q={\rm id}, \quad q \geq 0,
\eeq 
\beq 
\p_0 \circ h_{-1} = {\rm id}.
\eeq
\begin{theorem} \label{thm:ULH0}
$H_0(\mfkdl)$ is isomorphic to $\mfkDl$ and to the space of UL derivations. 
% $H_0(\mfkdal)$ is isomorphic to $\mfkDal$ and to the space of UAL derivations.
\end{theorem}
\begin{proof} 
To any UL 0-chain $\cha$ we attach a function $Y\mapsto \derA^Y=\sum_j \cha^Y_j$. The sum over $j$ is finite. It is easy to see that this function belongs to $\mfkDl$. Furthermore, it is trivially checked that if the 0-chain $\chA$ is exact, then the corresponding function vanishes. Thus we get a map $\rho_{l}: H_0(\mfkdl)\ra\mfkDl$. 

This map is surjective because a right inverse exists: to an element $\derA$ of $\mfkDl$ one can attach a UL 0-chain $h_{-1}(\derA)$. To prove injectivity, suppose that a UL 0-chain $\cha$ satisfies $\sum_j \cha_j^Y=0$ for all $Y\in \Br_d$. For any $j,k\in\Lambda$ let
\beq
\chb_{jk}=\sum_{Y\in \Br_d} \frac{\chi_j(Y)}{|Y\cap\Lambda|} \cha_k^Y -\sum_{Y\in \Br_d} \frac{\chi_k(Y)}{|Y\cap\Lambda|}\cha_j^Y.
\eeq
It is straightforward to check that the collection of observables $\chb_{jk}$ defines a UL 1-chain, and that $\partial \chb=\cha$. Thus the map $\rho_{l}$ is an isomorphism.

By definition of UL derivations, the map from $H_0(\mfkdl)$ to UL derivations defined by (\ref{eq:Fder}) is surjective. To prove injectivity, suppose $\delta_\cha=0$ for some UL 0-chain $\cha$. Then for any $\CA\in\SAl$ we have
\beq\label{eq:adFYsum}
\sum_{Y\in \Br_d} [\derA^Y,\CA]=0,
\eeq
where $\derA^Y=\sum_j \cha^Y_j$ is an element of $\mfkd^Y$. Let us pick an arbitrary brick $Z\in \Br_d$. Then for any $\CA\in\SA_Z$ the sum in (\ref{eq:adFYsum}) truncates to those $Y$ which have a nonempty intersection with $Z$. Thus if we define a traceless local observable $\CB$ by
\beq
\CB=\sum_{Y\in \Br_d, Y\cap Z\neq\emptyset} \derA^Y,
\eeq
then $\CB=1_Z\otimes \tilde \CB$, where $\tilde\CB\in \SA_{\overline{Z}}$. Therefore $\CB^X=0$ for any brick $X$ such that $X\subseteq Z$. On the other hand, from the definition of $\CB$ we have that for such bricks $\CB^X=\derA^X$. Since $Z$ was arbitrary, we conclude that $\derA^Y=0$ for all $Y\in \Br_d$.
\end{proof}

Thus, the augmented complex
\beq  \label{eq:ULNC}
\ldots \stackrel{\partial_{2}}{\ra} C_{1}(\mfkdl)\stackrel{\partial_{1}}{\ra} C_0(\mfkdl) \stackrel{\partial_{0}}{\ra} \mfkDl \stackrel{}{\ra} 0.
\eeq 
is contractible with a contracting homotopy $h_q$,\, $q \geq -1$. We call it the {\it uniformly local Noether complex}. It is graded by integers $q\geq -01$.
% \begin{remark}
% The space of anti-self-adjoint UAL derivations $\mfkDal$ can be regarded as a subspace of the space of interactions satisfying the following ``gauge condition'': $\Phi_\Gamma\in\mfkd^\Gamma$ if $\Gamma\in \Br_d$ and $\Phi_\Gamma=0$ if $\Gamma\notin \Br_d$.
% \end{remark}

% We conclude this subsection with the following definition.
% \begin{definition}
% The uniformly local Noether complex is the augmented complex $C_{\bullet}(\mfkdl)  \stackrel{\partial_{0}}{\ra} \mfkDl$. 
% \end{definition}
% \noindent
% The reason for this terminology will become apparent in the next sections. Theorems \ref{thm:poincareUL} and \ref{thm:ULH0} imply that UL Noether complexes has trivial homology. 
For any $\derF\in\mfkDl$ we denote the action of the corresponding UL derivation on $\CA\in\mfkdl$ by $\derF(\CA)$. Explicitly, $\derF(\CA)=\sum_Y [\derF^Y,\CA].$

\subsubsection{Brackets}

An important property of the space of UL derivations $\mfkDl$ is that it has the structure of a Lie algebra. This is easiest to see if we identify it with $H_0(\mfkdl)$. For given $\derF,\derG \in \mfkDl$ and $\chf,\chg \in C_0(\mfkdl)$ such that $\derF=\p \chf$ and $\derG=\p \chg$, the Lie bracket can be defined by
\beq \label{eq:Dliebracket}
\{\derF,\derG\} : = \p ([\chf,\chg\})
\eeq
where the components of a UL 0-chain $[\chf,\chg\}$ are defined as a finite sum
\beq
[\chf,\chg\}_j := \sum_{k \in \Lambda} [\chf_k,\chg_j].
\eeq
% It is easy to see that if both $\chF$ and $\chG$ are UAL 0-chains, then $[\chF,\chG\}$ is a UAL 0-chain as well. 

The bracket $[\cdot,\cdot\}$ on 0-chains is not skew-symmetric and awkward to work with. But one can express it through a more natural structure which exists on the augmented complex $C_{\bullet}(\mfkdl)\rightarrow\mfkDl$: the structure of a 1-shifted dg-Lie algebra. This means that there is a degree $1$ bracket $\{\cdot,\cdot\}$ on the augmented complex which is graded-skew-symmetric:
\beq
\{\chF,\chG\} = - (-1)^{(|\chF|+1)(|\chG|+1)} \{ \chG, \chF \},
\eeq
satisfies the graded Jacobi identity:
\begin{multline}
(-1)^{(|\chF|+1)(|\chH|+1)}\{ \chF ,\{ \chG , \chH \} \}  + (-1)^{(|\chG|+1)(|\chF|+1)} \{ \chG ,\{ \chH , \chF \} \} + \\ + (-1)^{(|\chH+1)(|\chG|+1)} \{ \chH ,\{ \chF , \chG \} \} = 0,
\end{multline}
and the graded Leibniz rule:
\beq
\p \{ \chF,\chG \} = \{ \p \chF, \chG \} + (-1)^{|\chF|+1} \{ \chF, \p \chG \}.
\eeq
Here $|\cdot|$ denotes the degree of a chain.
% Here $\partial$ denotes the differential on the augmented complexes.
For $\chf \in C_p(\mfkdl)$, $\chg \in C_q(\mfkdl)$, $\derF\in \mfkDl$ the bracket is defined by
\beq
\{\chf,\chg\}_{j_0...j_{p+q+1}} := \frac{1}{p!q!} [\chF_{j_0...j_{p}},\chG_{j_{p+1} ... j_{p+q+1}}] + \text{(signed permutations)}.
\eeq
\beq
\{\derF,\chG\}_{j_0\ldots j_q} := \derF\left(\chG_{j_0\ldots j_q}\right).
\eeq
while the bracket of two UL derivations is defined to be their Lie bracket eq. (\ref{eq:Dliebracket}). Then for any two UL 0-chains $\chF,\chG$ we can write
\beq
[\chF,\chG\}=\{\partial\chF,\chG\}.
\eeq
The non-skew-symmetric bracket $[\cdot,\cdot\}$ on 0-chains is an example of a ``derived bracket'' \cite{derivedbrackets}. 

There is an injective Lie algebra homomorphism from $\mfkdl$ to the Lie algebra of UL derivations $\mfkDl$ which sends $\CB\in\mfkdl$ to the derivation $\CA\mapsto [\CB,\CA]$. One can describe the image of this homomorphism more intrinsically by making the following definition.
\begin{definition}
An element $\derA\in\mfkDl$ is called summable if $\derA^Y \neq 0$ only for finitely many bricks.
\end{definition}
\noindent 
Physically, summable UL derivations correspond to interactions which are localized at a point. Obviously, summable UL derivations form a Lie sub-algebra of $\mfkDl$. In the following we identify it with $\mfkdl$.

\subsubsection{Integration}

If we interpret a 0-chain $\chb\in C_0(\mfkdl)$ as a density of energy or some other physical quantity, then it is natural to define energy in a region $A$ as a derivation $\chb_A\in\mfkDl$ which acts on $\CA\in\mfkdl$ by $\chb_A(\CA)=\sum_{j\in A} [\chb_j,\CA].$ Equivalently, $\chb_A^Y=\sum_{j\in A}\chb_j^Y$. Generalizing this, 
for any $\chb \in C_q(\mfkdl)$ the contraction of $\chb$ with regions $A_0,...,A_q \subset \RR^d$ is a derivation $\chb_{A_0\ldots A_q}\in\mfkDl$ defined by
\beq
\chb_{A_0\ldots A_q}^Y=  \sum_{j_k \in A_k, \, k = 0,...,q} \chb^Y_{j_0...j_q}.
\eeq
We may interpret $\chb_{A_0...A_q}$ as an ``integral'' of $\chb$ over $A_0,...,A_q$. Since chains are antisymmetric in $j_a$, without loss of generality we can assume that the regions are non-intersecting. 

Note that if all the regions $A_0,...,A_q$ are infinite, then in general the  derivation $\chb_{A_0\ldots A_q}$ is not summable. For the contraction to be a summable derivation (that is, an element of $\mfkdl$), one needs to choose the regions $A_0,\ldots,A_q$ with some care. For our purposes the following set of regions will suffice.
%For $d=0$, the lattice $\Lambda$ consists of a single point $j_0$, $\SAl=\SA_{j_0}$ is a matrix algebra, and $\mfkdl=\mfkd_{j_0}$ The complex of UL chains collapse to its degree-0 component which is simply $\mfkdl=\mfkDl$.
%For $d>0$ one can attach an element of $\mfkdl$ to a UL $d$-chain as follows. 
Let us pick a point $p\in\RR^d$ and a triangulation of $S^{d-1}$ as a boundary of a $d$-simplex. Let $\sigma_0,\ldots,\sigma_d$ be its open $(d-1)$-simplices and let $A_a$, $a=0,\ldots,d$, be an open subset of $\RR^d$ which in polar coordinates has the form $\RR_+\times \sigma_a$. We will say that $A_a$ is a conical region with base $\sigma_a$ and apex $p$. More generally, for a fixed $p$ and a fixed triangulation of $S^{d-1}$ into $d+1$ simplices we say that an open set $A_a$ is an eventually conical region with apex $p$ and base $\sigma_a$ if outside of a ball $B_p(r)$ it coincides with $\RR_+\times\sigma_a$. We say that an ordered partition $(A_0,...,A_d)$ of $\RR^d$ with $\Lambda$ being in the interior is a \textit{conical partition}, if $A_0,...,A_d$ are eventually conical regions with an apex $p$ and bases $\sigma_0,...,\sigma_d$ (see Fig. \ref{fig:simplpart}).
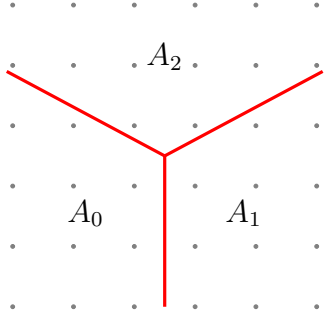
\begin{figure}
    \centering
\begin{tikzpicture}[scale=.4]
\draw[red, very thick] (0,0) -- (1.5*3.4641,1.4*2);
\draw[red, very thick] (0,0) -- (-1.5*3.4641,1.4*2);
\draw[red, very thick] (0,0) -- (0,-5);

\foreach \x in {-2.5,-1.5,...,2.5}{                           % Two indices running over each
    \foreach \y in {-2.5,-1.5,...,2.5}{                       % node on the grid we have drawn 
    \node[draw,gray,circle,inner sep=.5pt,fill] at (2*\x,2*\y) {}; % Places a dot at those points
    }
}

% \draw [gray, very thick] (4,0) arc [radius=4, start angle=0, end angle= 360];
\node  at (0,1.5*2+0.3) {$A_2$}; 
\node  at (-1.5*1.7321,-1.5*1-0.4) {$A_0$};
\node  at (1.5*1.7321,-1.5*1-0.4) {$A_1$};
% \node  at (3.5,3.5) {$D$};

\end{tikzpicture}
    \caption{An example of a conical partition of $\RR^2$.}
    \label{fig:simplpart}
\end{figure}
To any conical partition $(A_0,...,A_d)$ and a UL $d$-chain $\chb$ one can attach an element of $\mfkdl$:
% One can attach an element of $\mfkdal$ to a UL or UAL $d$-chain $\chF$ and an ordered set of $d+1$ non-intersecting eventually conical regions $A_0,\ldots,A_d$ in $\RR^d$ with apex $p$ and bases $\sigma_0,\ldots,\sigma_d$:
\beq\label{eq:cochaincontraction}
\chb_{A_0\ldots A_d}=\sum_{j_k\in A_k, \, k=0,\ldots, d} \chb_{j_0\ldots j_d}
\eeq
which is the contraction of $\chb$ with $A_0,\ldots,A_d$. This sum is finite. Note that the expression (\ref{eq:cochaincontraction}) does not depend on the ordering of the simplices $\sigma_a$ provided they correspond to a fixed orientation of $S^{d-1}$ and changes sign when the orientation is flipped. A version of Stokes' theorem  holds: $(\partial\chc)_{A_0\ldots A_d}=0$ for any $\chc \in C_{d+1}(\mfkdl)$.

\begin{remark} \label{rmk:cosheafcomplex}
With any conical partition $(A_0,...,A_d)$ we can also associate an integrated version of the complex eq. (\ref{eq:ULNC}) of {\Cech } type. Let $\Delta^k$ be the set of $(d-k-1)$-simplices of the triangulation of $S^{d-1}$ as a boundary of a $d$-simplex. For any $\sigma \in \Delta^k$ there is an associated set of cones $\{ A_{i_0},..., A_{i_{k}}\}$. We denote the image of the contraction of $C_k(\mfkdl)$ with $\{ A_{i_0},..., A_{i_{k}}\}$ by $\mfkDl^{(\sigma)}$. Each $\mfkDl^{(\sigma)}$ is a Lie subalgebra of $\mfkDl$. Then we have an exact sequence
\beq \label{eq:cohseafcomplex}
\mfkdl \xrightarrow[]{} \bigoplus_{\sigma \in \Delta^{d-1}} \mfkDl^{(\sigma)} \xrightarrow[]{} \bigoplus_{\sigma \in \Delta^{d-2}} \mfkDl^{(\sigma)} \xrightarrow[]{} ... \xrightarrow[]{}  \bigoplus_{\sigma \in \Delta^{0}} \mfkDl^{(\sigma)} \xrightarrow[]{} \mfkDl
\eeq
with the differential being a sum of signed injections.
\end{remark}

\subsection{The complex of rapidly decaying currents}

In the same way as $\mfkdl$ can be completed to $\mfkdal$ using a family of norms, the space $C_{q}(\mfkdl)$, $q\in\NN_0,$ of UL $q$-chains can be completed using the norms
% The space of UL $q$-chains can be topologized using a family of norms
% \beq
% \|\chF\|_{n}=\sup_{j_0,...,j_q \in \Lambda} \sup_{a = 0,1,...,q} \|\chF_{j_0...j_q}\|_{j_a,n}
% \eeq
\beq \label{eq:UALnorms}
\|\cha\|_{\alpha} := \sup_r (1+r)^{\alpha} f(\cha,r) = \sup_{a\in \{0,1,\ldots,q\}}\sup_{j_0,...,j_q \in \Lambda} \lVert\cha_{j_0...j_q}\rVert_{j_a,\alpha}, \quad\alpha\in\NN_0
\eeq
where
\beq
f(\cha,r):= \sup_{a\in\{0,1,...,q\}} \sup_{j_0,...,j_q \in \Lambda} f_{j_a}(\cha_{j_0...j_q},r)
\eeq
is defined for any $\cha \in C_q(\mfkdl)$. We call the completed space the space of uniformly almost local (UAL)  $q$-chains and denote it by $C_{q}(\mfkdal)$. This means that any element $\cha \in C_{q}(\mfkdal)$ can be represented by a sequence $\{\cha^{(n)}\},\, n \in \NN$ of UL $q$-chains $\cha^{(n)} \in C_q(\mfkdl)$ such that for any $\alpha\in\NN_0$  $\|\cha^{(n)}-\cha^{(m)}\|_{\alpha}$ can be made arbitrarily small by taking arbitrary sufficiently large $n,m$. 
% as the completion of $C_q(\mfkdl)$ with respect to these norms. We will denote this space $C_{q}(\mfkdal)$, $q\in\NN_0$.
% and regard $C_{\bullet}(\mfkdal)$ as a graded \Frechet\  space. 
\begin{lemma} \label{lma:Cdalequiv}
The following characterizations of UAL chains are all equivalent:
\begin{itemize}
    \item[(1)] A skew-symmetric function $\cha:\Lambda^{q+1}\ra\mfkdal$ defines an element of $C_{q}(\mfkdal)$ if $\|\cha\|_{\alpha}<\infty$ for any $\alpha \in \NN$.
    \item[(2)] A skew-symmetric function $\cha:\Lambda^{q+1}\ra\mfkdal$ defines an element of $C_{q}(\mfkdal)$ if there is a function $b(r) \in \Orf$  such that for any $j_0,...,j_q$ the observable $\cha_{j_0...j_q}$ is $b$-localized at $j_a$ for any $a \in \{0,1,...,q\}$.
    \item[(3)] $C_{q}(\mfkdal)$ is the completion of $C_q(\mfkdl)$ with respect to the norms $\|\cdot\|_{\alpha}$.
\end{itemize}
\end{lemma}
\begin{proof}
As in Lemma \ref{lma:Aal} the implication (3) $\Rightarrow$ (1) is straightforward.

If all the norms are finite, the function $f(\cha,r)$ can be upper-bounded by an element of $\Orfm$. Therefore (1) implies (2).

It is left to show (2) $\Rightarrow$ (3). Let $\{ \chA^{(n)} \}$ be a sequence of elements of $C_q(\mfkdl)$ such that  $\chA^{(n)}_{j_0...j_q}$ is a best possible approximation of $\chA_{j_0...j_q}$ by a traceless observable on $\mfkd_{B_{j_0}(n) \cap B_{j_1}(n) \cap ... \cap B_{j_q}(n)}$. Lemma \ref{lma:chainapprox} implies
\beq
\|\chA_{j_0\ldots j_q} - \chA^{(n)}_{j_0\ldots j_q}\| \leq 2 (2q+1) b(n).
\eeq
Therefore for any $a\in\{0,\ldots,q\}$ one has 
\beq
f_{j_a}\left(\chA_{j_0\ldots j_q} - \chA^{(n)}_{j_0\ldots j_q},r\right)\leq 2(2q+1){\rm min}(b(r),b(n))
\eeq
and therefore
\beq
\|\chA - \chA^{(n)}\|_{\alpha} \leq 2(2 q+1) \sup_{r\geq n} (1+r)^\alpha b(r).
\eeq
Thus the sequence $\{ \chA^{(n)} \}$  converges in the norms $\|\cdot\|_{\alpha}$ to $\chA$. Therefore (2) $\Rightarrow$ (3).
\end{proof}

% \footnote{An equivalent definition of a UAL $q$-chain involves a function $\chf:\Lambda^{q+1}\ra\mfkdal$ which is  skew-symmetric, concentrated near the diagonal of $\Lambda^{q+1}$ (in the sense that it decays faster than any negative power of $\diam(\{j_0,...,j_q\})$), and which can be approximated by a local observable on a ball of radius $r$ and center at any $j\in  \{j_0,...,j_q\}$ with an $\Or$ error which is uniform in $j$.}

% In the following if $\cha \in C_q(\mfkdal)$ satisfied the condition (2) from this Lemma for some $b(r) \in \Orf$, then we say that $\cha$ is $b$-localized. We also define 
% \beq
% f(\cha,r):= \sup_{a\in\{0,1,...,q\}} \sup_{j_0,...,j_q \in \Lambda} f_{j_a}(\cha_{j_0...j_q},r).
% \eeq
We use the family of norms $\|\cdot\|_{\alpha}$ to define a topology on $C_q(\mfkdal)$ as discussed in Appendix \ref{app:frechet}. Characterization (3) implies that this topology turns $C_q(\mfkdal)$ into a {\Frechet} space. Two other equivalent families of norms are described in Appendix \ref{app:normfamilies}.

Similarly, the space $\mfkDl$ can be completed using the norms
\beq
\|\derF\|^{br}_\alpha:=\sup_{Y\in \Br_d} (1+\diam(Y))^\alpha \|\derF^Y\|<\infty,\quad \forall  \alpha\in\NN_0.
\eeq
The resulting space is denoted by $\mfkDal$, and the resulting space of derivations of $\mfkdal$ is called the space of UAL derivations.

In Appendix \ref{app:complexcontinuity} we show that the boundary map $\p$, the contracting homotopy $h$, the bracket $\{\cdot,\cdot\}$, and the contraction maps on $C_{\bullet}(\mfkdl)$ can be extended to maps on $C_{\bullet}(\mfkdal)$ continuous in the {\Frechet} topology. Therefore we have the complex
\beq \label{eq:UALcomplex}
\ldots \stackrel{\partial_{2}}{\ra} C_{1}(\mfkdal)\stackrel{\partial_{1}}{\ra} C_0(\mfkdal) \stackrel{}{\ra} 0 
\eeq 
and the corresponding augmented exact complex
\beq \label{eq:UALcomplexAug}
\ldots \stackrel{\partial_{2}}{\ra} C_{1}(\mfkdal)\stackrel{\partial_{1}}{\ra} C_0(\mfkdal) \stackrel{\p_{0}}{\ra} \mfkDal \stackrel{}{\ra} 0
\eeq 
with the structure of a 1-shifted dg-{\Frechet}-Lie algebra. We call the latter the {\it uniformly almost local Noether complex} and  denote it $\No_\bullet$. It is graded by integers $q\geq -1$.

\begin{remark}
The space of UAL derivations $\mfkDal$ can be regarded as a subspace of the space of interactions satisfying the following ``gauge condition'': $\Phi_\Gamma\in\mfkd^Y$ if $\Gamma = Y \cap \Lambda$ for a brick $Y \in \Br_d$ and $\Phi_\Gamma=0$ otherwise.
\end{remark}

There is an injective Lie algebra homomorphism from $\mfkdal$ to the Lie algebra of UAL derivations which sends $\CB\in\mfkdal$ to the derivation $\CA\mapsto [\CB,\CA]$. One can describe this sub-space more intrinsically by making the following definition.
\begin{definition}
An element $\derF\in\mfkDal$ is called summable if the infinite sum $\sum_Y \derF^Y$ is absolutely  convergent in the \Frechet\ topology of $\mfkdal$.
\end{definition}
The image of $\mfkdal$ under the embedding into $\mfkDal$ consists precisely of summable elements of $\mfkDal$. Physically, summable UAL derivations correspond to interactions which are approximately localized at a point. Summable UAL derivations obviously form a Lie sub-algebra of $\mfkDal$. In fact, it is easy to see that they form an ideal. 

Finally, the contraction of any $\chb \in C_d(\mfkdal)$ with a conical partition $(A_0,...,A_d)$ of $\RR^d$ gives a summable element of $\mfkDal$, as   Prop.  {\ref{prop:contractionConical}} shows. Moreover the corresponding map $C_d(\mfkdal) \to \mfkdal$ is continuous. As in the Remark \ref{rmk:cosheafcomplex}, there is an almost local version of the integrated complex eq. (\ref{eq:cohseafcomplex}).

\subsection{Relation to energy and charge currents} \label{ssec:currents}

As mentioned in the introduction, currents on a lattice can be defined using the language of chains. Consider a lattice system with finite-range interactions. The dynamics for such a system is described by a Hamiltonian that can be regarded (after multiplication by $i$) as a UL derivation $\derH$. The Hamiltonian density can be defined as a UL 0-chain $\chh$ such that $\derH = \p \chh$. Obviously, for a fixed $\derH$ the 0-chain $\chh$ is far from unique. The ambiguity can be fully characterized, since by Theorem \ref{thm:ULH0} any two choices of $\chh$ differ by a boundary $\p$ of a 1-chain. Once $\chh$ is fixed, we can define an energy current $\chJ^E:\Lambda\times\Lambda\ra \SAl$ as a solution of the equation 
\beq\label{eq:conservation1}
\sum_k [\chh_k,\chh_j]=-\sum_k \chJ^E_{kj}.
\eeq
The observable $\chJ^E_{kj}$ represents the energy flow from site $j$ to site $k$. It is natural to require $\chJ^E_{jk}$ to be traceless, so that there is no energy flow between sites in the infinite-temperature state. Since $\chH$ is finite-range, it is also natural to require $\chJ^E_{jk}$ to vanish whenever $j$ and $k$ are sufficiently far apart, and to be localized near $j$ and $k$. Thus $\chJ^E$ is a UL 1-chain. The equation (\ref{eq:conservation1}) can be written using the algebraic operations on the UL Noether complex:
\beq\label{eq:conservation2}
\{\derH,\chh\}=-\partial \chJ^E.
\eeq
Note that this equation for $\chJ^E$ is guaranteed to have a solution because by the properties of the shifted Lie bracket the l.h.s. is closed, $\partial\{\derH,\chh\}=\{\derH,\derH\}=0,$ and thus exact.
Similarly, a Hamiltonian $\derH$ with rapidly decaying interactions can be regarded as a self-adjoint UAL 0-chain $\chh$ such that $\derH=\p \chh$, while an energy current is defined to be a self-adjoint UAL 1-chain $\chJ^E$ solving the equation (\ref{eq:conservation2}).
Both in the UL and UAL cases, the equation has an obvious solution:
\beq
\chJ^E_{kj}=-[\chH_k,\chH_j],
\eeq
which can be written using the operations on the Noether complexes as
\beq\label{eq:JEsolution}
\chJ^E=-\frac12\{\chH,\chH\}.
\eeq
Triviality of $H_1(\mfkdl)$ and $H_1(\mfkdal)$ ensures that any other solution differs from (\ref{eq:JEsolution}) by an exact 1-chain. 

Similarly, a continuous one-parameter symmetry of a lattice system is encoded into a charge density which, after multiplication by $i$, can be viewed as a 0-chain $\chQ$ (usually assumed to be uniformly local). The Hamiltonian $\derH$ is said to be $\chQ$-invariant if the derivations corresponding to $\chQ$ and $\derH$ commute, $\{\partial\chq,\derH\}=0$. If the symmetry group is compact, using the average over the group action we can always make sure that the Hamiltonian density $\chh$ satisfies  $\{\partial\chq,\chh\}=0$ . A current for the symmetry generated by $\chq$ is a UL or UAL 1-chain $\chj$ solving the equation
\beq
\sum_k [\chh_k,\chq_j]=-\sum_k \chj_{kj}.
\eeq
This equation can also be written using index-free notation:
\beq
\{\derH,\chq\}=-\partial\chJ.
\eeq
A solution always exists because the l.h.s. is closed. Any two solutions differ by an exact 1-chain. If $\{\partial\chq,\chh\}=0$, one can write an explicit solution 
\beq
\chJ=-\{\chH,\chQ\}.
\eeq
% If two Hamiltonian densities map to the same derivation, then they are often regarded as physically equivalent. However, one should keep in mind that they do not lead to physically equivalent energy currents. Indeed, if $\partial\chF'=\partial\chF$, then $\chF'=\chF-\partial\chA$ for some 1-chain $\chA$, and therefore replacing $\chF$ with $\chF'$ shifts $\chJ^E$ by  $\chA$. This affects the net current through a closed codimension-1 surface, which is a measurable quantity. Despite this, one expects that certain physical quantities, such as transport coefficients in linear response theory, are not affected by this ``gauge freedom''. Transformation properties of Kubo formulas under such  re-definitions of the Hamiltonian density have been analyzed in \cite{gauge,KapSpo}.
From the above discussion it is clear that all densities and currents have ambiguities. However they can be fully characterized, and one expects that physical quantities, such as transport coefficients in linear response theory, are not affected by this ``gauge freedom''. Transformation properties of Kubo formulas under such  re-definitions of the Hamiltonian density have been analyzed in \cite{gauge,KapSpo}.

\subsection{Locally generated automorphisms} \label{ssec:automorphisms}

Derivations in $\mfkDal$ and continuous one-parameter families of derivations can be integrated to automorphisms of the algebra $\SA$. Let $C([0,1],\mfkDal)$ be the \Frechet\ space of $\mfkDal$-valued functions on the interval $[0,1]$ (see Appendix A for a brief discussion of  functions valued in \Frechet\ spaces). In Appendix \ref{app:LRbound} we show that for any $\derG \in C([0,1],\mfkDal)$ there is a unique one-parameter family of automorphisms $\alpha_{\derG}: [0,1] \to \Aut(\SAal)$ such that for any $\CA\in\SAal$ the curve $s\mapsto\alpha_\derG(s)(\CA)$ is continuously differentiable and solves the differential  equation
\beq
\frac{d \alpha_{\derG}(s) (\CA)}{ds} = \alpha_{\derG}(s)(\derG(s)(\CA))
\eeq
with the initial condition $\alpha_\derG(0)={\rm id}$.
We call such one-parameter families {\it locally generated paths} (LGPs). One may also regard $\derG(s)$ as a component of a continuous $\mfkDal$-valued 1-form $\derG(s) ds$ on $[0,1]$. In what follow we will not distinguish between this 1-form and the function $\derG$ and for any continuous $\mfkDal$-valued 1-form $\derF$ denote by $\alpha_{\derF}$ a unique one-parameter family of automorphisms $\alpha_{\derF}: [0,1] \to \Aut(\SAal)$ defined by
\beq
d \alpha_{\derF}(s) (\CA) = \alpha_{\derF}(s)(\derF(s)(\CA)).
\eeq

The map $\derG\mapsto \alpha_\derG$ from continuous 1-parameter families of UAL derivations to LGPs is clearly 1-1. This allows us to identify the set of LGPs with the \Frechet\ space $C([0,1],\mfkDal)$ and thus make the former into a \Frechet\ manifold. The set of LGPs also has a group structure. The composition $\alpha_{\derF} \circ \alpha_{\derG}$ of two LGPs is an LGP generated by $\derG(s) + (\alpha_{\derG}(s))^{-1} (\derF(s)) $. The inverse $\alpha^{-1}_{\derF}$ is an LGP generated by $- (\alpha_{\derF}(s))^{-1} (\derF(s))$. By Prop. \ref{prop:GAsmoothness}, both the composition and the inverse are smooth maps of \Frechet\ spaces and thus the set of LGPs is  a \Frechet-Lie group. Its Lie algebra is $C([0,1],\mfkDal)$ with the Lie bracket
\beq
\{\derF,\derG\}(s)=\int_0^s \left(\{\derF(u),\derG(s)\}-\{\derG(u),\derF(s)\}\right) du.
\eeq

Restricting every LGP $\alpha_\derG(s)$ to the endpoint $s=1$ we get a homomorphism from the group of LGPs to ${\rm Aut}(\SAal)$. Such automorphisms of $\SAal$ will be called {\it locally generated automorphisms} (LGAs). In this paper we do not define any topology on the group of LGAs.

\begin{remark}
It is plausible that the group of LGAs is a \Frechet-Lie group integrating the Lie algebra of UAL derivations $\mfkDal$. The group of LGPs is supposed to be the group of based continuous paths in the group of LGAs. 
\end{remark}
The action of LGAs and LGPs on observables can be extended to an action on UAL derivations and chains in a straightforward way:
\begin{align}
\left(\alpha_\derG(s)(\chA)\right)_{j_0\ldots j_q}=\alpha_\derG(s)\left(\chA_{j_0\ldots j_q}\right),\quad \chA\in C_q(\mfkdal),\\
\left(\alpha_\derG(s)(\derA)\right)^Y=\sum_{Z\in\Br_d}\left(\alpha_\derG(s)\left(\derA^Z\right)\right)^Y,\quad \derA\in \mfkDal .
\end{align}
By Proposition \ref{prop:GAsmoothness}, this action is jointly continuous and smooth.

Let $\CM$ be a finite-dimensional manifold. We say that a family of LGPs $\beta_{\m} = \alpha_{\derG_{\m}}$, $\m \in \CM$, is smooth, if the corresponding map $\CM\ra C([0,1],\mfkDal)$ is smooth (smooth here means that derivatives of all orders exist, see Appendix A for a further discussion). This is equivalent to saying that $\derG(s,\m)$ is jointly continuous in $s$ and $\m$ and infinitely differentiable in $\m$. As explained in Remark \ref{rmk:alphaMderivative}, to any such family of LGPs one can assign a smooth 1-form  $\omega_\beta\in\Omega^1(\CM,\mfkDal)$ satisfying $d\left(\beta(1)(\CA)\right)=\beta(1)\left(\omega_\beta(\CA)\right)$ for any $\CA\in\mfkdal$. The 1-form $\omega_\beta$ is flat, i.e.  $d\omega_\beta+\frac12 \{ \omega_\beta, \omega_\beta \}=0$.

\section{Invariants of families of gapped states}\label{sec:descendants}

\subsection{Complexes associated to gapped states}  \label{ssec:pseudogapped}

The expectation value of an observable $\CA \in \SA$ in a state $\psi$ on $\SA$ will be denoted $\lal \CA \ral_{\psi}$. We say that two pure states $\psi_1$ and $\psi_2$ are in the same phase or LGA-equivalent if there is an LGA $\alpha$ such that $\psi_2 = \psi_1 \circ \alpha$. The trivial phase is defined to be  the LGA-equivalence class of a factorized pure state $\omega_0 = \underset{\Gamma \subset \Lambda} \varinjlim\, \bigotimes_{j \in \Gamma} \omega_j$, where $\{ \omega_j \}$ is a collection of pure states on the algebras $\{\SA_j\}$.

We say that a UAL derivation $\derF$ does not excite the state $\psi$ if for any $\CA \in \mfkdal$ one has $\lal \derF(\CA) \ral_{\psi} = 0$. Such derivations form a Lie sub-algebra in $\mfkDal$ which we denote  $\mfkDpal$. Similarly, the Lie sub-algebra $\mfkdpal \subset \mfkdal$ consists of $\CB\in\mfkdal$ such that $\lal [\CB,\CA] \ral_{\psi} = 0$ for all $\CA\in\mfkdal$. We also define a subcomplex $C_{\bullet}(\mfkdpal)\hookrightarrow C_{\bullet}(\mfkdal)$ of chains that do not excite $\psi$ and the complex
\beq \label{eq:HcomplexMain}
\ldots \stackrel{\partial_{2}}{\ra} C_{1}(\mfkdpal)\stackrel{\partial_{1}}{\ra} C_0(\mfkdpal)\stackrel{\p_0}{\ra} \mfkDpal \stackrel{}{\ra} 0
\eeq 
of chains and derivations that do not excite $\psi$. The latter is a sub-complex of the UAL Noether complex $\No_\bullet$. We will denote it $\No^\psi_\bullet$. The space of $k$-cycles of $\No^\psi$ will be denoted $Z_k(\No^\psi)$. As usual, the $k$-th homology group of the complex is the quotient of the space of $k$-cycles by the subspace of $k$-boundaries. 
%\beq \label{eq:complexMain}
%\ldots \stackrel{\partial_{2}}{\ra} C_{1}(\mfkdal)\stackrel{\partial_{1}}{\ra} C_0(\mfkdal)\stackrel{\p_0}{\ra} \mfkDal \stackrel{}{\ra} 0.
%\eeq 

In general, the homology of a complex of \Frechet\ spaces may be rather pathological (for example, if we endow it with a quotient topology, it may not be a Hausdorff space). But the situation is much simplified if $\psi$ is a {\it gapped ground state} of a UAL derivation. 

Recall that a pure state $\psi$ is said to be a ground state of a derivation $\derH \in \mfkDal$ if for any $\CA \in \SAal$ one has $-i\lal \CA^* \derH(\CA) \ral_{\psi} \geq 0$. Any such state is necessarily invariant under the 1-parameter group of automorphisms generated by $\derH$ \cite{bratteli2012operator2}. 
\begin{definition}
A pure state $\psi$ is said to be a gapped ground state of $\derH\in\mfkDal$ with a gap greater or equal than $\Delta>0$ if $-i\lal \CA^* \derH(\CA) \ral_{\psi} \geq \Delta \left(\lal \CA^*\CA\ral_\psi-|\lal\CA\ral_\psi|^2\right)$ for any $\CA\in\SAal$. 
\end{definition} 
If $\psi$ satisfies the above condition for some unspecified choice of $\derH\in\mfkDal$ and $\Delta>0$, we will say that $\psi$ is gapped.
We will also say that a derivation $\derH\in\mfkDal$ is gapped if there exists a state $\psi$ which is a gapped ground state for $\derH$ (such a state need not be unique).
\begin{remark}
If $\psi$ is a gapped ground state of $\derH\in\mfkDal$, then $\psi$ is the only vector state in the GNS representation of $\psi$ which is the ground state of $\derH$. Further, if $\hat H$ is the generator of the 1-parameter group of automorphisms $\alpha_\derH(s)$ in the GNS representation of $\psi$, then $\hat H$ is a positive operator annihilating the vacuum vector, and its spectrum on the orthogonal complement of the vacuum vector is contained in $[\Delta,+\infty)$.
\end{remark}

To analyze the homology of complex $\No^\psi$ we will make use of certain linear maps defined by means of  integral transforms.
For any $\derH \in \mfkDal$ and a piecewise-continuous function $f:\RR \to \RR$ satisfying $f(t)=\CO(|t|^{-\infty})$ we define a map $\mathscr{I}_{\derH,f}:\SAal\ra\SAal$ by
\beq
\mathscr{I}_{\derH,f}(\cdot) := \int_{-\infty}^{+\infty} f(t) \alpha_{ \derH}(t) (\cdot) d t.
\eeq
It is shown in Appendix \ref{app:gappedHamiltonians} that $\mathscr{I}_{\derH,f}$ is well-defined and continuous. More precisely, for any $a$-localized $\CA\in\mfkdal$ we have 
\begin{equation}
\lVert \mathscr{I}_{\derH,f}(\CA) \rVert_{j,\alpha} \leq C_{\alpha} \lVert \CA\rVert_{j,\alpha}, \quad \alpha \in \NN_0
\end{equation}
where $C_\alpha>0$ depends on $\derH,f$ and $a$. This estimate implies that $\scrI_{\derH,f}$ is continuous but is stronger than continuity because $C_{\alpha}$ does not depend on $j$. In other words, the map $\scrI_{\derH,f}$ is equicontinuous w. r. to a family of  metrics on $\mfkdal$ labeled by $j$. This ensures that $\scrI_{\derH,f}$ extend to continuous chain maps
$\scrI_{\derH,f}:\No_\bullet\ra \No_\bullet$. On $k$-chains with $k\ge 0$ it is defined by 
\beq
\scrI_{\derH,f}(\chA)_{j_0\ldots j_k}=\scrI_{\derH,f}(\chA_{j_0\ldots  j_k}),
\eeq
while on derivations it is defined by
\beq
\scrI_{\derH,f}(\derA)^Y=\sum_Z \left(\scrI_{\derH,f}(\derA^Z)\right)^Y .
\eeq

If $\alpha_\derH(t)$ preserves a state $\psi$, then $\mathscr{I}_{\derH,f}$ preserves the subspace $\mfkdpal$ and the subcomplex $\No^\psi$. Indeed,
for any $\CA,\CB\in\mfkdal$ we have 
\begin{equation}
\langle [\mathscr{I}_{\derH,f}(\CA),\CB]\rangle_\psi=\int_{-\infty}^{+\infty} f(t)\langle[\CA, \alpha_{\derH}(-t)(\CB)]\rangle_\psi dt.
\end{equation}
Therefore if $\CA\in\mfkdpal$, then $\mathscr{I}_{\derH,f}(\CA)\in\mfkdpal.$

If $\psi$ is a gapped ground state of $\derH$, then with a clever choice of $f$ a stronger result holds. 
\begin{lemma}\label{lma:wdelta}
Let $\psi$ be a gapped ground state of $\derH \in \mfkDal$ with a gap greater or equal than $\Delta > 0$. Let $w_{\Delta}(t)$ be an even continuous function $w_{\Delta}(t) = \CO(|t|^{-\infty})$ satisfying $\int w_{\Delta}(t) e^{- i \omega t} d t = 0$ for $|\omega| > \Delta'$ for some $0<\Delta'<\Delta$ and $\int w_{\Delta}(t) d t = 1$. \footnote{Such functions exist \cite{hastings2010quasi,bachmann2012automorphic}. See Lemma 2.3 from \cite{bachmann2012automorphic} for an explicit example.}  Then for any $\CA \in \mfkdal$ and $\CB \in \mfkdal$ we have
\beq\label{Iw}
\lal \mathscr{I}_{\derH,w_{\Delta}}(\CA) \CB \ral_{\psi} = \lal \mathscr{I}_{\derH,w_{\Delta}}(\CA) \ral_{\psi} \lal \CB \ral_{\psi} ,
\eeq
and thus $\mathscr{I}_{\derH,w_{\Delta}}(\CA) \in \mfkdpal$ and $\scrI_{\derH,w_{\Delta}}(\No_\bullet)\subseteq \No^\psi_{\bullet}$.
\end{lemma}
\begin{proof}
Let $\pi$ be the GNS representation corresponding to $\psi$ and $dP_\omega$, $\omega\in\RR$, be the projection-valued measure on $\RR$ corresponding to the self-adjoint operator $\hat H$. Then 
\begin{multline}
\lal \mathscr{I}_{\derH,w_{\Delta}}(\CA) \CB \ral_{\psi} =
\int_{-\infty}^{+\infty} d\omega\int_{-\infty}^{+\infty} dt\, w_{\Delta}(t) e^{-i\omega t} \langle 0\vert \pi(\CA) dP_\omega \pi(\CB)\vert 0\rangle=\\
=\int_{-\Delta'}^{+\Delta'} d\omega\int_{-\infty}^{+\infty} dt \, w_{\Delta}(t) e^{-i\omega t} \langle 0\vert \pi(\CA) dP_\omega \pi(\CB)\vert 0\rangle=\\
=\int_{-\infty}^{+\infty} dt \, w_{\Delta}(t)\langle\CA\rangle_{\psi} \langle \CB\rangle_{\psi} = \lal \mathscr{I}_{\derH,w_{\Delta}}(\CA) \ral_{\psi} \lal \CB \ral_{\psi} .
\end{multline}
\end{proof}
Using this lemma, one easily obtains the following result.
\begin{lemma} \label{lma:Htilde}
For any $\derH \in \mfkDal$ with a gapped ground state $\psi$ there exists $\chh^\psi \in C_0(\mfkdpal)$ such that $\derH = \p \chh^\psi$.
\end{lemma}
\begin{proof}
Let $w_\Delta=\CO(|t|^{-\infty})$ be a function as in the statement of Lemma \ref{lma:wdelta}, and suppose $\derH=\partial\chh$ for some $\chh \in C_0(\mfkdal)$. Let $\chh^\psi = \mathscr{I}_{\derH,w_{\Delta}} (\chh)$.
Then for any $\CA\in\SAal$ we have 
\begin{multline}
\partial\chh^\psi(\CA)=\int_{-\infty}^{+\infty} w_{\Delta}(t) \alpha_\derH(t) \derH(\alpha_\derH(-t)(\CA)) dt = \\ = \int_{-\infty}^{+\infty} w_{\Delta}(t) \derH(\CA) dt=\derH(\CA).
\end{multline}
\end{proof}
\begin{remark}
The construction of $\chh^\psi$ by means of the map $\mathscr{I}_{\derH,w_\Delta}$ is due to A. Kitaev \cite{kitaev2006anyons}.
\end{remark}

Another interesting choice for $f$ is described in the next lemma.
\begin{lemma}\label{lma:Wdelta}
Let $\psi$ be a gapped ground state of $\derH\in\mfkDal$. Let $W_{\Delta}(t) = \CO(|t|^{-\infty})$ be an odd piecewise-continuous function defined for $t>0$ by $W_{\Delta}(|t|)=-\int_{|t|}^{\infty} w_{\Delta}(s) d s$. Then for any $\CA\in\mfkdal$ we have \begin{equation}
\CA-\scrI_{\derH,W_\Delta}\left(\derH(\CA)\right)=\scrI_{\derH,w_\Delta}(\CA).
\end{equation}
\end{lemma}
\begin{proof}
Straightforward computation.
\end{proof}
\begin{remark}
The map $\mathscr{I}_{\derH,W_\Delta}$ first appeared in  \cite{osborne2007simulating}.
\end{remark}

In what follows we will use a shorthand $t^\psi$ for the chain map $\scrI_{\derH,w_{\Delta}}:\No_\bullet\ra\No_\bullet$ with $w_\Delta$ chosen as in Lemma \ref{lma:wdelta} and a shorthand $\CI^\psi$ for the chain map $\scrI_{\derH,W_{\Delta}}:\No_\bullet\ra\No_\bullet$ with $W_\Delta$ chosen as in Lemma \ref{lma:Wdelta}. These chain maps preserve the subcomplex $\No^\psi_\bullet$. In addition $t^\psi$ maps $\No_\bullet$ to $\No^\psi_\bullet$. Finally, they satisfy an identity
\beq\label{hCITidentity}
\chA-\CI^\psi\left(\{ \derH,\chA\}\right)=t^\psi(\chA),\quad \forall \chA\in \No_\bullet.
\eeq

We are now ready to prove
\begin{theorem}\label{thm:poincareH}
Let $\psi$ be a gapped state. Then the homology $H_{\bullet}(\No^\psi)$ is trivial. Moreover, there is a continuous map $h^{\psi}_k: Z_k(\No^\psi)\ra \No^{\psi}_{k+1}$ such that $\partial_{k+1}\circ h^\psi_k= \Id$.
\end{theorem} 
\begin{proof}
Let $\derH$ be a gapped Hamiltonian for $\psi$. By Lemma \ref{lma:Htilde}, there exists  $\chh^\psi\in C_0(\mfkdpal)$ such that $\partial \chh^\psi=\derH$. For any $\chF\in \No_\bullet$ let
\beq
s^\psi(\chF)=\CI^\psi\left(\{\chH^\psi,\chf\}\right).
\eeq
$s^\psi$ is a continuous linear map $\No_k\ra \No_{k+1}$ which maps $\No^\psi_k$ to  $\No^\psi_{k+1}$. Now suppose  $\chA\in Z_k(\No^\psi)$. By Prop. \ref{prop:contrHoContinuity}  $\chA=\partial \chB$, where $\chB=h(\chA)\in \No_{k+1}$ is a continuous linear function of $\chA$. 
% The identity (\ref{hCITidentity}) implies
% \beq
% \chB=t^\psi(\chB)+\partial(s^\psi(\chB))+s^\psi(\chA),
% \eeq
% and thus
% \beq
% \chA=\partial\chB=\partial(t^\psi\circ h(\chA)+s^\psi(\chA)). 
% \eeq
The identity (\ref{hCITidentity}) implies 
\beq
\cha = \p t^{\psi}(\chb) + \p s^{\psi}(\cha) = \partial(t^\psi\circ h(\chA)+s^\psi(\chA))
\eeq
Therefore we can set $h^\psi=t^\psi\circ h+s^\psi$.
\end{proof}
\begin{remark}
The vanishing of $H_\bullet(\No^\psi)$ can also be explained as follows. Eq. (\ref{hCITidentity}) implies that
$s^\psi$ induces a contracting homotopy on the quotient complex $\No_{\bullet}/\No_{\bullet}^\psi$, therefore the homology of $\No_{\bullet}/\No_{\bullet}^\psi$ is trivial. The homology of the UAL Noether complex is also trivial. Therefore the long exact sequence of homology groups (in the category of vector spaces) corresponding to the short exact sequence
\beq 
0 \to \No_{\bullet}^\psi \to \No_{\bullet} \to \No_{\bullet}/\No_{\bullet}^\psi \to 0
\eeq
implies that $H_{\bullet}(\No^\psi)$ is trivial. However, for our purposes it is important to know that a continuous linear map $h^\psi$ inverting $\partial$ exists.
\end{remark}
%\begin{theorem} \label{thm:poincareH}
%Let $\psi$ be a gapped ground state of $\derH \in \mfkDal$. Then the complex eq. (\ref{eq:HcomplexMain}) is exact. Moreover, 
%\end{theorem}
%\begin{proof}
%By Lemma \ref{lma:Htilde} we can represent $\derH = \p \chh^\psi$ where $\chh^\psi \in C_0(\mfkdpal)$ is defined by $\chh^\psi = \mathscr{I}_{\derH,w_{\Delta}} (\chh)$.  By Lemma \ref{lma:tildemap}, the map $\mathscr{I}_{\derH,W_{\Delta}}:\mfkdal \to \mfkdal$ is well-defined and continuous. For any $\CA\in\mfkdal$ one has $\CA-\mathscr{I}_{\derH,W_\Delta}(\derH (\CA))=\mathscr{I}_{\derH,w_\Delta}(\CA)\in\mfkdpal$. 

%Let $\cha \in C_q(\mfkdpal)$ such that $\p \cha = 0$. Let $\chb' = h_{q}(\cha)$, and let $\chb = \chb' - \p \{\chh, \mathscr{I}_{\derH,W_\Delta}(\chb') \}$. Then $\chb \in C_{q+1}(\mfkdpal)$ and $\cha = \p \chb$.

% \begin{remark}
% The existence of $\CI^\psi$ roughly means that any restriction of a state preserving Hamiltonian to a region (that in general does not preserve the state) can be modified at the boundary of the region, so that the resulting Hamiltonian still preserves the state. 
% \end{remark}

In this section we will define invariants of gapped states and smooth families of gapped states under LGA-equivalence. The definition depends exclusively on the existence of a derivation $\derH\in\mfkDpal$ and a map $\CI^\psi$ with the properties described above\footnote{The existence of $\CI^\psi$ roughly means that any restriction of an arbitrary state-preserving Hamiltonian to a region (that in general does not preserve the state) can be modified at the boundary of the region, so that the resulting Hamiltonian still preserves the state. A similar property has been considered in \cite{bachmann2019many}}. 
%Any such invariant also gives an invariant of gapped finite-range Hamiltonians and their families. 

% \begin{remark}
% Though the existence of UAL contracting homotopy may seem to be a non-local condition because of the component $h^{\psi}_{-1}:\mfkDpal \to C_{0}(\mfkdpal)$, it is in fact equivalent to the following explicitly local one. Let $s^{\psi}: C_{n}(\mfkdal) \to C_{n+1}(\mfkdal)$ be a UAL map for $n \geq 0$, such that for any $\chf \in C_{n}(\mfkdal)$ with $\p \chf \in C_{n-1}(\mfkdpal)$ we have $(\chf - \p \circ s^{\psi}(\chf)) \in C_{n}(\mfkdpal)$. Indeed, given such a map and a UAL contracting homotopy of the complex eq. (\ref{eq:complexMain}) we can define $h^{\psi} = h - \p \circ s^{\psi} \circ h$. Conversely, given a UAL contracting homotopy $h^{\psi}$ we can use $s^{\psi} = h \circ (\text{id}-h^{\psi} \circ \p)$.
% \end{remark}

\subsection{Smooth families of states}  \label{ssec:smoothfamiliesofstates}

Let $\CM$ be a compact connected smooth manifold. Consider a family of states $\psi$ parameterized by $\CM$. We denote a state at the point $m \in \CM$ by $\psi_m:\SA \to \mathbb{C}$. The averaging over states defines a function $\lal\, \cdot\, \ral_{\psi}: \CM \to \mathbb{C}$.
% $\psi:\CM \times \SA\ra{\mathbb C}$ parameterized by $\m \in \CM$. We will use a shorthand $\langle\,\cdot\,\rangle_\m=\langle\,\cdot\,\rangle_{\psi_\m}.$ 
\begin{definition}\label{def:smoothfamily} A smooth family of states $\psi$ over $\CM$ is a family of states for which  there exists  $\derG \in \Omega^1(\CM,\mfkDal)$ such that for any smooth path $\gamma: [0,1] \to \CM$ one has $\psi_{\gamma(s)}=\psi_{\gamma(0)}\circ  \alpha_{\gamma^* \derG}(s)$. 
\end{definition}
One motivation for this definition is a theorem of A. Moon and Y. Ogata \cite{moon2020automorphic}. As explained in Appendix \ref{app:gappedHamiltonians}, it implies  that for a smooth family of gapped UL  Hamiltonians $\derH \in \Omega^0(\CM, \mfkDal)$ with a unique gapped ground state $\psi_\m$ $\forall \m\in\CM$ and such that the function $\CM \mapsto \langle\CA\rangle_{\psi}$ is smooth for all $\CA\in\SAal$, the family $\psi$ is smooth in the sense of Def. \ref{def:smoothfamily}.
\begin{prop}
A family of states $\psi$ is smooth if and only if for any observable $\CA\in\SAal$ the function $\CM \mapsto \langle\CA\rangle_{\psi}$ is smooth and satisfies
\beq\label{parallelfamily}
d\langle\CA\rangle_{\psi} = \langle\derG(\CA)\rangle_{\psi}. 
\eeq
\end{prop}
That is, a family of states $\psi$ is smooth iff it is parallel with respect to the connection $d+\derG$ on the trivial vector bundle with fiber $\SAal$. 
\begin{proof}
Pick a point $\m_0\in \CM$. For any $\m\in\CM$ there is a smooth path $\gamma:[0,1]\ra M$ with $\gamma(0)=\m_0$, $\gamma(1)=\m$, and by assumption $\langle\CA\rangle_{\psi_{\gamma(s)}}=\langle \alpha_{\gamma^*\derG}(s)(\CA)\rangle_{\psi_{\m_0}}$. By Prop. \ref{lma:alphaODE} this function is smooth and its derivative at $s=1$ is $\langle \gamma^*(1)\derG(\CA)\rangle_{\psi_{\m}}$. This implies (\ref{parallelfamily}). Conversely, eq. (\ref{parallelfamily}) implies that for any such path $\gamma$ and any $\CA\in\SAal$ the function $s\mapsto\langle \alpha^{-1}_{\gamma^*(s)\derG}(\CA)\rangle_{\psi_{\gamma(s)}}$ is constant. Therefore $\psi_{\gamma(s)}=\psi_{\gamma(0)}\circ  \alpha_{\gamma^* \derG}(s)$.
\end{proof}
\begin{corollary}
Let $\{\CU^{(\sfa)}\}$ be an open cover of $\CM$. A family of states $\psi$ over $\CM$ is smooth if and only if it is smooth over each element of the cover. 
\end{corollary}
\begin{proof}
The only if direction is obvious. Now suppose one is given a $\mfkDal$-valued 1-form $\derG^{(\sfa)}$ on each $\CU^{(\sfa)}$ such that $\psi$ is parallel with respect to each of them. Then we can construct $\derG$ on $\CM$ with respect to which $\psi$ is parallel using a partition of unity for some open cover subordinate to $\{\CU^{(\sfa)}\}$.
\end{proof}
If every element of the cover $\{\CU^{(\sfa)}\}$ is smoothly contractible, one can also describe a smooth family of states using locally-defined smooth families of LGPs. 
\begin{prop}
Let $\{\CU^{(\sfa)}\}$ be a finite cover such that each element is smoothly contractible. A family $\psi_\m,\m\in\CM,$ is smooth if and only if there is a state $\psi_0$ and smooth families $\beta^{(\sfa)}_{\m}$, $\m \in \CU^{(\sfa)}$, of LGPs such that  $\psi_\m = \psi_0 \circ \beta^{(\sfa)}_\m(1)$ for any $\m\in\CU^{(\sfa)}$ and any $\sfa$.
\end{prop}
\begin{proof}
Suppose $\psi_\m$ is parallel with respect to $\derG \in \Omega^1(\CM,\mfkDal)$. We pick a point $\m_0 \in \CM$ and let $\psi_0=\psi_{\m_0}$. For every $\sfa$ let   $\gamma^{(\sfa)}:\CU^{\sfa}\times [0,1] \to \CM$ be a smooth homotopy between the map of $\CU^{(\sfa)}$ to $\m_0$ and the identity map. We regard it as a smooth family of paths in $\CM$ based at $\m_0$ and labeled by $\m$ and let  $\beta^{(\sfa)}_{\m}:= \alpha_{\gamma^{(\sfa)*}_{\m}\derG}$. 

In the opposite direction, given families of LGPs $\beta^{(\sfa)}$ we can define $\derG = \sum_{\sfa} f^{(\sfa)} \omega^{(\sfa)}$, where $\omega^{(\sfa)}:= \omega_{\beta^{(\sfa)}}$ and $f^{(\sfa)}$ is a partition of unity. Since the family of states $\psi$ is parallel with respect to the connection $d+\omega^{(\sfa)}$ on each $\CU^{(\sfa)}$, it is also parallel with respect to $d+\derG$.
\end{proof}

Given a smooth family of states $\psi_\m,\m\in\CM$, one may ask whether it is possible to choose a globally defined smooth family of LGPs $\beta_\m$, $\m\in\CM$, such that $\psi_\m=\psi_0\circ\beta_\m$. Since the space of LGPs is smoothly contractible, this is possible only if the family of states is smoothly homotopic to a constant family. For $d=0$, when the algebra $\SAal$ is simply the algebra of operators on a finite-dimensional Hilbert space $\CH$, it is well known that not all families of states are homotopic to a constant family. Indeed, for $d=0$ a smooth family of states is the same as a smooth line bundle $\CL$ over $\CM$. If the 1st Chern class of this line bundle is  non-trivial, the family is not homotopic to a constant one. If $\psi_\m$ is a ground state of a family of gapped Hamiltonians parameterized by $\CM$, it is the cohomology class of the Berry curvature which provides an obstruction for the existence of a globally defined family of LGPs. Below we will construct similar obstructions (``higher Berry classes'') for smooth families of gapped states with $d>0$. This includes smooth families of ground states of gapped finite-range Hamiltonians.

% By choosing any gapped Hamiltonian $\derH_0$ for $\psi_{\lambda_0}$ we can define a smooth family of Hamiltonians over $\CM$ by $\derH = \sum_{\sfa} f^{(\sfa)} \alpha^{(\sfa)}(\derH_0)$.

% In this section we define an invariant of a family of states using the second definition in terms of globally defined $\derG \in \Omega^1(\CM,\mfkDal)$.  

To construct these obstructions we will use a bi-complex of differential forms taking values in the complex $\No=\left(C_{\bullet}(\mfkdal)\partialarrow\mfkDal\right)$:
\beq \label{eq:OCcomplex}
\ldots \stackrel{\partial}{\ra} \Omega^{\bullet}(\CM, C_{1}(\mfkdal))\stackrel{\partial}{\ra} \Omega^{\bullet}(\CM,C_0(\mfkdal))\stackrel{\partial}{\ra} \Omega^{\bullet}(\CM,\mfkDal) \stackrel{\partial}{\ra} 0,
\eeq
with the second differential being the de Rham differential $d$. In this section it will be convenient to use a cohomological rather than homological grading on $\No$ and accordingly denote its degree $k$-component by $\No^k$. Then $d$ has bi-degree $(1,0)$, while $\partial$ has bi-degree $(0,1)$. We will also shift the grading on $\No$ so that $\mfkDal$ sits in degree $0$ and the bracket has degree $0$. Then $\No^\bullet$ becomes a non-positively graded DGFLA. 

The bi-complex (\ref{eq:OCcomplex}) is a graded tensor product of the DGFLA $\No^\bullet$ and the supercommutative algebra $\Omega^\bullet(\CM,\RR)$. Therefore it has a natural graded bracket which makes it into a graded Lie algebra. Explicitly, the bracket of decomposable elements $\chg\otimes\phi$ and $\chh\otimes \theta$, $\chg,\chh\in\No^\bullet$, $\phi,\theta\in\Omega^\bullet(\CM,\RR),$ is
\begin{equation}
[\chg\otimes\phi,\chh\otimes \theta]=(-1)^{|\phi||\chh|}[\chg,\chh]\otimes (\phi\wedge\theta).
\end{equation}
The space   $\Omega^\bullet(\CM,\No^\bullet)$ equipped with the total differential $d+\partial$ and the bracket is a DGFLA. 

When $\Omega^\bullet(\CM,\No^\bullet)$  is regarded as a complex with respect to $\partial$, it has a sub-complex consisting of chains and derivations preserving $\psi$ pointwise on $\CM$:
\beq \label{eq:PsiOCcomplex}
\ldots \stackrel{\partial}{\ra} \Omega^{\bullet}(\CM, C_{1}(\mfkdpal))\stackrel{\partial}{\ra} \Omega^{\bullet}(\CM,C_0(\mfkdpal))\stackrel{\partial}{\ra} \Omega^{\bullet}(\CM,\mfkDpal) \stackrel{\partial}{\ra} 0.
\eeq
We will denote it $\Omega^\bullet(\CM,\No^\bullet_\psi)$. This sub-complex is not preserved by $d$. But if the family $(\CM,\psi)$ is  parallel with respect to $\derG\in\Omega^1(\CM,\mfkDal)$, then it is preserved by a ``covariant differential'' $D:= d + \{\derG, \cdot\}$. Indeed, for any $\chb \in \Omega^m(\CM,\No^q_{\psi})$ and any $\CA \in \SAal$ we have
\begin{multline}
\lal \{ D \chb, \CA \} \ral_{\psi} = \lal D \{ \chb, \CA \} \ral_{\psi} + (-1)^{q+m} \lal \{ \chb, D \CA \} \ral_{\psi} = \\ = d \lal \{ \chb, \CA \} \ral_{\psi}  =  0.
\end{multline}
$D$ has bi-degree $(1,0)$ and supercommutes with $\partial$ provided we flip the sign of $\partial$ on odd-degree forms. It also satisfies $D^2 = \{ \derF,\cdot \}$, where the ``curvature'' $\derF := d \derG + \frac12\{\derG,\derG\}\in\Omega^2(\CM,\mfkDal)$ is covariantly constant: $D \derF = 0$. In addition, 
\beq
\lal \derF (\CA) \ral_{\psi} = \lal D^2 \CA \ral_{\psi} = d^2 \lal \CA \ral_{\psi} = 0.
\eeq
Therefore $\derF\in\Omega^2(\CM,\mfkDpal)$.
\begin{remark}
The above properties of $D$ and $\derF$ mean that $\Omega^\bullet(\CM,\No^\bullet_\psi)$ equipped with $D+\partial$ is a curved differential graded Lie algebra (CDGLA) with curvature $\derF$, while the inclusion of $(\Omega^\bullet(\CM,\No^\bullet_\psi),D+\partial)$ into $(\Omega^\bullet(\CM,\No^\bullet),d+\partial)$ is a morphism of CDGLAs.
\end{remark}

\begin{theorem}\label{thm:contrhoOmegaC}
Let $(\CM,\psi)$ be a smooth family of gapped states parameterized by a compact connected manifold $\CM$. Then the complex (\ref{eq:PsiOCcomplex}) is exact with respect to $\p$. Moreover, there is a continuous linear map $h^\psi_\CM$ from  $\Omega^{\bullet}(\CM,Z^k(\No_\psi))$ to $\Omega^{\bullet}(\CM,\No^{k-1}_\psi)$ which is right inverse to $\partial$.
\end{theorem}
\begin{remark}
Since all states in a smooth family are related by LGAs, to check whether a family is gapped it is sufficient to check whether any particular state in the family is gapped.
\end{remark}
\begin{proof}
Let us fix a finite smoothly contractible  cover $\{\CU^{(\sfa)}\}$ and families of LGPs $\beta^{(\sfa)}_\m$, such that $\psi_\m = \psi_0 \circ \beta^{(\sfa)}_\m(1)$ for some $\psi_0$. Since $\psi_0$ is gapped, the complex (\ref{eq:PsiOCcomplex}) becomes  exact if we replace the family $\psi_\m$ with the constant family $\psi_0$. The corresponding right inverse to $\partial$ is the map $h^{\psi_0}$ from Theorem  \ref{thm:poincareH}. Then the right inverse  to $\partial$ for the restriction of (\ref{eq:PsiOCcomplex}) to $\CU^{(\sfa)}$ is
$h^\psi_{\CU^{(\sfa)}}=\beta^{(\sfa)}(1)\circ h^{\psi_0}\circ \beta^{(\sfa)}(1)^{-1}$. Then picking a partition of unity $f^{(\sfa)}$ subordinate to the cover we get a right inverse to $\partial$ on the whole $\CM$ by letting
$h^\psi_\CM=\sum_{\sfa} f^{(\sfa)}h^\psi_{\CU^{(\sfa)}}.$

%Let $\chF \in \Omega^q(\CM,C_n(\mfkdpal))$ be such that $\p \chF = 0$. Using LGAs $\beta^{(\sfa)}_\m(1)$ and the fact that the LGAs act on the complex (\ref{eq:OCcomplex}) by automorphisms of the DGFLA structure, we can  construct $\chG^{(\sfa)} \in \Omega^q(\CU^{(\sfa)},C_{n+1}(\mfkdpal))$ such that on $\CU^{(\sfa)}$ we have $\chF = \p \chG^{(\sfa)}$. Then picking a partition of unity $f^{(\sfa)}$ subordinate to the cover we can write $\chF = \p \chG$ for $\chG = \sum_{\sfa} f^{(\sfa)} \chG^{(\sfa)}\in \Omega^q(\CM,C_{n+1}(\mfkdpal))$. Similarly, one can show that any $\derF \in \Omega^q(\CM,\mfkDpal)$ can be written as $\partial\chG$ for some $\chG\in\Omega^q(\CM,C_0(\mfkdpal))$. The contracting homotopy on $\Omega^{\bullet}(\CM,C_{\bullet}(\mfkdpal))$ is induced by the contracting homotopy at $\psi_0$.

\end{proof}

\subsection{Higher Berry classes}   \label{ssec:descendantsBerry}

For $d=0$ we have $\mfkDal=\mfkdal$. A smooth family of pure states can be identified with a smooth family of self-adjoint rank-1 projectors $P$ in $\SAal$ parameterized by $\CM$. The gapped condition is vacuous for $d=0$. Then $\derF$ is a closed  2-form on $\CM$ with values in $\mfkDal$ which commutes with $P$. If one defines a 2-form $f=\langle\derF\rangle_\psi={\rm Tr} P\derF$, then $f$ is closed and purely imaginary. If we identify $\derG$ with the adiabatic  connection arising from a family of gapped Hamiltonians for which $\psi$ is the family of ground states, then $f$ is the curvature of the Berry connection.

For $d>0$ the expression $\langle\derF\rangle_\psi$ does not make sense since the derivation $\derF$ need not be summable (i.e. the Berry curvature is divergent in the thermodynamic limit). Instead, following \cite{Kitaev_talk,kapustin2020higherB}, we will define descendants of $\derF$ and use them to construct an element of $\Omega^{d+2}(\CM,\mfkdpal)$ whose average will be a closed $(d+2)$-form on $\CM$. 

%Let $(\CM,\psi)$ be a smooth family of gapped states parameterized by a compact connected manifold $\CM$. At each point $\m\in\CM$ we have the complex 
%\beq
%... \ra C_1(\mfkd^{\psi_\m}_{al}) \stackrel{\partial}{\ra} C_0(\mfkd^{\psi_\m}_{al}) \stackrel{\partial}{\ra} \mfkDal.
%\eeq
%Since the $\partial$-cohomology of $\noethcom^{\bullet}_{\psi}$ is trivial, the cohomology of this complex as a vector space can be identified with $\mfkDal/\mfkDal^{\psi_\m}$, elements of which describe possible deformations of the state in a given phase. We denote the corresponding sheaf of non-positively graded DGFLA  by $\defcom_{\psi}^{\bullet}$, so that  $\defcom_{\psi}^{0}=\mfkDal$ and $\defcom_{\psi}^{-n-1}=C_{n}(\mfkdpal)$, $n\geq 0$.

Let $\mathbf{d} = d + \partial$. If $(\CM,\psi)$ is parallel with respect to $D=d+\derG$, then using Theorem  \ref{thm:contrhoOmegaC} we can recursively build $\chg^{(n)} \in \Omega^{n+2}(\CM,\No_{\psi}^{-n-1}),\, n\geq 0,$ such that $\derG^{\bullet} := \derG + \sum_{n=0}^{\infty} \chg^{(n)}$ satisfies the Maurer-Cartan (MC) equation
\beq \label{eq:descent}
\mathbf{d} \derG^{\bullet} + \frac12 \{\derG^{\bullet},\derG^{\bullet}\} = 0.
\eeq
Indeed, the MC equation is equivalent to the following system of equations:
\begin{align}
\derF &=-\partial \chg^{(0)}, & \\ 
D \chg^{(n)}+\frac12 \sum_{k=0}^{n-1} \left\{ \chg^{(k)},\chg^{(n-1-k)}\right\}&=-\partial \chg^{(n+1)}, & n=0,1,2,\ldots
\end{align}
It is straightforward to check that if the first $n$ equations are satisfied, then the left-hand side of the $(n+1)$-st equation is annihilated by $\partial$. Therefore by Theorem  \ref{thm:contrhoOmegaC} a solution for $\chg^{(n+1)}$ exists. 

Any solution to the MC equation 
defines a differential ${\mathbf D} = \mathbf{d} + \derG^{\bullet}$ on $\Omega^\bullet(\CM,\No^\bullet)$ which preserves $\Omega^\bullet(\CM,\No_\psi^\bullet)$ and satisfies ${\mathbf D}^2=0$. The differential ${\mathbf D}$ turns $\Omega^{\bullet}(\CM,\CN^\bullet_{\psi})$ into a DGFLA. Note also that
\beq
d \lal \chg^{(n)} \ral_{\psi} = -  \lal \p  \chg^{(n+1)} \ral_{\psi},
\eeq 
where we have used $\lal \{ \chg^{(k)} ,\cdot \} \ral_{\psi} = 0$. 

For a $d$-dimensional system and any conical partition $(A_0,...,A_d)$ of $\RR^d$ we have an evaluation operation $\lal (\,\cdot\,)_{A_0...A_d} \ral: \Omega^{\bullet}(\CM,\No^{\bullet}) \to \Omega^{\bullet}(\CM,i\RR)$ that takes value $\lal \cha_{A_0...A_d} \ral_{\psi}$ if $\cha \in \Omega^{\bullet}(\CM,\No^{-d-1})$ and is 0 otherwise.

Let $f:=\lal \derG^{\bullet}_{A_0 A_1 ... A_{d}} \ral_{\psi} \in\Omega^{d+2}(\CM,i\RR)$. We have
\beq
d f = d \lal \chg^{(d)}_{A_0 A_1 ... A_{d}} \ral_{\psi} = - \lal (\p \chg^{(d+1)})_{A_0 A_1 ... A_{d}} \ral_{\psi} = 0.
\eeq
Therefore $f$ defines a de Rham cohomology class $[f] \in H^{d+2}(\CM,i\RR)$. Clearly, $f$ is the same for all orderings of $(A_0,...,A_d)$ which correspond to the same orientation of bases on $S^{d-1}$.  Note also that $f$ is locally computable: if we  truncate the sum defining $\chg^{(d)}_{A_0 A_1 ... A_{d}}$ by replacing each $A_a$ with  $A_a\cap B_p(r)$, this will modify $f$ only by an $\Or$ quantity.

One and the same smooth family of states can be parallel with respect to many different connections $\derG\in\Omega^1(\CM,\mfkDal)$, and for a fixed connection $\derG$, the solution of the descent equation (\ref{eq:descent}) is also far from unique. We will now show that the class $[f]$ is not affected by these choices and defines an invariant of $(\CM,\psi)$.

\begin{theorem}  \label{thm:[f]}
The class $[f] \in H^{d+2}(\CM,i\RR)$ defines an invariant of the family of states $(\CM,\psi)$. In particular, it does not depend on the choice of $\derG^{\bullet}$ or a conical partition $(A_0,...,A_d)$ (provided it corresponds to some fixed orientation of $S^{d-1}$).
\end{theorem}
\begin{proof}
Suppose we have changed the regions $A_0,\ldots,A_d$ by reassigning a finite region $A_0 \to A_0 + B$ and $A_1 \to A_1 - B$. Then
\beq
\Delta f = \lal \chg^{(d)}_{B(A_0+A_1)A_2...A_d}  \ral_{\psi} = - \lal \p \chg^{(d)}_{B A_2...A_d}  \ral_{\psi} = d \lal \chg^{(d-1)}_{B A_2...A_d}  \ral_{\psi}
\eeq 
Thus $[f]$ does not depend on the choice of $A_0,\ldots,A_d$ for a fixed triangulation of $S^{d-1}.$ The same argument combined with local computability implies that $[f]$ does not depend on the choice of the triangulation of $S^{d-1}$. 

Let us fix a map $h^{\psi}_\CM$ as in Theorem \ref{thm:contrhoOmegaC}. Using $h^\psi_\CM$, for any two choices $\derG$,$\tilde{\derG}$ of the connection on a given family of states $(\CM,\psi)$ we can construct solutions $\chg^{\bullet}$, $\tilde{\chg}^{\bullet}$ of the MC  equation eq. (\ref{eq:descent}). Let $\chh$ and $\tilde \chh$ be the elements of $\Omega^1(\CM,C_0(\mfkdal))$ such that $\derG=\partial\chh$ and $\tilde\derG=\partial\tilde\chh$. Since $\lal [\derG-\tilde\derG, \cdot] \ral = 0$, we have $\derG - \tilde\derG = \p \chk$ for some $\chk \in \Omega^1(\CM,C_0(\mfkdpal))$. For a half-space $A$ defined by $x>0$ for some linear coordinate $x$ on $\RR^d$ we let $\hat \chh_j = \chh_j + \chi_{A}(j) \chk_j$. Then $\hat \derG = \p \hat \chh$ is a connection which preserves the family  $(\CM,\psi)$ and interpolates between $\derG$ for $x\ll 0$ and  $\tilde{\derG}$ for $x\gg 0$. Starting with $\hat \derG$ and using $h^\psi_\CM$ we can construct a solution $\hat{\chg}^{\bullet}$ of the descent equation that interpolates between $\chg^{\bullet}$ and $\tilde{\chg}^{\bullet}$. The corresponding class $[\hat f]\in H^{d+2}(\CM,\ZZ)$ is independent of the choice of conical partition, and therefore by choosing the apex of a conical partition with sufficiently large and negative $x$ (resp. sufficiently large and positive $x$) we can make it arbitrary close to the class $[f]$ corresponding to $\derG^\bullet$ (resp. the class $[\tilde f]$ corresponding to $\tilde\derG^\bullet$). Hence $[f]=[\hat f]=[\tilde f]$.

It is left to show that for a fixed conical partition $(A_0,...,A_d)$ and a  fixed $\derG$ the class $[f]$ does not depend on the choice of the solution of the MC equation eq. (\ref{eq:descent}). Let $\chg^{\bullet}$ be any solution of the MC equation, and let $\chg_k^{\bullet}$ be a solution obtained from $\derG$, $\chg^{(0)}$, ...,$\chg^{(k-1)}$ using $h^{\psi}_{\CM}$. Let $[f]$ and $[f_k]$ be the corresponding classes. Note that $[f_0]$ depends only on $\derG$ and $h^{\psi}_{\CM}$. Similarly to the argument from the previous paragraph, since $\p(\chg_{k+1}^{(k)}-\chg_{k}^{(k)})=0$ for any $k$ we can construct $\hat{\chg}_{k,k+1}^{\bullet}$ that interpolates between $\chg_k^{\bullet}$ and $\chg_{k+1}^{\bullet}$. Hence we have $[f_k] = [f_{k+1}]$. By definition $[f] = [f_{d+1}]$, and therefore $[f] = [f_0]$, which does not depend on the choice of $\chg^{\bullet}$.
\end{proof}

\begin{remark}
Let $\CO^{\times}$ be the sheaf of smooth $U(1)$-valued functions on $\CM$ and $\Omega^n$ be the sheaf of smooth purely imaginary $n$-forms on $\CM$.
Recall that a degree-$2$ Deligne-Beilinson $2$-cocycle is a 2-cocycle of the totalization of the \Cech\ cochain complex with coefficients in the complex of sheaves
\beq
[\CO^{\times} \xrightarrow[]{d\, \log} \Omega^1 \xrightarrow[]{d} \Omega^2].
\eeq
In the physics literature, such 2-cocycles are called 2-form gauge fields, and it is well-known that any 2-form gauge field determines a class in $H^3(\CM,\ZZ)$.
For a family $(\CM,\psi)$ of 1d states in the trivial phase one can construct a Deligne-Beilinson 2-cocycle and show that the corresponding class in $H^3(\CM,\ZZ)$ maps to $[f/2\pi i]$. Thus for families of 1d states in a trivial phase the class $[f]$ is quantized. The details will be discussed elsewhere.
\end{remark}

\subsection{Equivariant higher Berry classes } \label{ssec:descendantsSymmetryG}

\subsubsection{General construction}

Let $G$ be a compact connected Lie group with a Lie algebra $\mfkg$. In this section we define equivariant higher Berry classes for gapped states and their smooth families in the presence of $G$-symmetry. They generalize the Hall conductance of $U(1)$-invariant gapped 2d states defined in \cite{kapustin2020hall} and the Thouless pump \cite{Thouless} and its analogs in higher dimensions \cite{kapustin2020higherB}. 

We assume that each on-site Hilbert space $\CV_j$ is acted upon by a unitary representation of $G$. Let $\gamma:G \to \Aut(\SA)$ be the corresponding homomorphism. We denote the derivations corresponding to the infinitesimal generators of the symmetry by $\derQ \in \mfkDal \otimes \mfkg^{*}$, that is on $G$ we have $d \gamma^{(g)}(\,\cdot\,) = \gamma^{(g)} (\derQ^{(g)}_{\theta}(\,\cdot\,))$, where $\theta = g^{-1} d g \in \Omega^1(G,\mfkg)$ is the Maurer-Cartan form.

% For a manifold $\CM$ with a smooth action $L_g:\CM \to \CM, \, g \in G$ of the Lie group $G$ we denote the $\RR$-valued equivariant cohomology by $H^{\bullet}_G(\CM,\RR)$. In the Cartan model this cohomology can be defined as the cohomology of the complex of the equivariant differential forms $\l \Omega^{\bullet}(\CM,\RR) \otimes \text{Sym}^{\bullet}\,\mfkg^* \r^G$ with respect to the Cartan differential $d_C = d + \iota_v$ satisfying $d_C^2 = \CL_{v}$.

Let $\CM$ be a manifold with a smooth $G$-action $L_g:\CM \to \CM, \, g \in G$. Let $v \in \Gamma(T \CM) \otimes \mfkg^*$ be the vector fields corresponding to this action. If we choose a basis $r^a$ for $\mfkg^*$, then we can write $v=\sum_a v_a\otimes r^a$, where $v_a$ is a vector field on $\CM$.   There is an induced $G$-action on $\Omega^{\bullet}(\CM,\RR)$ via the pullback $\omega \mapsto L_{g^{-1}}^{*} \omega$. Infinitesimally, it is given by
$\CL_v = \sum_a \CL_{v^a} \otimes r^a$ where $\CL_{v_a}=\iota_{v_a} d+d\iota_{v_a}$ is the Lie derivative along $v_a$.

Let $\text{Sym}^{\bullet}\,\mfkg^*$ be the ring of polynomial functions on $\mfkg$. We regard it as a graded vector space, with linear functions on $\mfkg$ sitting in degree $2$. $G$ acts on it via the co-adjoint action $\Ad^*_g$. To define the Cartan model for the equivariant cohomology $H_G^{\bullet}(\CM,\RR)$ we consider the graded vector space 
$ \Omega^{\bullet}(\CM,\RR) \otimes \text{Sym}^{\bullet}\,\mfkg^*$ equipped with degree-1 maps $d \otimes 1$, $\iota_v = \sum_a \iota_{v^a} \otimes r^a$, and a degree-2 map $\CL_v=\iota_v d+d\iota_v$. Then the complex of the Cartan model $\Omega_G^{\bullet}(\CM,\RR)$ is defined to be $ (\Omega^{\bullet}(\CM,\RR) \otimes \text{Sym}^{\bullet}\,\mfkg^*)^G$, where $(\,\cdot\,)^G$ denoted the $G$-invariant part, with the differential being $d_C = d + \iota_v$. Note that $d_C^2 = \CL_v$ vanishes when restricted to  $\Omega_G^{\bullet}(\CM,\RR)$.

% Let $\CM$ be a manifold with a smooth $G$-action $L_g:\CM \to \CM, \, g \in G$. Let $v \in \Gamma(T \CM) \otimes \mfkg^*$ be the vector fields corresponding to this action. The contraction operation and the Lie derivative are denoted $\iota_{v}$ and $\CL_v$, respectively. We denote the $\RR$-valued equivariant cohomology by $H^{\bullet}_G(\CM,\RR)$. In the Cartan model this cohomology can be defined as the cohomology of the complex of the equivariant differential forms $\l \Omega^{\bullet}(\CM,\RR) \otimes \text{Sym}^{\bullet}\,\mfkg^* \r^G$ with respect to the Cartan differential $d_C = d + \iota_v$. By definition, equivariant differential forms are elements of $\Omega^{\bullet}(\CM,\RR) \otimes \text{Sym}^{\bullet}\,\mfkg^*$ that are annihilated by $\CL_v$. Thus $d_C^2 = \CL_{v}$ vanishes on $\l \Omega^{\bullet}(\CM,\RR) \otimes \text{Sym}^{\bullet}\,\mfkg^* \r^G$.

Let $(\CM,\psi)$ be a smooth family of gapped states over a compact connected manifold $\CM$.
\begin{definition}
A family $(\CM,\psi)$ is called $G$-equivariant if $\psi_m\circ\gamma^{(g)}=\psi_{L_{g^{-1}}(m)}$ for all $m\in \CM$ and all $g\in G$.
\end{definition} 
\noindent
We can define equivariant analogs of  differential forms taking values in observables, derivations and chains. They  are annihilated by $\CL_v - \derQ$. There is an averaging operation that projects spaces of such differential forms to their $G$-equivariant subspaces:
\beq
(\,\cdot\,)^{G} = \int_{G} d \mu_{G}(g) \, \gamma^{(g)}(L_{g}^*(\Ad^*_{g^{-1}}(\,\cdot\,))),
\eeq
where $\mu_{G}$ is the Haar measure with the normalization $\int_G d \mu_{G}(g) = 1$. The existence of this operation  ensures that $G$-equivariant versions of the complexes (\ref{eq:OCcomplex}) and (\ref{eq:PsiOCcomplex}) are exact.

Given a smooth $G$-equivariant family of gapped states $(\CM,\psi)$, we consider the complex of equivariant differential forms $\Omega_G^{\bullet}(\CM,\No^{\bullet})$ which is the totalization of $(\Omega^{\bullet}(\CM,\No^{\bullet}) \otimes \text{Sym}^{\bullet}\,\mfkg^*)^G$ with elements of $\mfkg^{*}$ being assigned degree $2$. Let $\mathbf{d}_C = d_C + \partial$. Then using the $G$-equivariant version of Theorem \ref{thm:contrhoOmegaC} we can recursively build  
$$
\chg^{(n)} \in \bigoplus_{0 \leq 2k \leq n+2} ( \Omega_G^{n+2-2k}(\CM,\No_{\psi}^{-n-1}) \otimes \text{Sym}^{k}\,\mfkg^*)^G,\,\,\, n\geq 0,$$
such that $\derG^{\bullet} := \derG + \sum_{n=0}^{\infty} \chg^{(n)}$ satisfies 
\beq \label{eq:descentEquivariant}
\mathbf{d}_C \derG^{\bullet} + \frac12 \{\derG^{\bullet},\derG^{\bullet}\} + \derQ = 0
\eeq
so that for ${\mathbf D}_C = \mathbf{d}_C + \derG^{\bullet}$ we have ${\mathbf D}_C^2 = \CL_v - \derQ $ that vanishes on $\Omega_G^{\bullet}(\CM,\No^{\bullet})$.

Similarly to the previous subsection, we can define the evaluation operation $\lal (\,\cdot\,)_{A_0...A_d} \ral: \Omega_G^{\bullet}(\CM,\No^{\bullet}) \to \Omega_G^{\bullet}(\CM,i\RR)$. For a conical partition $(A_0,...,A_d)$ we define an equivariant differential form
\beq
f = \lal \derG^{\bullet}_{A_0...A_d} \ral_{\psi}\in\Omega^\bullet_G(\CM,i\RR)
\eeq
which is closed with respect to $d_C$.
In general, unlike in the non-equivariant case, this form is not homogeneous when regarded as an ordinary form. The form $f$ represents a class of total degree $d+2$ in the equivariant cohomology:
\beq
[f]=[\lal \derG^{\bullet}_{A_0...A_d} \ral_{\psi}] \in H_G^{\bullet}(\CM,i\RR).
\eeq
In the same way as in Theorem \ref{thm:[f]} one can show that this class is independent of the choice of $(A_0,...,A_d)$ (provided it corresponds to some fixed orientation of $S^{d-1}$) and $\derG^{\bullet}$, and therefore defines an invariant of the family $(\CM,\psi)$.

\subsubsection{Special case: Hall conductance}
Let $d=2$, $\CM = \text{pt}$ and $G=U(1)$. The corresponding state $\psi$ must be $U(1)$-invariant. For $G=U(1)$ we can identify $(\Sym^{k+1}\mfkg^*)^G \cong \mathbb{R}$ via the isomorphism that sends the minimal integral element $t^{\otimes (k+1)}$ to $1$. It is easy to see from the Maurer-Cartan  equation eq. (\ref{eq:descentEquivariant}) that  $\derG^{\bullet}$ only has components of  even chain degree. Let $\chm^{(2k)}$ be the component of $\derG^{\bullet}$ of chain degree $2k$. The lowest two components of eq. (\ref{eq:descentEquivariant}) are 
\begin{align}
\derQ = - \p \chm^{(0)},\label{degreezero}\\
\frac12 \{\chm^{(0)},\chm^{(0)}\} = - \p \chm^{(2)}.\label{degreetwo}
\end{align}
In more detail, eq. (\ref{degreezero}) says that the derivation $\derQ$ has the form
\begin{equation}
\derQ(\CA)=-\sum_{j\in\Lambda} [\chm^{(0)}_j,\CA],   
\end{equation}
where for all $j\in\Lambda$ the observable $
\chm^{(0)}_j\in\mfkdpal$ is $U(1)$-invariant and does not excite $\psi$. Eq. (\ref{degreetwo}) reads in components:
\begin{equation}
[\chm^{(0)}_j,\chm^{(0)}_k]=-\sum_{l\in\Lambda} \chm^{(2)}_{ljk},
\end{equation}
where the observables $\chm^{(2)}_{ljk}\in\mfkdpal$ are $U(1)$-invariant and do not excite $\psi$. 
The contraction of $\chm^{(2)}$ with a conical partition $A_0,A_1,A_2$ defines an invariant
\beq
\sigma^{(2)} := 4 \pi i \lal \chm^{(2)}_{A_0 A_1 A_{2}} \ral_{\psi} \in H_{U(1)}^{4}(\text{pt},{\mathbb R}) \cong \RR.
\eeq
For ground states of gapped Hamiltonians in two dimensions $(\sigma^{(2)}/2\pi)$ coincides with the Hall conductance as explained in \cite{ThoulessHall}.

\begin{remark} \label{rmk:HallQuant}
It was shown by M. Hastings and S. Michalakis  \cite{hastingsmichalakis} that the Hall conductance of any gapped $U(1)$-invariant system on a torus of size $L$ with a non-degenerate ground state is integral up to $\OL$ corrections. Building on this work as well as \cite{bachmann2019many}, it was shown in \cite{ThoulessHall} that for invertible infinite-volume gapped states one has  $\sigma^{(2)} \in 2 \ZZ$ in the bosonic case and $\sigma^{(2)} \in \ZZ$ in the fermionic case.
\end{remark}

\subsubsection{Special case: non-abelian Hall conductance}

The case of $d=2$, $\CM = \text{pt}$ and arbitrary Lie group $G$ is similar, but now $\chm^{(0)} \in \left( C_0(\mfkdpal) \otimes \mfkg^* \right)^G$, satisfying $\derQ = - \p \chm^{(0)}$, and $\chm^{(2)} \in \left( C_2(\mfkdpal) \otimes \text{Sym}^2 \mfkg^* \right)^G$ solving
\begin{equation}
[\chm^{(0)}_j,\chm^{(0)}_k]=-\sum_{l\in\Lambda} \chm^{(2)}_{ljk},
\end{equation}
The contraction of $\chm^{(2)}$ with a conical partition $A_0,A_1,A_2$ defines an invariant
\beq \label{eq:Ghall}
\sigma^{(2)} := 4 \pi i \lal \chm^{(2)}_{A_0 A_1 A_{2}} \ral_{\psi} \in H_{G}^{4}(\text{pt},{\mathbb R}) \cong (\text{Sym}^2 \mfkg^*)^G.
\eeq
For a simple Lie algebra $\mfkg$ we have $(\text{Sym}^2 \mfkg^*)^G \cong \RR$. This agrees with the expectation that the response of a  gapped 2d system with a symmetry $G$ to a background gauge field is described by a 3d Chern-Simons action whose coefficient (level) is a topological invariant. 

\subsubsection{Special case: higher-dimensional generalizations of the Hall conductance}

More generally, let us take $\CM=\text{pt}$ but consider an arbitrary $d$ and an arbitrary compact Lie group $G$. Let $\psi$ be a $G$-invariant gapped state. As before,  $\derG^{\bullet}$ only has components of  even chain degree. Let $\chm^{(2k)}$ be the component of $\derG^{\bullet}$ of chain degree $2k$. For $d=2k$ the contraction with $A_0,...,A_d$ defines an invariant 
\beq
i\lal \chm^{(2k)}_{A_0...A_{2k}} \ral_{\psi} \in (\Sym^{k+1}\mfkg^*)^G \cong H_G^{2k+2}(\text{pt},{\mathbb R})
\eeq
% The factor $i$ is inserted to make the invariant real rather than purely imaginary.  

% For $G=U(1)$ we can identify $(\Sym^{k+1}\mfkg^*)^G \cong \mathbb{R}$ via the isomorphism that sends the minimal integral element $t^{\otimes (k+1)}$ to $1$. Then $\lal \chm^{(2k)}_{A_0...A_{2k}} \ral_{\psi}$ defines an $\RR$-valued invariant \beq \label{eq:Ghall}
% \sigma^{(2k)}:= 4 \pi i \lal \chm^{(2k)}_{A_0...A_{2k}} \ral_{\psi}.
% \eeq
% For ground states of gapped Hamiltonians in two dimensions $(\sigma^{(2)}/2\pi)$ coincides with the Hall conductance \cite{ThoulessHall}.

% \begin{remark} \label{rmk:HallQuant}
% It was shown by M. Hastings and S. Michalakis  \cite{hastingsmichalakis} that the Hall conductance of any gapped $U(1)$-invariant system on a torus of size $L$ with a non-degenerate ground state is integral up to $\OL$ corrections. Building on this work as well as \cite{bachmann2019many}, it was shown in \cite{ThoulessHall} that for invertible infinite-volume gapped states one has  $\sigma^{(2)} \in 2 \ZZ$ in the bosonic case and $\sigma^{(2)} \in \ZZ$ in the fermionic case.
% \end{remark}

\begin{remark} \label{rmk:effectiveCS}
It is often assumed that at long distances a gapped system with a Lie group symmetry can be effectively described by a field theory with the effective action that contains Chern-Simons terms. From that perspective the invariants $\lal \chm^{(2k)}_{A_0...A_{2k}} \ral_{\psi}$ should take values in the space of Chern-Simons forms of degree $2k+1$. By Chern-Weil theory, the space of Chern-Simons forms of degree $2k+1$ is $H^{2k+2}(BG,\RR)$ which agrees with the above result. 
%$HFor example, $\sigma^{(2k)}$ should correspond to the coefficients in front of $\int A d A ... d A$ for $U(1)$ gauge field $A$. 
\end{remark}

\subsubsection{Special case: Thouless pump and its generalizations}

Let $G=U(1)$ act trivially on $\CM$ and $\psi$ be a smooth family of $U(1)$-invariant states parameterized by $\CM$. We may identify $\mfkg^*\cong\RR$. Let $\cht^{(k)} \in ( \Omega^k(\CM,\No_{\psi}^{-k-1}) \otimes \mfkg^*)^G$ be the component of $\derG^{\bullet}$ of chain degree $k$ and form  degree $k$, $k=0,1,\ldots$. The Maurer-Cartan equation (\ref{eq:descentEquivariant}) implies an infinite number of equations for $\cht^{(k)}$, the first three of which look as follows:
\begin{align}
\derQ&=-\partial \cht^{(0)},\\
D\cht^{(0)}&=-\partial \cht^{(1)},\\
D\cht^{(1)}+ \{\chg^{(0)},\cht^{(0)}\}&=-\partial\cht^{(2)}.
\end{align}

For $d=1$ and $\CM = S^1$ the invariant $\int_{S^1} \lal \cht^{(1)}_{A_0 A_1} \ral_{\psi}$ computes the charge pumped through any point of the system in the process of the adiabatic evolution along a loop $S^1$ in the parameter space.

For $d=k$ and $\alpha\in H_k(\CM,\ZZ)$ the invariant  $\int_{\alpha} \lal \cht^{(k)}_{A_0 A_1...A_k} \ral_{\psi}$ is a  higher-dimensional generalization of the Thouless pump introduced in \cite{kapustin2020higherA}.

\subsection{A family of states with a nonzero higher Berry class}       \label{ssec:BerryFrom2dSPT}

Examples of states with a non-trivial equivariant Berry class (e.g. with a non-trivial Hall conductance) are well known. In this section we give an example of a family of 1d states with a nonzero non-equivariant higher Berry class. This family is associated with a compact connected Lie group $G$ and a $G$-invariant 2d state in the trivial phase but with a nonzero non-abelian Hall conductance (such 2d states are known as Symmetry Protected Topological states). Other examples of families with a nonzero higher Berry class recently appeared in \cite{flow}.

Let us consider a two-dimensional lattice system $(\Lambda,\SA)$ with an action of $G$ as defined in Section \ref{ssec:descendantsSymmetryG}.
Since $\derQ = \p \chq$ for $ \chq \in C_0(\mfkdal)\otimes \mfkg^*$, for any $\Gamma \subset \RR^2$ we can define an  action of $G$ generated by $\derQ|_{\Gamma} = \chq_{\Gamma}$. We denote the corresponding homomorphism by $\gamma_{\Gamma}:G \to \text{Aut}(\SA)$. We denote the left-invariant Maurer-Cartan form by $\theta := g^{-1} \, d g  \in \Omega^{1}(G,\mfkg)$.

Let $\psi_0$ be a factorized pure state. Let $\beta$ be an LGP such that the state $\omega = \psi_0 \circ \beta(1)$ is also $G$-invariant, with a $G$-action defined by $\derQ\in \left(\mfkDal\otimes \mfkg^*\right)^G$. The equivariant Berry class eq. (\ref{eq:Ghall}) for $\omega$ defines an invariant quadratic form $\sigma^{(2)}$ on $\mfkg$. This form is the non-abelian analog of the Hall conductance.

Let $(A_0,A_1,A_2)$ be a conical partition. We define a smooth family of states $(\CM,\psi)$ with $\CM = G$, such that at a point $g \in G$ we have $\psi_g = \psi_0 \circ \beta(1) \circ \gamma^{(g)}_{A_0} \circ \beta(1)^{-1}$. Note that though the state is defined on a two-dimensional lattice, far from the boundary of $A_0$ it is almost factorized. Hence we may call it ``quasi-one-dimensional".

We have
\beq
d \lal \CA \ral_{\psi} = \lal \{\beta(1)(\gamma_{A_0}^{-1}(\chq_{A_0})),\CA\} \ral_{\psi}.
\eeq
Since $\derQ \in (\mfkDal^{\omega}\otimes \mfkg^*)^G$, there are $G$-invariant $\tilde{\chq} \in C_0 (\mfkd^{\omega}_{al}) \otimes \mfkg^*$ and $\chk \in C_1 (\mfkdal) \otimes \mfkg^*$ such that $\chq = \tilde{\chq} - \p \chk$. Therefore $(\CM,\psi)$ is a family of quasi-one-dimensional states parallel with respect to $\derG = \beta(1)(\gamma_{A_0}^{-1}(\chk_{A_0 \overline{A}_0}(\theta))) \in \Omega^1(G,\mfkDal)$.

We have
\begin{multline}
\gamma_{A_0} \circ \beta(1)^{-1} \l d \derG + \frac12 \{\derG,\derG\} \r = \\ = \{\chq_{A_0}(\theta),\chk_{A_0 \overline{A}_0}(\theta)\} + \chk_{A_0 \overline{A}_0}(d \theta) + \frac12 \{\chk_{A_0 \overline{A}_0}(\theta),\chk_{A_0 \overline{A}_0}(\theta)\} = \\ = -
\frac12 \{\tchq_{\overline{A}_0}(\theta),\tchq_{A_0}(\theta)\} = -
\frac12 \{\tchq_{A_1}(\theta),\tchq_{A_0}(\theta)\}-
\frac12 \{\tchq_{A_2}(\theta),\tchq_{A_0}(\theta)\}.
\end{multline}
% \begin{multline}
% w^{(g)}(\derF) = - \{ \chq_A^a \theta_a, \chk_{A \overline{A}}^b \theta_b \} + \frac12 f^{abc} \chk^c_{A \overline{A}} \theta_a \theta_b + \frac12 \{ \chk_{A \overline{A}}^a \theta_a,\chk_{A \overline{A}}^b \theta_b \} = \\
% = \frac12 \l [\tchq_A^a,\tchq_A^b] - f^{a b c} \tchq_A^c \r \theta_a \theta_b = - \frac12 [\tchq^a_A,\tchq^b_{\overline{A}}] \theta_a \theta_b.
% \end{multline}
Hence the contraction of the Berry class $[f] \in H^3(\CM,i \RR)$ for a family $(\CM,\psi)$ with $[\CM]$ is given by
\begin{multline}
\lal [f],[\CM] \ral = \int_G f = \\ = \frac12  \lal \int_G \{\chq_{A_0}(\theta)-\chk_{A_0 \overline{A}_0}(\theta),\{\tchq_{A_1}(\theta),\tchq_{A_0}(\theta)\}\} \ral_{\omega}  = \\ =  -\frac14  \lal \int_G \{\tchq_{A_1}(\theta),\{\tchq_{A_0}(\theta),\tchq_{A_0}(\theta)\}\} \ral_{\omega} = \\ = - \frac{1}{6} \lal \int_G \{\tchq_{A_1}(\theta),\tchq_{A_0}(d\theta)\} \ral_{\omega} = \frac{1}{12 \pi i} \int_G \lal \theta , d \theta \ral_{\sigma^{(2)}}.
\end{multline}
Note that all derivations inside the averages $\lal\, \cdot\,\ral_{\omega}$ belong to the subspace of summable derivations $\mfkd^{\omega}_{al}$, and therefore the average is well-defined. If $\sigma^{(2)}$ does not vanish, the Berry class $[f]$ is non-trivial.

\newpage
\appendix
\section{{\Frechet} spaces}  \label{app:frechet}

A seminorm on a real or complex vector space $V$ is a map $V\ra\RR$, $v\mapsto \|v\|$ such that $\|v\|\geq 0$ for all $v\in V$, $\|v+v'\|\leq \|v\|+\|v'\|$ for all $v,v'\in V$, and $\|c v\|=|c|\|v\|$ for all $v\in V$ and all scalars $c$. A seminorm is a norm if $\|v\|=0$ implies $v=0$.

A \Frechet\ space is a complete Hausdorff topological vector space whose topology is determined by a countable family of seminorms $\|\cdot\|_\alpha,$ $\alpha\in\NN_0$. A base of neighborhoods of zero for such a topology consists of sets \begin{equation}
U_{(\alpha_1,\eps_1)\ldots(\alpha_n,\eps_n)}=\{v\in V:\, \|v\|_{\alpha_i}<\eps_i,\, i=1,\ldots,n\},
\end{equation}
where $n\in\NN$, $\alpha_i\in\NN_0$, and $\eps_i>0$. Any finite-dimensional Euclidean vector space is a special case where all the seminorms happen to be the same and equal to the Euclidean norm. 

In this paper we will be often dealing with a situation where the seminorms satisfy $\|\cdot\|_0\leq \|\cdot\|_1\leq\|\cdot\|_2\leq\ldots$. One calls such \Frechet\ spaces graded \Frechet\ spaces. Then the sets $U_{\alpha,\eps}=\{v\in V :\, \|v\|_\alpha <\eps\}$, $\alpha\in\NN_0,$ also form a base of neighborhoods of zero. A linear map $f:V\ra V'$ between graded \Frechet\ spaces is continuous iff for any $\alpha\in\NN_0$ there is a $\beta\in\NN_0$ and a constant $C_\alpha$ such that $\|f(v)\|_\alpha\leq C_\alpha \|v\|_\beta$. The Cartesian product of two graded \Frechet\ spaces $V,V'$ is also a graded \Frechet\ space, with the seminorms $\|(v,v')\|_\alpha= \|v\|_\alpha+\|v'\|'_\alpha$.

Different families of seminorms on $V$ may define the same topology; in that case one says that the families are equivalent. A family of seminorms $\|\cdot\|'_{\beta}$, $\beta\in\NN_0$ is equivalent to a family $\|\cdot\|_\alpha$, $\alpha\in\NN_0$, if for any $\beta$ there is an $\alpha$ and a constant $C_{\beta}$ such that $\|\cdot\|'_{\beta} \leq C_{\beta} \|\cdot\|_\alpha$, and vice versa, for any $\alpha$ there is an $\beta$ and a constant $C'_\alpha$ such that $\|\cdot\|_\alpha\leq C'_\alpha\|\cdot\|'_\beta$. 

If $X$ is a compact topological space   and $V$ is a graded \Frechet\ space, then the space $C(X,V)$ of continuous $V$-valued functions on $X$ is also a graded \Frechet\ space. The corresponding family of seminorms is $\|f\|_\alpha=\sup_{x\in X} \|f(x)\|_\alpha$, $\alpha\in\NN_0$. In the case when $X=[a,b]\subset\RR$ elements of $C([a,b],V)$ are called continuous curves in $V$. Most basic rules of calculus (such as the existence of integrals of continuous functions, the  Fundamental Theorem of Calculus, the Mean Value Theorem, the continuous dependence of integrals of continuous functions on parameters, etc.) hold in the setting of continuous functions on regions in $\RR^n$ valued in a \Frechet\ space $V$, see \cite{Hamilton} for a review. 

%If $X$ is non-compact, but contains a countable family of compact subsets $K_0\subset K_1\subset\ldots$ such that $\cup_n K_n=X$, then $C(X,V)$ is a \Frechet\ space whose topology can be define by a family of seminorms $\|f\|_{n,\alpha}=\sup_{x\in K_n} \|f(x)\|_\alpha $. One can turn this \Frechet\ space into a graded \Frechet\ space by defining an equivalent family of seminorms $\|f\|'_\beta=\sum_{n+\alpha=\beta}\|f\|_{n,\alpha}.$ In this paper we only consider non-compact spaces of the form $X=Y\times\RR$ where $Y$ is compact, and then set $K_n=Y\times [-n,n]$.

If $X$ is a non-compact topological space which is a union of compact subsets $K_0\subset K_1\subset \ldots$ such that $K_i\subseteq {\rm int}(K_{i+1})$ for all $i$, then $C(X,V)$ is a \Frechet\ space whose topology can be defined using seminorms $\|f\|_{n,\alpha}=\sup_{x\in K_n} \|f(x)\|_\alpha$, $n,\alpha\in\NN_0$. This topology is independent of the choice of the compact subsets $K_i$. Similar seminorms are also useful when defining \Frechet\  topology on spaces of smooth functions and differential forms on manifolds. Let $\CM$ be a compact manifold of dimension $M$ and $V$ be a graded  \Frechet\ space. A function $f:\CM\ra V$ is called smooth iff derivatives of all orders exist and are continuous. In particular, a smooth function $f:[a,b]\ra V$ is called a smooth curve in $V$. The space of smooth $V$-valued functions on $\CM$ is denoted $C^\infty(\CM,V)$. One can define a \Frechet\ topology on $C^\infty(\CM,V)$ as follows. First, we choose an atlas $\left\{\left(U^{(\sfa)},\kappa^{(\sfa)}:U^{(\sfa)}\ra\RR^M\right)\right\}$ for $\CM$. Then we define a family of seminorms labeled by a chart index $\sfa$, a compact subset $K_n^{(\sfa)}\subset U^{(\sfa)}$,  $k\in\NN_0$ , and $\alpha\in\NN_0$:
\beq
\|f\|_{\sfa,n,k,\alpha}=\max_\sfa \max_{|I|\leq k} \sup_{\m\in K_n^{(\sfa)}} \|\partial_I f(\m)\|_\alpha .
\eeq
Here $I=\{i_1,\ldots,i_M\}\in \NN_0^M$, is a multi-index and $|I|=\sum_{q=1}^M i_q$. Similarly, one can define a \Frechet\ topology on the space of smooth $p$-forms $\Omega^p(\CM,V)$ by regarding the restriction of a $p$-form $\omega$ to $U^{(\sfa)}$ as a collection of $\binom{M}{p}$ $V$-valued functions. One can show that the topologies thus defined do not depend on the choice of the atlas.

We will also need the notion of a possibly nonlinear smooth map from a \Frechet\ space $W$ to a \Frechet\ space $V$.  One says that $f:W\ra V$ is continuously differentiable if the directional derivative
\beq
df(w,\Delta w)=\lim_{t\ra 0}\frac{f(w+t\Delta w)-f(w)}{t}
\eeq
exists and is a continuous function on $W\times W$.
Iterating this definition, one says that a function $f:W\ra V$ is smooth if directional derivatives of all orders exist and are continuous functions on $W\times W\times\ldots$. 
In particular, continuous linear maps are smooth. So are continuous bilinear maps from $W\times W'$ to $V$. Composition of smooth maps is smooth, other basic rules of calculus also hold true \cite{Hamilton}. We will use these notions to construct certain smooth maps from a manifold $\CM$ to a \Frechet\ space $V$ as  compositions of smooth maps from $\CM$ to a \Frechet\ space $W$ and smooth maps from $W$ to $V$. 
%This is not the only possible notion of smoothness, however. $f:U\ra V$ is called conveniently smooth if it maps smooth curves in $U$ to smooth curves in $V$ \cite{convenient}. This condition is often easier to check than smoothness. If $W,V$ are general locally convex spaces, a conveniently smooth map between them need not be continuous and thus may not be smooth in the ordinary sense. But for \Frechet\ spaces convenient smoothness is equivalent to smoothness \cite{convenient}. So for our purposes we need not distinguish them.

\section{Algebra of almost local observables} \label{app:Aal}

% \begin{lemma}
% \label{lma:goodoldfr}
% Let $\CA$ be an almost almost local on $j$ observable. Then there is a monotonically positive decreasing (MDP) function $f(r)=\Or$, that depends on seminorms $\|\CA\|_{j,n}$ only, and a decomposition $\CA = \sum_{r=1}^{\infty} \CA^{(n)}$, such that $\CA^{(r)} \in \SA_{B_j(r)}$ and $\|\CA^{(r)}\| \leq f(r)$.
% \end{lemma}

% \begin{proof}

% \end{proof}

Fix $j\in\Lambda$. Since $\|\cdot\|'_{j,\alpha}$ are norms and $\SAal$ is complete in these norms, the topology induced by $\|\cdot\|'_{j,\alpha}$ makes $\SAal$ into a complete Hausdorff space, and therefore into a \Frechet\ space.
% \begin{prop}
% The topology on $\SAal\subset\SA$ defined by the family of norms $\lVert\cdot\rVert'_{j,\alpha}$, $\alpha\in\NN$, makes $\SAal$ into a \Frechet\ space.
% \end{prop}
% \begin{proof}
% Since $\|\cdot\|'_{j,\alpha}$ are norms, $\SAal$ is a Hausdorff space. To show that it is complete, consider a Cauchy sequence $\CA^{(n)}\in\SAal, n\in\NN$. Since $\SA$ is complete with respect to $h_{j,0}$, the sequence  converges to $\CA\in\SA$. To show that $\CA\in\SAal$, fix $\alpha\in \NN$ and $M>0$ such that $f_{j,\alpha}(\CA_n-\CA_m)\leq 1$ for all $n,m\geq M$. Then for all $n,m\geq M$ and all $r\geq 0$ we have
% \begin{equation}
% \inf_{\CB\in\SA_j(r)} \lVert \CA_n-\CA_m-B\rVert \leq (1+r)^{-\alpha}. 
% \end{equation}
% Passing to the limit $m\ra \infty$ we get $f_{j,\alpha}(\CA_n-\CA)\leq 1$ for all $n\geq M$.
% \end{proof}
\begin{prop}
$\SAal$ is a \Frechet\ algebra.
\end{prop}
\begin{proof}
For any $\CA,\CA'\in\SAal$ and all $r\geq 0$ one has
\begin{multline}
f_j(\CA\CA',r)\leq f_j(\CA,r) f_j(\CA',r)+\lVert \CA\rVert f_j(\CA',r)+\lVert \CA'\rVert f_j(\CA,r) \leq \\ \leq \frac{3}{2} \l \lVert \CA\rVert f_j(\CA',r)+\lVert \CA'\rVert f_j(\CA,r) \r.
\end{multline}
Therefore for any $\alpha\in\NN$ one has
% $$
% f_{j,\alpha}(\CA\CA')\leq \frac{5}{2}f_{j,\alpha}(\CA)\lVert \CA'\rVert + \frac{5}{2}f_{j,\alpha}(\CA')\lVert \CA\rVert
% $$
\beq
\|\CA \CA'\|'_{j,\alpha} \leq \frac{3}{2} \|\CA \|'_{j,\alpha} \|\CA'\|'_{j,\alpha}
\eeq
that implies joint continuity of the multiplication. 
\end{proof} 

\begin{prop}
The topology on $\SAal$ defined by the norms $\|\cdot\|'_{j,\alpha}$ for a fixed $j\in\Lambda$ is independent of the choice of $j$.
\end{prop}
\begin{proof}
%It is easy to see that the seminorms $$ for different $j$ are %related as follows:
%$$
%f_{j',k}\leq (1+D)^k f_{j,0}+\sum_{n=0}^k \binom{k}{n} D^{k-n} %f_{j,n},
%$$
%where $D=|j-j'|$. This proves that the subspace $\SAal$ is independent of the choice of $j$, and moreover that the topologies defined by these two families of 
For any $j,k\in\Lambda$ with $R=|j-k|$ we have $f_j(\CA,r+R)\leq f_k(\CA,r)$. Hence
\beq
\|\CA\|'_{j,\alpha} \leq \|\CA\| + \sup_r (1+r+R)^{\alpha} f_k(\CA,r) \leq (1+R)^{\alpha} \|\CA\|'_{k,\alpha}.
\eeq
Therefore the families of norms $\{\|\cdot\|'_{j,\alpha}\}$ and $\{\|\cdot\|'_{k,\alpha}\}$ are equivalent.
\end{proof}

\begin{lemma} \label{lma:chainapprox}
Let $j,k\in\Lambda$. Suppose an observable $\CA \in \SA$ satisfies $f_{j}(\CA,r)\leq \eps_1$ and $f_{k}(\CA,r)\leq \eps_2$. Let $\CB \in \SA_{B_j(r) \cap B_j(r)}$ be a best possible approximation of $\CA$ on $B_j(r) \cap B_k(r)$. Then
\beq
\|\CA - \CB\| \leq \eps_1 + \eps_2 + \min(\eps_1,\eps_2).
\eeq
\end{lemma}
\begin{proof}
Let $\CA^{(1)}$ and $\CA^{(2)}$ be  best possible approximations of $\CA$ on $B_j:=B_j(r)$ and $B_k:=B_k(r)$, respectively, and let $\CB^{(1)}$ (resp. $\CB^{(2)}$) be a best possible approximation of $\CA^{(1)}$ (resp. $\CA^{(2)}$) on $B_{jk}:=B_j(r) \cap B_k(r)$. Then 
\begin{multline}
\|\CA-\CB\| \leq \|\CA-\CB^{(1)}\| \leq \|\CA-\CA^{(1)}\| + \|\CA^{(1)}-\CB^{(1)}\|\leq \\ \leq \eps_1 + \|\Pi_{B_k\backslash B_{jk}}(\CA^{(1)}-\CA^{(2)}) \| \leq 2\eps_1 + \eps_2.
\end{multline}
\end{proof}

\section{Brick expansion} \label{app:bricks}

A brick in $\RR^d$ is a subset of $\RR^d$ of the form $\{(x_1,\ldots,x_d): n_i\leq x_i < m_i,\ i=1,\ldots, d\}$, where $n_i$ and $m_i$ are integers satisfying  $n_i<m_i$. The empty subset is also regarded as a brick. A unit brick is a brick with $m_i=n_i+1$ for all $i$. The intersection of any two bricks is a brick. The set of all bricks in $\RR^d$ (including the empty set) is denoted  $\Br_d$. It is a poset with a partial order given by inclusion. This poset has a lower bound (the empty set) and is locally finite (i.e. for any two bricks $Y,Y'\in\Br_d$ the set $\{Z\in\Br_d, Y\leq Z\leq Y'\}$ is finite). 
%To each finite subset of $\Lambda$ we assign the intersection of all bricks which contain it. This gives a map ${\mathcal S}: {\rm Fin}(\Lambda)\ra \Br_d$. 

In Section 2, we defined the subspace $\mfkd^Y\subset\mfkd_Y$ as an orthogonal complement of 
\beq
\sum_{\substack{Z\in\Br_d\\ Z\subsetneq Y}}\mfkd_Z 
\eeq
with respect to the inner product $\langle\CA,\CB\rangle=\langle\CA^*\CB\rangle_\infty$. One can give a more explicit description of this subspace using a {\it Pauli basis} of $\SAl$. This is a basis obtained by choosing an orthonormal self-adjoint basis $\CE_j^k$, $k=0,\ldots,d_j^2-1$ for each $\SA_j$ (where the inner product is the same as above) so that $\CE^0_j$ is the identity element in $\SA_j$. The resulting basis elements of $\SAl$ can be labeled by functions $\nu:\Lambda\ra\NN_0$ with $\nu(j)<d^2_j$ which vanish outside of a finite set. The identity element in $\SAl$ corresponds to $\nu$ being identically zero. If we denote by $\supp(\nu) \in {\rm Fin}(\Lambda)$ the support of $\nu$, then a basis for $\mfkd_Y$ consists of those $\CE_\nu$ for which $\supp(\nu)$ is nonempty and $\supp(\nu)\subseteq Y$. A basis for $\mfkd^Y$ consists of those $\CE_\nu$ for which, in addition, $\supp(\nu)\nsubseteq Z$ for any brick $Z\subsetneq Y$. Note that since $\bD<1/2$, every unit brick (and therefore also every nonempty brick) contains at least one point of $\Lambda$. Therefore $\mfkd^Y$ is nonzero for every nonempty brick $Y$.
%for any $\CA \in \SAal$ and any $Y\in \Br_d$ we have
%\beq 
%\CA^{Y} := \sum_{\nu:\, {\mathcal S}(\supp(\nu)) = Y} \lal \CA \CE_{\nu} \ral_{\infty} \CE_{\nu}.
%\eeq
%On the other hand, the restriction of $\CA$ to $Y\in \Br_d$ can be written as
%\beq
%\CA\vert_Y=\sum_{\nu:\, {\mathcal S}(\supp(\nu)) \subseteq Y} \lal \CA \CE_{\nu} \ral_{\tau} \CE_{\nu}.
%\eeq
For any $\CA\in\mfkdl$ we have
\beq\label{eq:Yexpansion}
\CA\vert_Y=\sum_{Y'\in \Br_d, Y'\subseteq Y} \CA^{Y'} .
\eeq

Recall \cite{Jacobson} that to any locally finite poset $P$ which is bounded from below one can attach its \Mobius\ function $\mu_P: P\times P\ra\ZZ$. This function has the following property. Let  $f:P\ra V$ be any function with values in a vector space $V$. Let us define another function $g:P\ra V$ by $g(y)=\sum_{z\leq y} f(z)$. Then $f$ can be expressed through $g$ by
\beq
f(y)=\sum_{z\leq y} \mu_P(z,y) g(z).
\eeq
The function $\mu_P(z,y)$ is uniquely defined by this property if we demand $\mu(z,y)=0$ for $z\nleq y$. The \Mobius\ functions is multiplicative under Cartesian product: if $P,Q$ are locally finite posets bounded from below, and $P\times Q$ is given the obvious partial order, then
$\mu_{P\times Q}=\mu_P\cdot \mu_Q.$

\begin{prop}\label{lma:locmobiusestimate}
For any $\CA\in\mfkdl$ we have $\lVert \CA^Y\rVert\leq 4^d \lVert\CA\rVert$.
\end{prop}
\begin{proof}
Note first that the poset  $\Br_d$ is the Cartesian product of $d$ copies of $\Br_1$. The \Mobius\ function of $\Br_1$ is easily computed:
\beq
\mu_{\Br_1}([n',m'),[n,m)) = \begin{cases}
(-1)^{(n'-n)+(m-m')},\, \text{if } (n'-n),(m-m') \in \{0,1\},\\
0,\quad \text{otherwise}.
\end{cases}
\eeq
% $\mu_{\Br_1}([n',m'),[n,m))=(-1)^{(n'-n)+(m-m')}$ if $n'-n\in\{0,1\}$ and $m-m'\in \{0,1\}$ and vanishes otherwise. 
Therefore $\mu_{\Br_d}(Y',Y)\in \{0,1,-1\}$ for all $Y',Y\in \Br_d$ and is nonzero if and only if the integers $n_i,m_i$ and $n'_i,m'_i$ defining $Y$ and $Y'$ satisfy $m_i-m'_i\in \{0,1\}$ and $n'_i-n_i\in \{0,1\}$ for all $i$.
Applying the inversion formula to (\ref{eq:Yexpansion}) we get  
\beq \label{eq:mobiuslemma}
\CA^Y=\sum_{Y'\in \Br_d, Y'\subseteq Y} \mu_d(Y',Y)\CA\vert_{Y'},
\eeq
Since the sum on the r.h.s. contains exactly $4^d$ terms, and since $\lVert\CA\vert_Y\rVert\leq \lVert\CA\rVert$, we get the desired estimate.
\end{proof}

Similarly, if for any $\CA\in\mfkdal$ we define $\CA^Y$ to be the $\mfkd^Y$ component of $\CA\vert_Y\in\mfkd_Y$, then we have the following estimate.

\begin{prop}
\label{lma:bricklemma}
For any $a\in\Orfm$ and any $\CA \in \mfkdal$ which is $a$-localized at $j\in\Lambda$ we have 
\beq \label{eq:brickestimate}
\| \CA^Y \| \leq b(\diam(\{j\}\cup Y))
\eeq
where 
\beq
b(r)=2^{2d+1} a(\max(0,r/(2\sqrt d)-2)).
\eeq
\end{prop}
\begin{proof}
% Let 
% \beq
% \kappa_d (f)(r)=4^d f(\max(0,r/(2\sqrt d)-2)).
% \eeq
% It is easy to check that $\kappa_d$ is a continuous map $\Orf\ra\Orf$ which maps $\Orfm$ to itself. We will now show that it satisfies the remaining condition.

Let $J=(J_1,\ldots,J_d)$ be a point of $\ZZ^d$ such that $|J-j|< 1/2$ (such a point exists because we normalized the metric so that $\bD<1/2$). Suppose a nonempty brick $Y\in \Br_d$ is defined by integers $n_i<m_i$, $i=1,\ldots,d$. Let $K=\text{max}\{|m_1-J_1|,...,|m_d-J_d|,|n_1-J_1|,...,|n_d-J_d|\}$. Let $Z$ be the brick $[J_1-K,J_1+K)\times\ldots \times [J_d-K,J_d+K)$ and let  $Z'=[J_1-K+1,J_1+K-1)\times\ldots \times [J_d-K+1,J_d+K-1)$. Obviously, $Z$ contains both $Y$ and $Z'$. It is also easy to see that $B_j(r)\subset Z$ for all $r< K-1/2$, and $B_j(r)\subset Z'$ for all $r< K-3/2$. Note also that $Y$ contains a unit brick which does not intersect $Z'$, and thus $Y$ contains at least one point of $\Lambda$ which is not in $Z'$.

Since ${\rm diam}(\{j\} \cup Y)\leq 2K \sqrt d,$ for $K< 2$ we get ${\rm diam}(\{j\} \cup Y)<4\sqrt d,$ and thus $b({\rm diam}(\{j\} \cup Y))=2^{2d+1} a(0)$. Note also that $\|\CA\vert_Y\|\leq \|\CA\|\leq 2 a(0)$. Therefore by Prop. \ref{lma:locmobiusestimate} we have  $\|\CA^Y\|\leq 2^{2d+1} a(0)$ and the condition (\ref{eq:brickestimate}) is satisfied.

It remains to consider the case $K\geq 2$.
For any $r>0$ let $\CB^{(r)}$ be a best possible approximation of $\CA$ on a ball $B_j(r)$. Since $B_j(K-2)\subset Z'$, $\CB^{(K-2)}\vert_Y$ is supported on the brick $Y\cap Z'$. Since $Y$ contains points of $\Lambda$ which are not in $Y\cap Z'$, $\left(\CB^{(K-2)}\right)^Y=0$ and thus  
$\CA^Y=(\CA-\CB^{(K-2)})^Y$. Therefore by Prop.  \ref{lma:locmobiusestimate} we get $\|\CA^Y\|\leq 2^{2d+1} a(K-2)$. Since ${\rm diam}(\{j\} \cup Y)\leq 2K \sqrt d,$ the condition (\ref{eq:brickestimate}) is satisfied.
\end{proof}

\begin{corollary}
For any $\CA\in \SAal$ the sum $\sum_{Y \in \Br_d} \CA^Y$ \Frechet-converges to $\CA$.
\end{corollary}
\begin{proof}
By the above lemma, $\|\CA^Y\|_{j,\alpha}\leq (1+r)^\alpha b(r),$ where $r={\rm diam}(\{j\}\cup Y)$ and $b(r) \in \Orfm$. This proves convergence. By eq. (\ref{eq:Yexpansion}), the sum is $\CA$. 
\end{proof}

\section{Chains} \label{app:chains}

\subsection{Other families of norms} \label{app:normfamilies}

In the body of the paper $\mfkdal$ was defined as a completion of $\mfkdl$ with respect to a family of norms $\|\cdot\|_{j,\alpha}$ where $j\in\Lambda$ was fixed and $\alpha$ ranged over $\NN_0$ . Sometimes it is useful to consider two other natural families of norms on $\mfkdl$ labeled by the same data: 
\beq
\|\CA\|^{cev}_{j,\alpha} := \sup_{r} (1+r)^{\alpha} \|\CA-\CA|_{B_j(r)}\|.
\eeq
\beq
\|\CA\|^{br}_{j,\alpha} := \sup_{Y \in \mathbb{B}_d} (1+\diam(\{j\} \cup Y))^{\alpha} \|\CA^Y\|.
\eeq
For a fixed $j$, all three family of norms are non-decreasing with $\alpha$. 
\begin{prop} \label{prop:normequivalence}
The families of norms $\{\|\cdot\|_{j,\alpha}\}$, $\{\|\cdot\|^{cev.}_{j,\alpha}\}$, and $\{\|\cdot\|^{br}_{j,\alpha}\}$ define the same topology on $\mfkdl$, and thus give rise to the same completion $\mfkdal$. 
\end{prop}
\begin{proof}
To show the equivalence of two non-decreasing families of norms, we need to show that each norm from the first family is upper-bounded by a multiple of a norm from the second family, and vice versa. 

It follows directly from the definition that $\|\CA\|_{j,\alpha}\leq \|\CA\|^{cev}_{j,\alpha}$ for all $\alpha\in\NN_0$ and all $\CA\in\mfkdl$.

Prop. \ref{lma:bricklemma} implies
\beq
\|\CA^Y\|\leq 2^{2d+1} f_j\left(\CA,{\rm max}\left(0,\frac{1}{2 \sqrt{d}} \diam(\{j\} \cup Y)-2\right)\right).
\eeq 
Hence $\|\CA\|^{br}_{j,\alpha}\leq  C_{\alpha}\|\CA\|_{j,\alpha}$ for some $C_{\alpha}>0$.

%Let us fix $j$ and $Y \in \mathbb{B}_d$. There is a ball $B_{j}(r)$ of radius $r>\frac13 \diam(\{j\} \cup Y)-2$, such that for any observable $\CA \in \SA_{B_{j}(r)}$ we have $\CA^Y = 0$. Therefore by Lemma \ref{lma:locmobiusestimate} for any $\CA \in \mfkdal$ we have
%\beq
%\|\CA^Y\| \leq 4^d f_j(\CA,r).
%\eeq
%Hence
%\beq \label{eq:bricknormbound}
%\|\CA\|^{br.}_{j,\alpha} \leq 4^d C_{\alpha} %\|\CA\|_{j,\alpha}
%\eeq
%for some constant $C_{\alpha}$ that depends on $\alpha$ only, that implies the dominance of (1) over (3).

Finally, let us show that $\|\CA\|^{cev}_{j,\alpha}$ is upper-bounded by a multiple of $\|\CA\|^{br}_{j,\alpha+2d+1}$. 
%Let $Z_r$ be a maximal brick that is contained inside $B_j(r)$. 
For any observable $\CA \in \mfkdl$ and any brick $Y\subset B_j(r)$ we have $\CA^Y\vert_{B_j(r)}=\CA^Y$. Therefore 
\beq
\CA-\CA\vert_{B_j(r)}=\sum_{Y\nsubseteq B_j(r)}\left(\CA^Y-\CA^Y\vert_{B_j(r)}\right),
\eeq
and thus
\beq
\| \CA - \CA|_{B_j(r)} \| \leq 2\sum_{Y \nsubseteq B_j(r)} \|\CA^Y\|.
\eeq
Since for any $Y\nsubseteq B_j(r)$ we have $\diam(\{j\} \cup Y)\geq r$, we get
\beq
\|\CA\|^{cev}_{j,\alpha} \leq 2 C_d  \|\CA\|^{br}_{j,\alpha+2d+1},
\eeq
where we have used
\beq
\sum_{Y  \in \Br_d} (1+\diam(Y \cup \{j\}))^{-(2d+1)} \leq C_d
\eeq
for some constant $C_d$ that depends on $d$ only.
\end{proof}

Similarly, in addition to $\{\|\cdot\|_{\alpha}\}$ on $C_q(\mfkdal)$ we can introduce families of norms:
\beq 
\|\cha\|^{cev}_{\alpha} = \sup_{a \in \{0,1,...,q\}}\sup_{j_0,...,j_q \in \Lambda} \| \cha_{j_0...j_q}\|^{cev}_{j_a,\alpha},
\eeq
\beq \label{eq:bricknorms}
\|\cha\|^{br}_{\alpha} = \sup_{a \in \{0,1,...,q\}}\sup_{j_0,...,j_q \in \Lambda} \| \cha_{j_0...j_q}\|^{br}_{j_a,\alpha}.
\eeq
Prop. \ref{prop:normequivalence} implies that all these families of norms are equivalent with the following dominance relations
\beq \label{eq:dominancerelations}
\|\cha\|_{\alpha} \leq \|\cha\|^{cev}_{\alpha},\quad \|\cha\|^{cev}_{\alpha} \leq C \|\cha\|^{br}_{\alpha+2d+1},\quad
\|\cha\|^{br}_{\alpha} \leq C_{\alpha} \|\cha\|_{\alpha}.
\eeq

% \begin{lemma}
% The families of norms (1) $\{\|\cdot\|_{\alpha}\}$, (2) $\{\|\cdot\|^{cev.}_{\alpha}\}$, (3) $\{\|\cdot\|^{br.}_{\alpha}\}$ are all equivalent.
% \end{lemma}

% \begin{proof}
% It follows directly from the definition . 

% On the other hand for any observable $\CA \in \mfkdal$ we have
% \beq
% f_j(\CA,r) \leq \sum_{Y \nsubseteq Z_r} \|\CA^Y\|,
% \eeq
% where $Z_r$ is a maximal brick that is contained inside $B_j(r)$, that implies the dominance of (2) over (1).
% \end{proof}

% For any $\chF\in C_0(\SAal)$ and any $\CA\in\SAal$ we defined
% \beq
% \delta_\chF(\CA)=\sum_j [\chF_j,\CA].
% \eeq
% The infinite sum on the r.h.s. converges in norm  because
% $$
% \lVert [\chF_j,\CA]\rVert \leq 2 \lVert \chF\rVert_0 f_{j_0}(\CA,r/2)+2 f_j(\chF_j,r/2)\lVert \CA\rVert + 2 f_{j_0}(\CA,r/2) f_j(\chF_j,r/2),
% $$
% where $j_0\in\Lambda$ is a fixed site and $r=|j-j_0|$. In fact, a stronger statement holds.

\subsection{Continuity of chain maps} \label{app:complexcontinuity}

\begin{prop} \label{prop:partialContinuity}
The boundary operator $\p_q: C_{q}(\mfkdal) \to C_{q-1}(\mfkdal)$ for $q \geq 0$ is well defined and continuous.
\end{prop}
\begin{proof}
For $q=0$, let $\cha \in C_0(\mfkdal)$. Then
\beq
\|\p \cha\|^{br}_{\alpha} \leq \sup_{Y \in \mathbb{B}_d} \sum_{j \in \Lambda} (1+\diam(Y \cup \{j\}))^{\alpha}  \| \cha^{Y}_j\| \leq C \|\cha\|^{br.}_{\alpha+d+1}
\eeq
where we have used
\beq
\sum_{j  \in \Lambda} (1+\diam(Y \cup \{j\}))^{-(d+1)} \leq C
\eeq
for some constant $C$ that depends on the lattice only.

Similarly, for $q>0$, let $\cha \in C_q(\mfkdal)$. Then
\begin{multline} \label{eq:partialNormEstimate}
\|\p \cha\|^{br}_{\alpha} \leq \\ \leq \sup_{Y \in \mathbb{B}_d} \sup_{a \in \{1,...,q\}} \sup_{j_1,...,j_q \in \Lambda} \sum_{j_0  \in \Lambda} (1+\diam(Y \cup \{j_a\}))^{\alpha}  \| \cha^{Y}_{j_0...j_q}\| \leq \\ \leq C' \|\cha\|^{br}_{\alpha+d+1}.
\end{multline}
for some constant $C'$ that depends on the lattice only. Thus, the map $\p_q$ is well-defined and continuous for any $q \geq 0$. 
\end{proof}

Recall that on UL chains we have a map $h_q : C_q (\mfkdl) \to C_{q+1} (\mfkdl)$ for $q \geq -1$ defined by
\beq
h_q(\cha)_{j_0\ldots j_{q+1}}=\sum_{Y\in \Br_d}\sum_{k=0}^{q+1} (-1)^k \frac{\chi_Y(j_k)}{|Y\cap\Lambda|} \cha^{Y}_{j_0\ldots \widehat j_k \ldots j_{q+1}}
\eeq
for $q \geq 0$ and
\beq
h_{-1}(\derA)_{j_0}=\sum_{Y\in \Br_d} \frac{\chi_Y(j_0)}{|Y\cap\Lambda|} \derA^{Y}
\eeq
for $q=-1$, that gives a contracting homotopy for the augmented complex $C_{\bullet}(\mfkdl) \to \mfkDl$, i.e. $h_{q-1}\circ\partial_q +\partial_{q+1}\circ h_q={\rm id}$ and $\p_0 \circ h_{-1} = {\rm id}$.

\begin{prop}  \label{prop:contrHoContinuity}
The map $h_q$ extends to a continuous linear map $h_q : C_q (\mfkdal) \to C_{q+1} (\mfkdal)$, that gives a contracting homotopy for the augmented complex $C_{\bullet}(\mfkdal) \to \mfkDal$.
\end{prop}

\begin{proof}
For $\cha \in C_{q}(\mfkdal)$ with $q\geq 0$ we have
\beq
\|h_q(\cha)\|^{br}_{\alpha} \leq (q+2) \|\cha\|^{br}_{\alpha}
\eeq
Similarly, for $\derA \in \mfkDal$
\beq
\|h_{-1}(\derA)\|^{br}_{\alpha} \leq \|\derA\|^{br}_{\alpha}.
\eeq
\end{proof}

\begin{prop} \label{prop:bracketContinuity}
The bracket $\{ \cdot, \cdot \}: C_{p}(\mfkdal) \times C_{q}(\mfkdal) \to C_{p+q+1}(\mfkdal)$ for $p,q \geq -1$ is well defined and jointly continuous.
\end{prop}

\begin{proof}
Let $\CA,\CB \in \mfkdal$ and $j,k \in \Lambda$. By the definition eq. (\ref{eq:fjAr}) we can choose $\CA_n \in \SA_{B_j(n)}$, such that $\|\CA-\CA_n\| \leq f_{j}(\CA,n)$. Then for $\CA^{(n)} = \CA_{n+1}-\CA_{n}$ we have 
\beq 
\CA=\sum_{n\in\NN_0} \CA^{(n)},
\eeq
with $\CA^{(n)}\in \SA_{B_{j}(n+1)}$ and $\| \sum_{m \geq n}^\infty\CA^{(m)}\| \leq f_{j}(\CA,n)$. Similarly we define $\CB^{(n)}$ so that $\| \sum_{m\geq n}^{\infty} \CB^{(m)}\| \leq f_k(\CB,n)$.
Let $\mathcal{C}^{(n,m)}=[\CA^{(n)},\CB^{(m)}]$. Clearly, $\mathcal{C}^{(n,m)}=0$ if $|j-k| \geq n+m+2.$
On the other hand, for $|j-k| < n+m+2$ the observable $\mathcal{C}^{(n,m)}$ is localized on a ball of radius $2(m+n+2)$ centered at $j$ and has a norm bounded from above by $8 f_j(\CA,n) f_{k}(\CB,m)$. Then we have an estimate
% for $|j-k|\leq n+m+2$ we have an  estimate
\beq
\sup_r (1+r)^\alpha f_{j}(\mathcal{C}^{(n,m)},r)\leq 8(1+2(m+n+2))^\alpha f_j(\CA,n) f_k(\CB,m). 
\eeq
Hence
\beq
\|[\CA,\CB]\|_{j,\alpha}\leq 8 \sum_{n,m \in \NN_0} (1+2(m+n+2))^\alpha f_j(\CA,n) f_{k}(\CB,m)
\eeq
Moreover, for any $\alpha \geq 0$ we have an estimate
\beq
\sum_{n\in\NN_0} (n+1)^{\alpha} f_{j}(\CA,n)\leq \|\CA\|_{j,\alpha+2}\sum_{n\in\NN} 1/(n+1)^2.
\eeq
Therefore 
\beq
\|[\CA,\CB]\|_{j,\alpha} \leq C_{\alpha} \|\CA\|_{j,\alpha+2} \|\CB\|_{k,\alpha+2} 
\eeq
where $C_{\alpha}$ is some constant depending on $\alpha$ only.

Similarly, for any $\cha \in C_{p\geq0}(\mfkdal)$ and $\chb \in C_{q \geq0}(\mfkdal)$ we have
\beq  \label{eq:bracketNormEstimate}
\|\{\cha,\chb \} \|_{\alpha} \leq C_{\alpha,p,q} \|\cha\|_{\alpha+2} \|\chb\|_{\alpha+2}.
\eeq
for some constant $C_{\alpha,p,q}$.

For $\derA \in \mfkDal$ and $\chb \in C_q(\mfkdal)$ using eq. (\ref{eq:partialNormEstimate}) we have
\beq\label{eq:bracketnormderivation}
\|\{ \derA,\chb \} \|^{br}_{\alpha} = \|[\p( h_{-1}(\derA)),\chb]\|^{br}_{\alpha} \leq C_{\alpha,q} \| \derA \|^{br}_{\alpha+d+3} \| \chb \|_{\alpha+d+3}
\eeq
for some constant $C_{\alpha,q}$, where $h_{-1}$ is a contracting homotopy from Proposition \ref{prop:contrHoContinuity}.

Finally, for $\derA, \derB \in \mfkDal$ we have
\beq
\|\{ \derA,\derB \} \|^{br}_{\alpha} = \|\p \{ \derA,h_{-1}(\derB) \} \|^{br}_{\alpha} \leq C_{\alpha} \| \derA \|^{br}_{\alpha+2d+4} \| \derB \|^{br}_{\alpha+2d+4}.
\eeq
Thus the map $\{\cdot,\cdot\}$ is well-defined and continuous.

% First, let us assume $q,p \geq 0$. Let $\cha \in C_{q}(\mfkdal)$. We can choose approximating sequence $\cha^{(n)} \in C_{q}(\mfkdl)$ of range $n$ such that $\| \sum_{n \geq k} \cha^{(n)}_{j_0...j_q} \| \leq f(\cha_{j_0...j_q},k)$. Similarly, let $\chb \in C_{p}(\mfkdal)$ with an approximating sequence $\chb^{(n)} \in C_{p}(\mfkdl)$. 

% Let $\CB^{(n,m)}_{j_0...j_{p+q+1}} = [\cha^{(n)}_{j_0...j_q},\chb^{(m)}_{j_{q+1}...j_{p+q+1}}]$. 
% $\max_{a\leq q, b >q}|j_a-j_b|\geq n+m$ we have $\CB^{(n,m)}_{j_0...j_{p+q+1}}=0$, while for $\max_{a\leq q, b >q}|j_a-j_b| < n+m$ we have the following estimate
% \beq
% \sup_r (1+r)^{\alpha} f(\CB^{(n,m)}_{j_0...j_{p+q+1}},r) \leq 8 (1+2n+2m)^{\alpha} f(\cha,n) f(\chb,m)
% \eeq
% Hence for $\chc = \{\cha,\chb\}$ we have
% \beq
% \|\chc\|_{\alpha} \leq C \sum_{n,m \in \NN} (...)^d(1+2n+2m)^{\alpha} f(\cha,n) f(\chb,m)
% \eeq
% where $C$ is some constant that depends on $p,q,d$ and the lattice only.
\end{proof}

% \begin{corollary} \label{prop:bracketContinuity}
% The bracket $\{ \cdot, \cdot \}: \mfkDal \times C_{q}(\mfkdal) \to C_{q}(\mfkdal)$ for $q \geq -1$ is well defined, continuous and smooth.
% \end{corollary}

% \begin{proof}
% Let $h_{-1}$ be a contracting homotopy from Proposition \ref{prop:contrHoContinuity}.
% For $q \geq 0$, let $\cha \in C_q(\mfkdal)$ and $\chf \in C_0(\mfkdal)$. Using eq. (\ref{eq:partialNormEstimate}), eq. (\ref{eq:bracketNormEstimate}) and norm estimates from the proof of Proposition \ref{prop:normequivalence} we have
% \beq

% \eeq
% \end{proof}

\begin{prop}    \label{prop:contractionContinuity}
The contraction $\chb_{A_0...A_q} \in \mfkDal$ of $\chb \in C_q(\mfkdal)$ with regions $A_0,...,A_q$ is a well-defined and continuous.
\end{prop}
\begin{proof}
By the same argument as in the proof of the Proposition \ref{prop:partialContinuity}, we have
\beq
\|\chb_{A_0}\|^{br.}_{\alpha} \leq C \|\chb\|^{br.}_{\alpha+d+1}
\eeq
where $\chb_{A_0} \in C_{q-1}(\mfkdal)$ is a (possibly partial) contraction defined by 
\beq
(\chb_{A_0})^Y_{j_1...j_q} := \sum_{j_0 \in A_0} \chb^Y_{j_0...j_q}. 
\eeq
Therefore
\beq
\|\chb_{A_0...A_q}\|^{br.}_{\alpha} \leq C^{q+1} \|\chb\|^{br.}_{\alpha+(q+1)(d+1)}.
\eeq
\end{proof}

\begin{prop} \label{prop:contractionConical}
For any conical partition $(A_0,...,A_d)$ of $\RR^d$ the contraction $(\cdot)_{A_0...A_d}$ is en element of $\mfkdal$ and defines a linear continuous map $(\cdot)_{A_0...A_d}: C_d(\mfkdal) \to \mfkdal$.
\end{prop}

\begin{proof}
Without loss of generality, we can assume that all $A_a$ are conical regions, since otherwise the map $(\cdot)_{A_0...A_d}$ differs from a contraction with a conical partition by a manifestly continuous linear map.

Let $p \in \RR^d$ be the apex of $(A_0,...,A_d)$. Note that the number of tuples $\{j_0,...,j_d\}$ such that $j_a \in A_a$ and $|j_a-p|\leq R$ is less than $C R^{d(d+1)}$ for some constant $C$. Also not that when at least one $j_a$ belongs to $B^c_p(R)$ we have $\diam(\{j_0,...,j_q\}) \leq C' R$ for some constant $C'$. In the latter case by Lemma \ref{lma:chainapprox} we have $\|\chb_{j_0...j_d}\| \leq 3 f(\chb,C'R/2)$ for any $\chb \in C_d(\mfkdal)$.

Hence for any $\chb \in C_d(\mfkdal)$ we have
\begin{multline}
f_p(\chb_{A_0...,A_d},r) \leq C (r/2)^{d(d+1)} f(\chb,r/2) + \\ + \sum_{n=0}^{\infty} C'' (1+n+r/2)^{d(d+1)} f(\chb,C'(n+r/2)/2)
\end{multline}
for some constant $C''$, that implies 
\beq
\|\chb_{A_0...,A_d}\|_{p,\alpha} \leq C_{\alpha} \|\chb\|_{\alpha + d(d+1) + 2}.
\eeq
\end{proof}

% \begin{lemma} \label{lma:contmapapproximation}
% Let $h:C_{q}(\mfkdal) \to C_{p}(\mfkdal)$ be a continuous linear map, and let $\Gamma_r = B_k(r)$ for some $k \in \Lambda$. Then for any $\cha \in C_{q}(\mfkdal)$ and any $j_0,...,j_p \in \Gamma_r \cap \Lambda$ we have $\|h(\cha|_{\Gamma_{2r}})_{j_0...j_p} - h(\cha)_{j_0...j_p}\| = \Or$.
% \end{lemma}

% \begin{proof}

% \end{proof}

\section{Some consequences of the Lieb-Robinson bound} 

\subsection{Reproducing functions} \label{app:reproducing}

We say that $f:\RR_{\geq 0} \to \RR_{\geq 0}$ is reproducing for $\Lambda$, if
\beq
C_f := \sup_{j,k \in \Lambda} \sum_{l \in \Lambda} \frac{f(|j-l|)f(|l-k|)}{f(|j-k|)} < \infty.
\eeq
Note that $1/(1+r)^{\nu}$ is reproducing for any $\Lambda \subset \RR^d$ if $\nu > d$, but not every $f \in \Orfm$ is reproducing.

\begin{lemma} \label{lma:weakreproducing}
For any $f \in \Orfm$ there  $\tilde f \in \Orfm$ that  upper-bounds $f$ and $A>0$ such that $\tilde f(r) \tilde f(s)/\tilde f(r+s) \leq A$ for all $r, s \in \RR_{\geq 0}$.
\end{lemma}

\begin{proof}
Without loss of generality we can assume $f(0)=1/2$.

Let $h(r) := - (\log f(r))$. This is a monotonically increasing positive function. If $h(r)\geq Cr+D$ for some $C>0$, the function $\tilde f(r)=e^{-Cr-D}$ satisfies the required conditions (with $A=e^{-D}$). Otherwise, 
let $\tilde{h}(r) = r \inf_{0\leq s \leq r}(h(s)/s)$ and $\tilde{f} = e^{-\tilde{h}(r)}$. It is easy to check that $\tilde f$ satisfies the required conditions.
\end{proof}

\begin{lemma} \label{lma:Lambdareproducing}
Any $f \in \Orfm$ can be upper-bounded by $\tilde{f} \in \Orfm$ which is reproducing for $\Lambda$.
\end{lemma}

\begin{proof} 
By Lemma \ref{lma:weakreproducing} $f$ can be upper-bounded by $f' \in \Orfm$ such that for some $A>0$ we have $f'(r) f'(s) \leq A f'(r+s)$ for any $r, s \in \RR_{\geq 0}$. Let $\tilde{f}(r) = B \sqrt{f'(r)}/(1+r)^{d+1}$ with $B =\sqrt{\|f'\|_{2d+2}}$. Then $\tilde{f}(r) \geq f'(r)$ for all $r\geq 0$ and
\begin{multline}
\sum_{l \in \Lambda} \tilde{f}(|j-l|) \tilde{f}(|l-k|)  = B^2 \sum_{l \in \Lambda} \frac{\sqrt{f'(|j-l|)}}{(1+|j-l|)^{d+1}} \frac{\sqrt{f'(|l-k|)}}{(1+|l-k|)^{d+1}} \leq \\ \leq B^2 \sum_{l \in \Lambda} \frac{A^{1/2} \sqrt{f'(|j-k|)}}{(1+|j-l|)^{d+1}(1+|l-k|)^{d+1}} \leq \\ \leq C' A^{1/2} B^2 \frac{\sqrt{f'(|j-k|)}}{(1+|j-k|)^{d+1}}  = C' A^{1/2} B \tilde{f}(|j-k|)
\end{multline}
where $C'$ is some constant that depends on the lattice $\Lambda$ only.
\end{proof}

\begin{lemma} \label{lma:grepboundsfn}
For any sequence $\{f_n\},\,n \in \NN$ of functions $f_n \in \Orfm$ converging to $f \in \Orfm$, there is a reproducing for $\Lambda$ function $g \in \Orfm$ that upper-bounds $f$ and $f_n$ for any $n \in \NN$.
\end{lemma}

\begin{proof}
% If that is true for $f=0$, then for $f \neq 0$ we can 
% It is enough to show that for $f=0$. Indeed, if $f \neq 0$ there is a sequence $h_n \geq |f-f_n|$ that converges to $0$.
Let $\tilde{g}(r) := \sup_{n \in \NN} f_n(r)$. We have an estimate
\beq
\|\tilde g\|_\alpha\leq \sup_{n\in\NN}\|f_n-f\|_\alpha+\|f\|_\alpha
\eeq
Since $f_n$ converges to $f$, this implies $\tilde g\in\Orf$ and thus can be upper-bounded by a function from $\Orfm$. 
%Since the sequence $f_n$ converges, the supremum is achieved for some $n=n_r$ and therefore $\|\tilde{g}\|_{\alpha} = \|f_{n_{\alpha}}\|_{\alpha} < \infty$ for some $n_{\alpha}$. Thus $\tilde{g} \in \Orf$ and can be upper-bounded by a function from $\Orfm$. 
Hence by Lemma \ref{lma:Lambdareproducing} there exists $g \in \Orfm$ which is reproducing and upper-bounds $\tilde{g}$ and $f$.
\end{proof}

\subsection{Locally generated automorphisms}\label{app:LRbound}

\begin{prop} \label{lma:alphaODE}
For any $\derG \in C([0,1],\mfkDal)$ there is a family of automorphisms $\alpha_{\derG}:[0,1] \to \Aut(\SA)$ such that $\forall\CA \in \SAal$ and  $\forall s\in [0,1]$ we have $\alpha_{\derG}(s)(\CA) \in \SAal$ and the function $\alpha_{\derG}(\CA):[0,1] \to \SAal, s\mapsto \alpha_\derG(s)(\CA)$ is continuously differentiable and satisfies
\beq\label{eq:alphaGdiffequation}
\frac{d \alpha_{\derG}(s)(\CA)}{ds} = \alpha_{\derG}(s)(\derG(s)(\CA)).
\eeq
\end{prop}
\begin{proof}
To show this we invoke the version of the Lieb-Robinson bound from \cite{HastingsKoma,nachtergaele2006propagation,nachtergaelesimsyoung} for an interaction defined in terms of $\derG^Y$, $Y \in \Br_d$.

Let $h\in \Orfm$ be the function $h(r)=\sup_s\sup_{Y: \diam(Y)\geq r} \|\derG^{Y}(s)\|$. We can choose $g \in \Orfm$ (e.g. we can take $A (h(r))^{\alpha}$ for some constants $A$ and $0<\alpha<1$) such that
\beq \label{eq:Gbound}
\sup_{j,k \in \Lambda}\sup_s \sum_{\substack{Y \in \mathbb{B}_d\\ Y \ni j,k}} \frac{\|\derG^Y(s)\|}{g(|j-k|)} \leq 1.
\eeq 
Moreover, by Lemma \ref{lma:Lambdareproducing} $g$ can be chosen to be reproducing, with a reproducing constant $C_g>0$.

% for any $s \in \RR$ and for some reproducing\footnote{See Appendix \ref{app:reproducing} for a definition.} function $g \in \Orfm$. Indeed, for any $\derG \in C(\RR,\mfkDal)$, for which we have $\|\derG^Y(s)\| \leq h(\diam(Y))$ for some $h \in \Orfm$, we can always choose a function $g \in \Orfm$ such that the above equation holds, and as explained in Appendix \ref{app:reproducing} without loss of generality we can assume that $g(r)$ is reproducing. 

Let $\{\Gamma_n\}$ be an exhausting sequence of bricks. Let $\derG^{(n)} = \derG|_{\Gamma_n}$. It is easy to see that for any $\CA\in\SAl$ the sequence of $\SAal$-valued functions $\{\derG^{(n)}(s)(\CA)\}$ converges to $\derG(s)(\CA)$ in the \Frechet\ topology of $C([0,1],\SAal)$.

Let $\alpha_{\derG^{(n)}}$ be the automorphism of $\SA_{\Gamma_n}$ defined by
\beq\label{eq:alphaGndiffequation}
\frac{d \alpha_{\derG^{(n)}}(s)(\CA)}{ds} = \alpha_{\derG^{(n)}}(s)(\derG^{(n)}(s)(\CA)).
\eeq
and the initial condition $\alpha_{\derG^{(n)}}(0)={\rm id}$. It is well-known that such an automorphism exists and is unique (it is the holonomy of the parallel transport with respect to the connection $\frac{d}{ds}+\derG^{(n)}$ on a trivial bundle with fiber $\SA_{\Gamma_n}$). We extend $\alpha_{\derG^{(n)}}$ to the whole $\SAl$ in the obvious way. Theorem 2.1 from \cite{bachmann2012automorphic} (or more precisely, its version for time-dependent interactions from \cite{nachtergaelesimsyoung})  then guarantees the following estimate
%\footnote{Though in \cite{nachtergaele2006propagation} the interactions are assumed to be independent of the parameter $s$, the proof of the estimate works in exactly the same way for derivations depending on $s$.}
\beq \label{eq:LRbestimate2}
\frac{\|[\alpha_{\derG^{(n)}}(s)(\CA),\CB]\|}{\|\CA\| \|\CB\|} \leq  \frac{2}{C_g}(e^{2 C_g s}-\theta(R-r)) \sum_{\substack{j \in B_j(r)\\ k \in B^c_j(R)}} g(|j-k|) =: 2 h(r,R)
\eeq
for any $\CA \in \SA_{B_j(r)}$ and $\CB \in \SA_{B^c_j(R)}$, while Theorem 2.2 in \cite{nachtergaele2006propagation} proves the existence of the limit $\lim_{n \to \infty} \alpha_{\derG^{(n)}}(s) =: \alpha_{\derG}(s) \in \Aut(\SA)$ satisfying the same estimate. We have an estimate
\begin{multline} \label{eq:LRbestimate}
f_j(\alpha_{\derG}(s)(\CA),r) \leq f_j(\alpha_{\derG}(s)(\CA^{(r/2)}),r) + \| \CA- \CA^{(r/2)}\| \leq \\ 
\leq \| \alpha_{\derG}(s)(\CA^{(r/2)}) - (\alpha_{\derG}(s)(\CA^{(r/2)}))|_{B_j(r)}\| + f_j(\CA,r/2) \leq \\
\leq 2 \|\CA^{(r/2)}\| h(r/2,r) + f_j(\CA,r/2) \leq \\
\leq f_j(\CA,r/2) + 2 \|\CA\| h(r/2,r) + 2 f_j(\CA,r/2) h(r/2,r),
\end{multline}
where $\CA^{(r/2)}$ is a best possible approximation of $\CA$ on $B_j(r/2)$ and we used (\ref{eq:commutatorlowerbound}) to go from the second to the third line. This implies
\begin{equation}\label{eq:alphaGestimate}
\lVert \alpha_{\derG}(s)(\CA)\rVert'_{j,\beta}\leq C_{g,\beta} \lVert \CA\rVert'_{j,\beta}, \quad \beta \in \NN_0
\end{equation}
for some $C_{g,\beta}>0$. Therefore  $\alpha_{\derG}(s)(\CA) \in \SAal$ for any $\CA \in \SAal$.

Let $\CA\in\SAl$. Eq. (\ref{eq:alphaGndiffequation}) implies
\beq
\alpha_{\derG^{(n)}}(s+\Delta s)(\CA)-\alpha_{\derG^{(n)}}(s)(\CA)=\int_s^{s+\Delta s} \alpha_{\derG^{(n)}}(u)\left(\derG^{(n)}(u)(\CA)\right) du.
\eeq
Since according to Theorem 2.2 of \cite{nachtergaele2006propagation} for a $\alpha_{\derG^{(n)}}(s)(\CA)$ converges in norm to its $n\ra\infty$ limit uniformly in $s$ on any compact subset of $\RR$, we may exchange the limit $n\ra\infty$ and integration and get
\beq\label{eq:alphaGintegral}
\alpha_{\derG}(s+\Delta s)(\CA)-\alpha_{\derG}(s)(\CA)=\int_s^{s+\Delta s} \alpha_{\derG}(u)\left(\derG(u)(\CA)\right) du.
\eeq
To deduce this for general $\CA\in\SAal$, we choose a sequence of local observables $\CA^{(n)}$ converging to $\CA$ in the \Frechet\ topology and use the uniform convergence of $\derG(s)(\CA^{(n)})$ to $\derG(s)(\CA)$ on $[0,1]$  and Prop. \ref{prop:bracketContinuity} to show that (\ref{eq:alphaGintegral}) holds for $\CA\in\SAal$.

\end{proof}

\begin{lemma} \label{lma:GAjointcontinuity}
The map $\alpha(1):C([0,1],\mfkDal)\times \SAal \to \SAal$, $(\derG, \CA) \mapsto \alpha_{\derG}(1)(\CA)$ is continuous.
\end{lemma}
\begin{proof}
First, let us show continuity in $\derG$. Let $\{ \Delta \derG_n (s)\}, \, n \in \NN_0,$ be a sequence in $C([0,1],\mfkDal)$ converging to 0. 
Note that by Lemma \ref{lma:grepboundsfn} we can find $g \in \Orfm$, such that for any $n$ eq. (\ref{eq:Gbound}) holds for $\derG$ replaced with $\derG + \Delta \derG_n$. Therefore eq. (\ref{eq:LRbestimate}) implies
\beq
\|\alpha_{\derG+\Delta \derG_n}(u)(\CA)\|'_{j,\alpha} \leq B_{\alpha} \|\CA\|'_{j,\alpha} 
\eeq
for any $u \in [0,1]$ and any $n \in \NN_0$ and some constants $B_{\alpha}>0$ depending on $g$ only. Hence
\begin{multline}
\lVert\alpha_{\derG+\Delta\derG_n}(1)(\CA) - \alpha_{\derG}(1)(\CA)\rVert'_{j,\alpha} =
\\
=\lVert \int_{0}^1 du \frac{d}{du} \l \alpha_{\derG+\Delta\derG_n}(u)(\alpha_{\derG}(u)^{-1}(\alpha_{\derG}(1)(\CA))) \r \rVert'_{j,\alpha} \leq 
\\ \leq 
B_\alpha\int_{0}^1 du \lVert \left(\Delta\derG_n(u)\right)\left(\alpha_{\derG}(u)^{-1}(\alpha_{\derG(1)}(\CA))\right)\rVert'_{j,\alpha}\leq\\
\leq \tilde{B}_{\alpha} \| \CA \|'_{j,\alpha + d+3} \|\Delta \derG_n\|^{br}_{\alpha+d+3}.
\end{multline}
for some constant $\tilde{B}_{\alpha}$. Here we have used the same estimate as in (\ref{eq:bracketnormderivation}). Since the r.h.s. converges to zero as $n\ra\infty$, this proves continuity in $\derG$.

To show joint continuity, we similarly choose $g$ for a converging sequence $(\derG_n,\CA_n) \to (\derG,\CA)$. Then
\begin{multline}
\|\alpha_{\derG + \Delta \derG}(1)(\CA+\Delta \CA)-\alpha_\derG(1)(\CA)\|'_{j,\alpha} \leq \|\alpha_{\derG + \Delta \derG}(1)(\Delta \CA)\|'_{j,\alpha}+\\
+\|\alpha_{\derG + \Delta \derG}(1)(\CA)-\alpha_\derG(1)(\CA)\|'_{j,\alpha} \leq B_{\alpha} \|\Delta \CA\|'_{j,\alpha}
+ \\ + \|\alpha_{\derG + \Delta \derG}(1)(\CA)-\alpha_\derG(1)(\CA)\|'_{j,\alpha}.
\end{multline}
\end{proof}

\begin{corollary}\label{cor:GsAcontinuity}
The map $\alpha:C([0,1],\mfkDal) \times [0,1]\times \SAal  \to \SAal$, $(\derG, s,\CA) \mapsto \alpha_{\derG}(s)(\CA)$ is continuous.
\end{corollary}
\begin{proof}
Consider the ``rescaling map'' $\lambda:C([0,1],\mfkDal) \times [0,1]\ra C([0,1],\mfkDal)$, $(\derG,s)\mapsto \lambda(\derG,s)(u)=s\derG(s u)$. It is easy to check that this map is continuous. It is also straightforward to check that $\alpha(\derG,s,\CA)=\alpha(1)(\lambda(\derG,s),\CA)$. Therefore by Lemma \ref{lma:GAjointcontinuity} the map $\alpha$ is continuous.
\end{proof}

\begin{prop}\label{prop:GAsmoothness}
The map $\alpha(1):C([0,1],\mfkDal) \times \SAal \to \SAal$ defined by $(\derG, \CA) \mapsto \alpha_{\derG}(1)(\CA)$ is smooth.
\end{prop}
\begin{proof}
$\alpha_\derG(1)(\CA)$ is linear in $\CA$ and by Lemma \ref{lma:GAjointcontinuity} is jointly continuous in $\derG$ and $\CA$. Therefore it is sufficient to show that it is a smooth function of $\derG$. As in the proof of Lemma \ref{lma:GAjointcontinuity}, we write
\begin{multline}
\alpha_{\derG+t\Delta\derG}(1)(\CA) - \alpha_{\derG}(1)(\CA)
= \\ = \int_{0}^1 \frac{d}{du} \left[ \alpha_{\derG+t\Delta\derG}(u)\circ\alpha_{\derG}(u)^{-1}\circ\alpha_{\derG}(1)(\CA)\right]du=\\
=t\int_0^1  \alpha_{\derG+t\Delta\derG}(u)\left(\Delta\derG(u)\left(\alpha_{\derG}(u)^{-1}\circ\alpha_{\derG}(1)(\CA)\right)\right)du.
\end{multline}
Using Cor. \ref{cor:GsAcontinuity} we get  
\begin{multline}
\lim_{t\ra 0} \frac{\alpha_{\derG+t\Delta\derG}(1)(\CA) - \alpha_{\derG}(1)(\CA)}{t} = \\ = \int_0^1  \alpha_{\derG}(u)\left(\Delta\derG(u)\left(\alpha_{\derG}(u)^{-1}\circ\alpha_{\derG}(1)(\CA)\right)\right) du.
\end{multline}
This shows that the directional derivative of $\alpha(1)$ with respect to $\derG$  exists. Moreover, by Cor. \ref{cor:GsAcontinuity} and Prop. \ref{prop:bracketContinuity}  the derivative is continuous. Iterating the argument, we infer that $\alpha_\derG(1)(\CA)$ is a smooth function of $\derG$.
\end{proof}

\begin{remark}\label{rmk:alphaMderivative}
It follows from the above computation that if $\derG$ is a smooth $C([0,1],\mfkDal)$-valued function on a manifold $\CM$, then $\alpha_\derG(1)(\CA)$ is a smooth $\SAal$-valued function on $\CM$ whose $d_\CM$-derivative is given by
\beq
d_\CM\alpha_\derG(1)(\CA)=\int_0^1  \alpha_{\derG}(u)\left(d_\CM\derG(u)\left(\alpha_{\derG}(u)^{-1}\circ\alpha_{\derG}(1)(\CA)\right)\right) du.
\eeq
Somewhat schematically, we can also write
\beq\label{eq:alphader}
\alpha_\derG(1)^{-1}\circ d_\CM\alpha_\derG(1)=\int_0^1 \left(\alpha_\derG(1)^{-1}\circ \alpha_\derG(u)\right)(d_\CM \derG(u)) du.
\eeq
This formula is schematic because $\alpha_\derG(1)$ is a function on $\CM$ valued in automorphisms of $\SAal$, and we do not introduce any topology on the set of automorphisms. The proper interpretation of this formula is as follows. Note that the r.h.s. of eq. (\ref{eq:alphader}) is an element of $\Omega^1(\CM,\mfkDal)$. Let us denote it $\omega_\derG$. Then for any $\CB\in\SAal$ the 1-form  $\alpha_\derG(1)^{-1} \circ d_\CM\alpha_\derG(1)(\CB)\in\Omega^1(\CM,\SAal)$ is equal to $\omega_\derG(\CB)$. More generally, if $\CB$ is a smooth $\SAal$-valued function on $\CM$, then 
\beq
d_\CM\alpha_\derG(1)(\CB)=\alpha_\derG(1)\left(d_\CM\CB+\omega_\derG(\CB)\right).
\eeq
This implies that the covariant differential $d_\CM+\omega_\derG(\,\cdot\,)$ on the trivial bundle with fiber $\SAal$ is flat, i.e. $d_\CM\omega_\derG+\frac12 \{ \omega_\derG,\omega_\derG \} = 0$.
\end{remark}

\section{Ground states of gapped Hamiltonians} \label{app:gappedHamiltonians}

For any $\derH \in \mfkDal$ and a piecewise-continuous function $f:\RR \to \RR$ satisfying $f(t)=\CO(|t|^{-\infty})$ let $\mathscr{I}_{\derH,f}:\SAal\ra\SAal$ be the map
\beq
\mathscr{I}_{\derH,f}(\cdot) := \int_{-\infty}^{+\infty} f(t) \alpha_{ \derH}(t) (\cdot) d t.
\eeq
\begin{lemma} \label{lma:tildemap}
The map $\mathscr{I}_{\derH,f}$ is a well-defined continuous map. 
%If $\alpha_{\derH}(t)$ leaves a state $\psi$ invariant, then $\mathscr{I}_{\derH,f}$ preserves the subspace $\mfkdpal.$
\end{lemma}

\begin{proof}
Let us choose $h \in \Orfm$ such that
\beq
\sup_{j \in \Lambda} \sum_{\substack{Y \ni j \\ \diam(Y) \geq r}} \|\derH^Y\| \leq h(r).
\eeq
By Theorem 2.1 from \cite{matsuta2017improving}, for any $\CA \in \SA_{B_j(R)}$ and $\CB \in \SA_{B^c_j(R+r)}$ with $r>1$ and any $0<\sigma<1$ we have
\begin{multline}
\frac{\|[\alpha_\derH(t) (\CA),\CB]\|}{\|\CA\|\|\CB\|} \leq C_1 R^{d} e^{v t-r^{1-\sigma}} + C_2 t R^d (1+r)^d h(r^{\sigma}) + \\ + C_3 t e^{v t-r^{1-\sigma}} R^{2d} r^{\sigma+d} h(r^{\sigma})
\end{multline}
for some constants $C_1,C_2,C_3,v$ independent of $j,t,r,R$.

Let $t_0 = r^{1-\sigma}/2v$. Then
\begin{multline}
\frac{\|[\mathscr{I}_{\derH,f}(\CA),\CB]\|}{\|\CA\| \|\CB\|} \leq \\ \leq 2 \l \int_{0}^{t_0} |f(t)| \frac{\|[\alpha_{\derH}(t)(\CA),\CB]\|}{\|\CA\| \|\CB\|} dt \r  + 4 \int_{t_0}^{\infty} |f(t)| d t \leq \\ \leq  2 C' R^{2d} ( C_1 t_0 e^{v t_0 - r^{1-\sigma}} + C_2 t^2_0 (1+r)^d h(r^{\sigma}) + \\ + C_3 t^2_0 e^{v t_0 -r^{1-\sigma}} r^{\sigma+d} h(r^{\sigma}) ) + 4 \int_{t_0}^{\infty} |f(t)| d t =: g(R,R+r)
\end{multline}
for some constant $C'$. Since $g(r/2,r) \in \Orf$, in the same way as in eq. (\ref{eq:LRbestimate}), we get an estimate for any $\CA \in \SAl$
\beq
f_j(\mathscr{I}_{\derH,f}(\CA),r) \leq f_j(\CA,r/2) + 2 \|\CA\| g(r/2,r) + 2 f_j(\CA,r/2) g(r/2,r),
\eeq
which implies
\begin{equation}
\lVert \mathscr{I}_{\derH,f}(\CA) \rVert'_{j,\alpha} \leq C_{g,\alpha} \lVert \CA\rVert'_{j,\alpha}, \quad \alpha \in \NN_0
\end{equation}
for some $C_{g,\alpha}>0$.  Therefore $\mathscr{I}_{\derH,f}:\SAal \to \SAal$ is well-defined and continuous.

%Finally, if $\alpha_{\derH}(t)$ leaves a state $\psi$ invariant, then for any $\CA,\CB\in\mfkdal$ we have 
%$$
%\langle [\mathscr{I}_{\derH,f}(\CA),\CB]\rangle_\psi=\int_{-%\infty}^{+\infty} f(t)\langle[\CA, %\alpha_{\derH}(-t)(\CB)]\rangle_\psi dt.
%$$
%Therefore if $\CA\in\mfkdpal$, then %$\mathscr{I}_{\derH,f}(\CA)\in\mfkdpal.$

\end{proof}

\begin{remark}
The version of the Lieb-Robinson bounds proved in \cite{bachmann2012automorphic} is sufficient to prove the existence of the map $\mathscr{I}_{\derH,f}$ for UL Hamiltonians or Hamiltonians with exponential decay, but not for arbitrary UAL Hamiltonians. It was pointed to us by Bruno Nachtergaele that the case of UAL Hamiltonians can be dealt with using the improved Lieb-Robinson bounds from  \cite{matsuta2017improving,elsemachadonayakyao}. Further implications of these improved bounds for gapped Hamiltonians are studied in \cite{nachtergaelesimsyoungupcoming}.
\end{remark}

Using the result of \cite{moon2020automorphic} one can show that a smooth family of gapped UL Hamiltonians under certain additional assumptions defines a smooth family of gapped states in the sense of Definition \ref{def:smoothfamily} (though we expect that a similar result should hold for a smooth family of gapped UAL Hamiltonians):

\begin{prop}
Let $\CM$ be a compact manifold, and let $\derH$ be a $\mfkDl$-valued function which is smooth when regarded as $\mfkDal$-valued function. Suppose for any $\m\in\CM$ the derivation $\derH_{\m}$ is gapped with a unique ground state $\psi_\m$. Suppose also that for any $\CA \in \SAal$ the average $\lal \CA \ral_{\psi_m}$ is a smooth function on $\CM$. Then $\psi$ is a smooth family of gapped states.
\end{prop}

\begin{proof}

Since $\CM$ is compact, there exists $\Delta>0$ which bounds from below the gap of $\derH_{\m}$ for any $\m\in\CM$.

Let us define $\derG \in \Omega^1(\CM,\mfkDal)$ by $\derG = - \mathscr{I}_{\derH,W_{\Delta}}(d \derH)$, where $W_{\Delta}(t) = \CO(|t|^{-\infty})$ is an odd function such that $\int W_{\Delta}(t) e^{-i \omega t} d t = \frac{i}{\omega}$ for $|\omega|>\Delta'$ for some $0<\Delta'<\Delta$.

For any smooth path $p:[0,1] \to \CM$ the family $p^* \derH$ satisfies the conditions of the Theorem 1.3 from \cite{moon2020automorphic} that guarantees $p^* \psi (s) = p^* \psi(0) \circ \alpha_{p^* \derG}(s)$.

\end{proof}

\printbibliography

@article{ThoulessHall,
  title={Hall conductance and the statistics of flux insertions in gapped interacting lattice systems},
  author={Kapustin, Anton and Sopenko, Nikita},
  journal={Journal of Mathematical Physics},
  volume={61},
  number={10},
  pages={101901},
  year={2020},
  publisher={AIP Publishing LLC}
}

@article{kitaev2006anyons,
  title={Anyons in an exactly solved model and beyond},
  author={Kitaev, Alexei},
  journal={Annals of Physics},
  volume={321},
  number={1},
  pages={2--111},
  year={2006},
  publisher={Elsevier}
}

@article{HastingsKoma,
  title={Spectral gap and exponential decay of correlations},
  author={Hastings, Matthew B and Koma, Tohru},
  journal={Communications in mathematical physics},
  volume={265},
  number={3},
  pages={781--804},
  year={2006},
  publisher={Springer}
}

@article{bachmann2019many,
  title={A many-body index for quantum charge transport},
  author={Bachmann, Sven and Bols, Alex and De Roeck, Wojciech and Fraas, Martin},
  journal={Communications in Mathematical Physics},
  volume={375},
  pages={1249--1272},
  year={2020}
}

@article{bachmann2012automorphic,
  title={Automorphic equivalence within gapped phases of quantum lattice systems},
  author={Bachmann, Sven and Michalakis, Spyridon and Nachtergaele, Bruno and Sims, Robert},
  journal={Communications in Mathematical Physics},
  volume={309},
  number={3},
  pages={835--871},
  year={2012},
  publisher={Springer}
}

@book {bratteli2012operator2,
    AUTHOR = {Bratteli, Ola and Robinson, Derek W.},
     TITLE = {Operator algebras and quantum statistical mechanics. 2. Equilibrium states. Models in quantum statistical mechanics},
    SERIES = {Texts and Monographs in Physics},
   EDITION = {2nd ed.},
 PUBLISHER = {Springer-Verlag, Berlin},
      YEAR = {1997},
}

@article{hastings2010quasi,
  title={Quasi-adiabatic continuation for disordered systems: Applications to correlations, Lieb-Schultz-Mattis, and Hall conductance},
  author={Hastings, Matthew B},
  journal={arXiv:1001.5280},
  year={2010}
}

@article{nachtergaele2006propagation,
  title={Propagation of correlations in quantum lattice systems},
  author={Nachtergaele, Bruno and Ogata, Yoshiko and Sims, Robert},
  journal={Journal of statistical physics},
  volume={124},
  number={1},
  pages={1--13},
  year={2006},
  publisher={Springer}
}

@article{kapustin2020higherA,
  title={Higher-dimensional generalizations of the Thouless charge pump},
  author={Kapustin, Anton and Spodyneiko, Lev},
  journal={arXiv preprint arXiv:2003.09519},
  year={2020}
}

@article{kapustin2020higherB,
  title={Higher-dimensional generalizations of the Berry curvature},
  author={Kapustin, Anton and Spodyneiko, Lev},
  journal={arXiv preprint arXiv:2001.03454},
  year={2020}
}

@ARTICLE{hastingsmichalakis,
       author = {{Hastings}, Matthew B. and {Michalakis}, Spyridon},
        title = "{Quantization of Hall Conductance for Interacting Electrons on a Torus}",
      journal = {Communications in Mathematical Physics},
         year = 2015,
        month = feb,
       volume = {334},
       number = {1},
        pages = {433-471},
}

@article{Thouless,
  title = {Quantization of particle transport},
  author = {Thouless, D. J.},
  journal = {Phys. Rev. B},
  volume = {27},
  issue = {10},
  pages = {6083--6087},
  year = {1983},
  month = May,
  publisher = {American Physical Society},
}

@book {QImeetsQM,
    AUTHOR = {Zeng, Bei and Chen, Xie and Zhou, Duan-Lu and Wen, Xiao-Gang},
     TITLE = {Quantum information meets quantum matter: From quantum entanglement to topological phases of many-body
              systems.},
    SERIES = {Quantum Science and Technology},
 PUBLISHER = {Springer, New York},
      YEAR = {2019},
}

@article {ogataTindex,
    AUTHOR = {Ogata, Yoshiko},
     TITLE = {A {${\mathbb {Z}}_2$}-index of symmetry protected topological
              phases with time reversal symmetry for quantum spin chains},
   JOURNAL = {Comm. Math. Phys.},
  FJOURNAL = {Communications in Mathematical Physics},
    VOLUME = {374},
      YEAR = {2020},
    NUMBER = {2},
     PAGES = {705--734},
}

@article {ogata2019classification,
    AUTHOR = {Ogata, Yoshiko},
     TITLE = {A classification of pure states on quantum spin chains
              satisfying the split property with on-site finite group
              symmetries},
   JOURNAL = {Trans. Amer. Math. Soc. Ser. B},
  FJOURNAL = {Transactions of the American Mathematical Society. Series B},
    VOLUME = {8},
      YEAR = {2021},
     PAGES = {39--65},
}

@article {ogatabourne,
    AUTHOR = {Bourne, Chris and Ogata, Yoshiko},
     TITLE = {The classification of symmetry protected topological phases of
              one-dimensional fermion systems},
   JOURNAL = {Forum Math. Sigma},
  FJOURNAL = {Forum of Mathematics. Sigma},
    VOLUME = {9},
      YEAR = {2021},
     PAGES = {Paper No. e25, 45},
}

@article{moon2020automorphic,
  title={Automorphic equivalence within gapped phases in the bulk},
  author={Moon, Alvin and Ogata, Yoshiko},
  journal={Journal of Functional Analysis},
  volume={278},
  number={8},
  pages={108422},
  year={2020},
  publisher={Elsevier}
}

@article{osborne2007simulating,
  title={Simulating adiabatic evolution of gapped spin systems},
  author={Osborne, Tobias J},
  journal={Physical review A},
  volume={75},
  number={3},
  pages={032321},
  year={2007},
  publisher={APS}
}

@article{kapustin2020hall,
  title={Hall conductance and the statistics of flux insertions in gapped interacting lattice systems},
  author={Kapustin, Anton and Sopenko, Nikita},
  journal={Journal of Mathematical Physics},
  volume={61},
  number={10},
  pages={101901},
  year={2020},
  publisher={AIP Publishing LLC}
}

@article{ogata2021h,
  title={A $ H^{3}(G,{\mathbb T})$-valued index of symmetry protected topological phases with on-site finite group symmetry for two-dimensional quantum spin systems},
  author={Ogata, Yoshiko},
  journal={arXiv preprint arXiv:2101.00426},
  year={2021}
}

@article{derivedbrackets,
  title={The sh Lie structure of Poisson brackets in field theory},
  author={Barnich, Glenn and Fulp, Ronald and Lada, Tomasz and Stasheff, Jim},
  journal={Communications in mathematical physics},
  volume={191},
  number={3},
  pages={585--601},
  year={1998},
  publisher={Springer}
}

@ARTICLE{kapustin2020classification,
       author = {{Kapustin}, Anton and {Sopenko}, Nikita and {Yang}, Bowen},
        title = "{A classification of invertible phases of bosonic quantum lattice systems in one dimension}",
      journal = {Journal of Mathematical Physics},
         year = 2021,
        month = aug,
       volume = {62},
       number = {8},
        pages = {081901},
archivePrefix = {arXiv},
       eprint = {2012.15491},
 primaryClass = {quant-ph},
}

@ARTICLE{sopenko2021,
       author = {{Sopenko}, Nikita},
        title = "{An index for two-dimensional SPT states}",
      journal = {Journal of Mathematical Physics},
         year = 2021,
        month = nov,
       volume = {62},
       number = {11},
        pages = {111901},
}

@article{Hamilton,
author = {Richard S. Hamilton},
title = {{The inverse function theorem of Nash and Moser}},
volume = {7},
journal = {Bulletin (New Series) of the American Mathematical Society},
number = {1},
publisher = {American Mathematical Society},
pages = {65 -- 222},
year = {1982},
}

@ARTICLE{gauge,
       author = {{Baroni}, Stefano and {Bertossa}, Riccardo and {Ercole}, Loris and {Grasselli}, Federico and {Marcolongo}, Aris},
        title = "{Heat transport in insulators from ab initio Green-Kubo theory}",
        year=2018,
    archivePrefix={arXiv},
    eprint = {1802.08006},
    primaryClass={cond-mat.stat-mech}
    }

@ARTICLE{KapSpo,
       author = {{Kapustin}, Anton and {Spodyneiko}, Lev},
        title = "{Microscopic formulas for thermoelectric transport coefficients in lattice systems}",
      journal = {Physical Review B},
         year = 2021,
        month = jul,
       volume = {104},
       number = {3},
        pages = {035150},
}

@book {Jacobson,
    AUTHOR = {Jacobson, Nathan},
     TITLE = {Basic algebra. {I}},
 PUBLISHER = {W. H. Freeman and Co., San Francisco, Calif.},
      YEAR = {1974},
}

@article{CooperHalperinRuzin,
  title = "{Thermoelectric response of an interacting two-dimensional electron gas in a quantizing magnetic field}",
  author = {Cooper, N. R. and Halperin, B. I. and Ruzin, I. M.},
  journal = {Phys. Rev. B},
  volume = {55},
  issue = {4},
  pages = {2344--2359},
  numpages = {0},
  year = {1997},
}

@article {Kitaev_talk,
    AUTHOR = {Kitaev, Alexei},
     TITLE = {Differential forms on the space of statistical mechanical lattice models},
     Note = {Talk at "Between Topology and Quantum Field Theory: a conference in celebration of Dan Freed's 60th birthday"}, 
     URL = {https://web.ma.utexas.edu/topqft/talkslides/kitaev.pdf},
     year = {2019}
}

@ARTICLE{diabolical,
       author = {{Hsin}, Po-Shen and {Kapustin}, Anton and {Thorngren}, Ryan},
        title = "{Berry phase in quantum field theory: Diabolical points and boundary phenomena}",
      journal = {Physical Review B},
         year = 2020,
       volume = {102},
       number = {24},
          eid = {245113},
        pages = {245113},
}

@ARTICLE{flow,
       author = {{Wen}, Xueda and {Qi}, Marvin and {Beaudry}, Agn{\`e}s and {Moreno}, Juan and {Pflaum}, Markus J. and {Spiegel}, Daniel and {Vishwanath}, Ashvin and {Hermele}, Michael},
        title = "{Flow of (higher) Berry curvature and bulk-boundary correspondence in parametrized quantum systems}",
      journal = {arXiv e-prints},
         year = 2021,
        month = dec,
          eid = {arXiv:2112.07748},
        pages = {arXiv:2112.07748},
archivePrefix = {arXiv},
       eprint = {2112.07748},
 primaryClass = {cond-mat.str-el},
}

@ARTICLE{nachtergaelesimsyoung,
       author = {{Nachtergaele}, Bruno and {Sims}, Robert and {Young}, Amanda},
        title = "{Quasi-locality bounds for quantum lattice systems. I. Lieb-Robinson bounds, quasi-local maps, and spectral flow automorphisms}",
      journal = {Journal of Mathematical Physics},
     keywords = {Mathematical Physics, Quantum Physics, 82B10, 82C10, 82C20},
         year = 2019,
        %month = jun,
        volume = {60},
        number = {6},
          %eid = {061101},
        pages = {061101},
         % doi = {10.1063/1.5095769}
}

@inproceedings{matsuta2017improving,
  title={Improving the Lieb--Robinson bound for long-range interactions},
  author={Matsuta, Takuro and Koma, Tohru and Nakamura, Shu},
  booktitle={Annales Henri Poincar{\'e}},
  volume={18},
  number={2},
  pages={519--528},
  year={2017},
  organization={Springer}
}

@article{nachtergaelesimsyoungupcoming,
    author = {{Nachtergaele}, Bruno and {Sims}, Robert and {Young}, Amanda},
    journal = {in preparation}
    }

@ARTICLE{elsemachadonayakyao,
       author = {{Else}, Dominic V. and {Machado}, Francisco and {Nayak}, Chetan and {Yao}, Norman Y.},
        title = "{An improved Lieb-Robinson bound for many-body Hamiltonians with power-law interactions}",
      journal = {arXiv e-prints},
        year = 2018,
        %month = sep,
          %eid = {arXiv:1809.06369},
archivePrefix = {arXiv},
       eprint = {1809.06369},
 primaryClass = {quant-ph},
       
}

@inproceedings{kitaev2009periodic,
  title={Periodic table for topological insulators and superconductors},
  author={Kitaev, Alexei},
  booktitle={AIP conference proceedings},
  volume={1134},
  number={1},
  pages={22--30},
  year={2009},
  organization={American Institute of Physics}
}

@INPROCEEDINGS{ryuetal,
       author = {{Schnyder}, Andreas P. and {Ryu}, Shinsei and {Furusaki}, Akira and {Ludwig}, Andreas W.~W.},
        title = "{Classification of Topological Insulators and Superconductors}",
        booktitle={AIP conference proceedings},
  volume={1134},
  number={1},
  pages={10--21},
  year={2009},
  organization={American Institute of Physics}
}

@ARTICLE{ryuetal2,
       author = {{Ryu}, Shinsei and {Schnyder}, Andreas P. and {Furusaki}, Akira and {Ludwig}, Andreas W.~W.},
        title = "{Topological insulators and superconductors: tenfold way and dimensional hierarchy}",
      journal = {New Journal of Physics},
         year = 2010,
        month = jun,
       volume = {12},
       number = {6},
          eid = {065010},
        pages = {065010},
}

\end{document}